\documentclass[12pt]{article}
\usepackage{amsthm,amsmath,amssymb,graphicx}

\newtheorem{theorem}{Theorem}
\newtheorem{assump}{Assumption}

\newtheorem{definition}[theorem]{Definition}

\newtheorem{lemma}[theorem]{Lemma}

\newtheorem{prop}[theorem]{Proposition}
\newtheorem{remark}[theorem]{Remark}

\newcommand\EE {\mathbb E}

\newcommand\RR {\mathbb R}

\newcommand\QQ {\mathbb Q}

\newcommand{\be}{\begin{equation}}
\newcommand{\ee}{\end{equation}}

\newcommand{\ed}{\end{document}}

\newcommand{\sop}[2]{\frac{\partial^2 #1}{\partial #2^2}}

\def\bone{\mathbf{1}}

\begin{document}


\baselineskip=24pt \parskip=12pt

\begin{titlepage}

\vfill
\begin{center}
{\bf \Huge Local Variance Gamma and Explicit Calibration to Option Prices}\\
\vspace{1in}
\begin{tabular}{lcl}
Peter Carr   &   & Sergey Nadtochiy  \\
Courant Institute of Mathematical Sciences & & Department of Mathematics at University of Michigan\\
New York University & & University of Michigan\\
New York, NY 10012 & & Ann Arbor, MI  48109 \\
pcarr@nyc.rr.com & & sergeyn@umich.edu \\
\end{tabular}\\

\vfill

\begin{center}
This version: \today\footnote{We are very grateful for comments from Laurent Cousot, Bruno Dupire, David Eliezer, Travis Fisher, Bjorn Flesaker, Alexey Polishchuk, Serge Tchikanda, Arun Verma, Jan Ob{\l}{\'o}j, and Liuren Wu. We also thank the anonymous referee for valuable remarks and suggestions which helped us improve the paper significantly. 
We are responsible for any remaining errors.}\\
File reference: lvg8.tex\\
First draft: April 11, 2012
\end{center}
\vfill

{\bf Abstract}
\end{center}
In some options markets (e.g. commodities), options are listed with only a single maturity for each underlying. In others, (e.g. equities, currencies), options are listed with multiple maturities. In this paper, we analyze a special class of pure jump Markov martingale models and provide an algorithm for calibrating such model to match the market prices of European options of multiple strikes and maturities. This algorithm matches option prices exactly and only requires solving several one-dimensional root-search problems and applying elementary functions. We show how to construct a time-homogeneous process which meets a single smile, and a piecewise time-homogeneous process which can meet multiple smiles.

\end{titlepage}

\newpage
\section{Introduction}

\begin{quote}
Why is there always so much month left at the end of the money?
  -- Sarah Lloyd
\end{quote}

The very earliest literature on option pricing imposed a process on the underlying asset price
and derived unique option prices as a consequence of the dynamical assumptions and no arbitrage.
We may characterize this literature as going ``From Process to Prices''.
However, once the notion of implied volatility was introduced,
the inverse problem of going ``From Prices to Process'' was established.
The term that practitioners favor for this inverse problem
is calibration - the practice of determining the required inputs to a model
so that they are consistent with a specified set of market prices.
Implied volatility is just the simplest example of this calibration procedure, wherein
a single option price is given and the volatility input to the Black Scholes model is
determined so as to gain exact consistency with this one market price.

When the number of calibration instruments is expanded to
several options of different maturities, the Black Scholes model
can be readily adapted to be consistent with this information set.
One simply assumes that the instantaneous variance is a piecewise constant
function of time, which jumps at each option maturity.
The staircase levels are
chosen so that the time-averaged cumulative variance matches
the implied variance at each maturity.
So long as the given option prices are arbitrage-free,
this deterministic volatility version of the Black Scholes model
is capable of achieving exact consistency with any given
term structure of market option prices.
As a bonus, the ability to have closed form solutions for European option prices is 
retained.

Unfortunately, when the number of calibration instruments is instead expanded to
several co-terminal options of different strikes,
there is no unique simple extension of the Black Scholes model
which is capable of meeting this information set.
Whenever the implied volatility at a single term is a non-constant function of the strike price,
there are, instead, many ways to gain exact consistency with the associated option prices.
The earliest work on this problem seems to be Rubinstein\cite{r94} in his presidential address.
Working in a discrete time setting, he assumed that the price of
the underlying asset is a
Markov chain
evolving on a binomial lattice.
Assuming that the terminal nodes of the lattice fell on option strikes, he was able to
find the nodes and the transition probabilities that determine this lattice.
A continuous time and state version of the Rubinstein result
can be found in Carr and Madan\cite{cm98}.
Other methods of calibrating to a single smile are presented in Cox, Hobson and Ob{\l}{\'o}j\cite{CoxHobsonObloj}, \cite{EkstromHobsonTysk}, \cite{Noble}, 
and to multiple smiles, in Madan and Yor\cite{my02}.
However, all these works assume the availability of continuous smiles, i.e. option prices for a continuum of strikes. 

In this paper, we present a new way to go from a given set of option prices 
to a Markovian martingale in a continuous time setting. 
This calibration method can be successfully applied to continuous smiles, as well as a finite family of option prices
with multiple strikes and maturities.
In order to implement our algorithm, one only needs to solve several one-dimensional root-search problems and 
apply the elementary functions.
To the authors' knowledge, this is the first example of explicit exact calibration
to a finite set of option prices with multiple strikes and maturities, such that the calibrated (continuous time) process has continuous distribution at all times.
In addition, if the given market options are co-terminal, the calibrated process becomes time-homogeneous. 

Suppose that the risk-neutral process for the (forward) price of an asset, underlying
a set of European options, is a driftless time-homogeneous diffusion
running on an independent and unbiased Gamma process.
We christen this model ``Local Variance Gamma'' (LVG),
since it combines ideas from both the Local Variance model of Dupire\cite{d94}
and the Variance Gamma model of Madan and Seneta\cite{ms90}.
Since the diffusion is time-homogeneous and the subordinating Gamma process is L\'evy,
their independence implies that the spot price process is also Markov and time-homogeneous.
Since the subordinator is a pure jump process, the LVG process governing the
underlying spot price is also pure jump.
In addition, the forward and backward equations governing options prices in the LVG
model turn into much simpler partial differential difference equations (PDDE's).
The existence of these PDDE's permits both explicit calibration of the LVG model
and fast numerical valuation of contingent claims.
As a result of the forward PDDE holding, the diffusion coefficient can be explicitly represented (calibrated) in terms of a single (continuous) smile.
The backward PDDE, then, allows for efficient valuation of other contingent claims,
by successively solving a finite sequence of second order linear ordinary differential equations (ODE's) in the spatial variable.

While the single smile results are relevant for commodity option markets,
they are not as relevant for market makers in equity and currency
options markets where multiple option maturities trade simultaneously.
In order to be consistent with multiple smiles, we also present an extension of the LVG model that results in a piecewise time-homogeneous process for the underlying asset price. The calibration procedure remains explicit in the case of multiple smiles: in particular, it does not require application of any optimization methods.

The above results allow us to calibrate LVG model, or its extension, to continuous arbitrage-free smiles, implying, in particular, that option prices for a continuum of strikes must be observed in the market. To get rid of this unrealistic assumption, we show how to use the PDDE's associated with the LVG model to construct continuous arbitrage-free smiles from a finite family of option prices, for multiple strikes and maturities. 
To the authors' knowledge, this is the first construction of a continuously differentiable arbitrage-free interpolation of implied volatility across strikes, that only requires solutions to one-dimensional root-search problems and application of elementary functions.


This paper is structured as follows.
In the next section, we present the basic assumptions and
construct the LVG process, i.e. a driftless time-homogeneous diffusion
subordinated to an unbiased gamma process.
In the following section, we derive the forward and backward PDDE's
that govern option prices.
In the penultimate section, we discuss calibration strategies.
To meet multiple smiles, we construct a piecewise time-homogeneous
extension of the LVG process and develop the corresponding forward and backward PDDE's.
The algorithm for constructing continuous smiles from a finite set of option prices, along with the corresponding theorem and numerical results, is presented in Subsection $\ref{subse:ContFromFinite}$.
The final section summarizes the paper and makes some suggestions for future research.

\section{Local Variance Gamma Process}

\subsection{Model Assumptions}
\label{subse:ModelAss}

In this subsection,
we lay out the general financial and mathematical assumptions
used throughout the paper.
For simplicity, we assume zero carrying costs for all assets.
As a result, we have zero interest rates, dividend yields, etc.
It is straightforward to extend our results to the case when these
quantities are deterministic functions of time (the associated numerical issues are discussed in Remark $\ref{rem:shortMat}$).
We also assume frictionless markets and no arbitrage.
Motivated by the fundamental theorem of asset pricing, we assume that there exists a probability measure $\mathbb{Q}$ such that market prices of all assets are $\mathbb{Q}$-martingales.
Following standard terminology, we will refer to $\mathbb{Q}$ as a risk-neutral probability measure.

We assume that the market includes a family of European call options written on a common
underlying asset whose prices process we denote by $S$.
We assume that the initial spot price $S_0$ is known.
Throughout this paper, we denote by $(L,U)$, with $-\infty \leq L < U \leq \infty$, the spatial interval on which
the $S$ process lives. The boundary points $L$ and $U$ may or may not be attainable. If a boundary point is attainable, we assume that the process is absorbed at that point. If a boundary point is infinite, we assume it is not attainable. Since we interpret $S$ as the price process, for simplicity, one can think of $L$ and $U$ as $0$ and $\infty$, respectively.

In the following subsections, we consider several specific classes of the underlying processes and use them as building blocks to construct the Local Variance Gamma process.
Namely, the desired pure-jump process arises by
subordinating a driftless diffusion to an unbiased gamma process.

\subsection{Driftless Diffusion}
\label{subse:DD}
To elaborate on this additional structure,
let $W$ be a standard Brownian motion.
We define $D$ as a driftless time-homogeneous diffusion with the generator $\frac{1}{2}a^2(D) \partial^2_{DD}$,
\be
dD_s = a(D_s) dW_s, \qquad s \in [0,\zeta],
\label{D}
\ee
and with the initial value $x\in(L,U)$. The stopping time $\zeta$ denotes the first time the diffusion exits the interval $(L,U)$. The process is stopped (absorbed) at $\zeta$. We assume that the diffusion coefficient $a:(L,U)\rightarrow(0,\infty)$ is a piecewise continuous function, with a finite number of discontinuities of the first order (i.e. the left and right limits exist at each point), bounded uniformly from above and away from zero, and having finite limits at those boundary points that are finite. Under these assumptions, for any initial condition $x\in(L,U)$, the SDE (\ref{D}) has a weak solution which is unique in the sense of probability law (cf. Theorem 5.7, on p. 335, in \cite{KaratzasShreve}). 
It is easy to see that, under these assumptions, a boundary point, $L$ or $U$, is accessible if and only if it is finite.
Note also that one can extend the set of initial conditions to the entire real line by assuming that the solution to (\ref{D}) remains at $x$, for any $x\in\RR\setminus(L,U)$. The collection of distributions of the weak solutions to (\ref{D}), for all $x\in\RR$, forms a Markov family, in the sense of Definition 2.5.11, on p. 74, in \cite{KaratzasShreve}.
More precisely, on the canonical space of continuous paths $\Omega^D = C\left([0,\infty)\right)$, equipped with the Borel sigma algebra $\mathcal{B}\left(C\left([0,\infty)\right)\right)$, we consider a family of probability measures $\left\{ \QQ^{D,x} \right\}$, for $x\in\RR$, such that every $\QQ^{D,x}$ is the distribution of a weak solution to (\ref{D}), with initial condition $x$. 
As follows, for example, from Theorem 5.4.20 and Remark 5.4.21, on p. 322, in \cite{KaratzasShreve}, $\left\{ \QQ^{D,x} \right\}$ and the canonical process $D:\omega\mapsto (\omega(t),\,\,t\geq0)$ form a Markov family.
Due to the growth restrictions on $a$, $D$ is a true martingale, with respect to its natural filtration, under any measure $\QQ^{D,x}$. 

It is worth mentioning that there is a reason why we construct $D$ in this particular way, introducing the family $\left\{ \QQ^{D,x} \right\}$. Namely, in order to carry out the constructions in Subsections \ref{fpidec} and, in turn, \ref{subse:calib.multsmiles}, we need to consider the diffusion process $D$ as a function of its initial condition $x$. 
In particular, we use certain properties of its distribution $\QQ^{D,x}$, such as the measurability of the mapping $x\mapsto\QQ^{D,x}$, which follows from the definition of a Markov family (cf. Definition 2.5.11, on p. 74, in \cite{KaratzasShreve}).


We can compute prices of \emph{European options} in a model where the risk-neutral dynamics of the underlying are given by the above driftless diffusion. Namely, given a Borel measurable and exponentially bounded payoff function $\phi$, we recall the time $t$ price of the associated European type claim, with the \emph{time of maturity} $T$:
$$
V^{D,\phi}_t(T) = \EE^x \left(\left.\phi(D_{T}) \right| \mathcal{F}^{D}_t\right)
= \left.\EE^y \left( \phi(D_{T-t})\right)\right|_{y=D_t},
$$
which holds for all $x\in\RR$ and $t\in[0,T]$.
In the above, we denote by $\EE^x$ the expectation with respect to $\QQ^{D,x}$ and by $\mathcal{F}^{D}$ the filtration generated by $D$. The last equality is due to Markov property. 
Thus, in a driftless diffusion model, the price of a European type option, at any time, can be computed via the price function:
\begin{equation}\label{eq.DD.Vphi}
V^{D,\phi}(\tau,x) = \EE^x \left( \phi(D_{\tau})\right).
\end{equation}
Throughout the paper, $\tau$ is used as an auxiliary variables, which, often, has a meaning of the \emph{time to maturity}: $\tau=T-t$. In addition, we differentiate prices, as random variables, from the price functions by adding a subscript (typically, ``$t$").
The option price function in a driftless diffusion model is expected to satisfy the Black-Scholes equation in $(\tau,x)$:
\begin{equation}\label{eq.DD.BS}
\partial_{\tau} V^{D,\phi}(\tau,x) = \frac{1}{2} a^2(x) \partial^2_{xx} V^{D,\phi}(\tau,x),
\,\,\,\,\,\,\,\,\,\, V^{D,\phi}(0,x) = \phi(x). 
\end{equation}
Notice, however, that the coefficient in the above equation can be discontinuous, therefore, we can only expect the value function to satisfy this equation in a weak sense.

\begin{theorem}\label{th:DD.th1}
Assume that $\phi:(L,U)\rightarrow\RR$ is exponentially bounded and continuously differentiable, and that $\phi'$ is absolutely continuous, with a square integrable derivative. Then, $V^{D,\phi}$ (defined in (\ref{eq.DD.Vphi})) is the unique weak solution to (\ref{eq.DD.BS}), in the sense that: $V^{D,\phi}$ is continuous, its weak derivatives $\partial_{\tau}V^{D,\phi}$ and $\partial^2_{xx}V^{D,\phi}$ are square integrable in $(0,T)\times(L,U)$, for any fixed $T>0$, and $V^{D,\phi}$ satisfies (\ref{eq.DD.BS}).
\end{theorem} 

The proof of Theorem $\ref{th:DD.th1}$ is given in Appendix A.
For a special class of contingent claims, we can derive an additional parabolic PDE, satisfied by the option prices. Recall that, in a driftless diffusion model, the price of a call option with strike $K$ and maturity $T$ is given by:
$$
C^{D}_t(T,K) = \EE^x \left(\left.(D_{T} - K)^+\right| \mathcal{F}^{D}_t\right) = C^D(T-t,K,D_t),
$$
for all $x\in\RR$ and $t\in[0,T]$.
In the above, we denote by $C^D$ the call price function, which is defined as
\begin{equation}\label{eq.DD.CD}
C^D(\tau,K,x) = \EE^x \left((D_{\tau} - K)^+\right).
\end{equation}
The call price function satisfies another parabolic PDE, in $(\tau,K)$, known as the Dupire's equation:
\begin{equation}\label{dup}
\partial_{\tau} C^D(\tau,K,x) = \frac{1}{2} a^2(K) \partial^2_{KK} C^D(\tau,K,x).
\end{equation}
Due to the possible discontinuity of $a$, we can only establish this equation in a weak sense.

\begin{theorem}\label{th:DD.th2}
For any $x\in(L,U)$, the call price function $C^D(\tau,K,x)$ (defined in (\ref{eq.DD.CD})) is absolutely continuous as a function of $K\in[L,U]$. Its partial derivative $\partial_K C^D(\tau,\,.\,,x)$ has a unique nondecreasing and right continuous modification, which defines a probability measure on $[L,U]$ (choosing $\partial_K C^D(\tau,U,x)=0$). Moreover, for any bounded Borel function $\phi$, with a compact support in $(L,U)$, we have
\begin{equation}\label{eq.DD.dup1}
\int_L^U C^D(\tau,K,x)\phi(K) dK = \int_L^U (x-K)^+ \phi(K) dK + \frac{1}{2} \int_0^{\tau} \int_L^U a^2(K) \phi(K)\, d \left( \partial_K C^D(u,K,S)\right) du,
\end{equation}
for all $\tau>0$ and all $x\in(L,U)$.
\end{theorem}

The proof of Theorem $\ref{th:DD.th2}$ is given in Appendix A.

\subsection{Gamma Process}
Let $\{\Gamma_t(t^*, \alpha), t > 0\}$ be an independent gamma process with parameters $t^*>0$ and $\alpha>0$.
As is well known, a gamma process is an increasing L\'evy process whose L\'evy density is given by:
\be
k_{\Gamma}(t) =  \frac{e^{- \alpha t}}{t^* t}, \qquad t > 0,
\label{ld}
\ee
with parameters $t^*>0$ and $\alpha>0$.
Intuitively, jumps whose size lies in the interval $[t,t+dt]$ occur as a
Poisson process with intensity $k_{\Gamma}(t)dt$.
The fraction $1/t^*$  controls the rate of jump arrivals,
while  $1/\alpha$ controls the mean jump size, given that a jump has occurred.
One can also get direct intuition on $t^*$ and $\alpha$, rather than their reciprocals.
Since a gamma process has infinite activity, the number of jumps over any
finite time interval is infinite for small jumps and finite for large jumps.
If we ignore the small jumps, then the larger is $t^*$,
the longer one must wait on average for a fixed number of large jumps to occur.
Furthermore, the larger is $\alpha$, the longer it takes the running sum of these large jumps to reach a fixed positive level.
For an introduction to gamma processes, and L\'evy processes more generally,
see Bertoin\cite{b99}, Sato\cite{s99}, or Applebaum\cite{a04}.
For their application in a financial context, see
Schoutens\cite{s03} or Cont and Tankov\cite{ct04}.

The marginal distribution of a gamma process at time
$t \geq 0$ is a gamma distribution:
\be
\mathbb{Q} \{\Gamma_t \in ds \} = \frac{\alpha^{t/t^*}}{\Gamma(t/t^*)} s^{t/t^* - 1}e^{-\alpha s} ds, \qquad s>0,\, t > 0,
\label{gpdf}
\ee
for parameters $\alpha>0$ and $t^*>0$.
For $t < t^*$,
this PDF has a singularity at $s=0$, while for
$t > t^*$, the PDF vanishes at $s=0$.
At $t=t^*$, the PDF is exponential with mean $\frac{1}{\alpha}$.
As a result, we henceforth refer to the parameter $t^*>0$ as the {\em characteristic time} of the Gamma process.
The characteristic time $t^*$ of a Gamma process $\Gamma$ is the unique deterministic waiting time
$t$ until
the distribution of $\Gamma_t$ is exponential.

Recall that the mean of $\Gamma_t$ is given by:
\be
E^{\mathbb{Q}} \Gamma_t = \frac{t}{\alpha t^*},
\label{mean}
\ee
for all $t \geq 0$.
If we  set the parameter $\alpha = 1/t^*$, then the
gamma process becomes unbiased, i.e. $E^{\mathbb{Q}} \Gamma_t = t$ for all $t\geq 0$.
In general, the variance of $\Gamma_t$ is:
\be
\frac{t}{\alpha^2 t^*}, \qquad t>0.
\label{varg}
\ee
As a result, the standard deviation of
an unbiased gamma process $\Gamma_t$ is just
the geometric mean of $t$ and $t^*$.
Setting $\alpha = 1/t^*$ in (\ref{varg}), we obtain the \emph{unbiased} Gamma process.
The distribution of the unbiased gamma process $\Gamma$ at time $t \geq 0$ is:
\be
\mathbb{Q} \{\Gamma_t \in ds \} = \frac{s^{t/t^* - 1}e^{-s/t^*}}{(t^*)^{t/t^*} \Gamma(t/t^*)} ds, \qquad s>0,
\label{gpdfub}
\ee
When $t = t^*$, this PDF is exponential, and the fact that the gamma process is unbiased implies that the mean and the standard deviation of $\Gamma_{t^*}$ are both $t^*$.

\begin{remark}\label{rem:biasedGamma}
The choice of $\alpha=1/t^*$ in the above construction is motivated merely by the desire to have a Gamma process which is unbiased: its expectation at time $t$ is equal to $t$. We consider this as a natural property, since, in what follows, we use Gamma process as a time change. However, it is not at all necessary for the Gamma process to be unbiased. In fact, the results of the subsequent sections will hold for any $\alpha>0$. In particular, if the unbiased Gamma process produces unrealistic paths (e.g. having a lot of very small jumps and very few extremely large ones), one may change the parameter $\alpha$ to obtain more realistic dynamics.\footnote{We thank Jan Ob{\l}{\'o}j for pointing out that the simulated paths of the unbiased Gamma process may look unrealistic, for certain values of $t^*$}
\end{remark}

\subsection{Construction of the Local Variance Gamma Process}
\label{fpidec}

We assume that the risk-neutral process of the underlying spot price $S$ is obtained by subordinating the driftless diffusion $D$ to an independent
unbiased gamma process $\Gamma$. Recall that $D$ is constructed as the canonical process on the space of continuous paths $\Omega^D = C\left([0,\infty)\right)$, equipped with the Borel sigma algebra $\mathcal{B}\left(C\left([0,\infty)\right)\right)$, and with the probability measures $\QQ^{D,x}$, for $x\in\RR$. 
On a different probability space $\left(\Omega^{\Gamma},\mathcal{F}^{\Gamma}\right)$, with a probability measure $\QQ^{\Gamma}$, we construct an unbiased gamma process $\Gamma$, with parameter $t^*$. Finally, on the product space $\left(\Omega^D\times\Omega^{\Gamma},\mathcal{B}\left(C\left([0,\infty)\right)\right)\otimes\mathcal{F}^{\Gamma}\right)$ we define the risk-neutral dynamics of the underlying, for every $(\omega_1,\omega_2)\in\Omega^D\times\Omega^{\Gamma}$, as follows:
\be
S_t(\omega_1,\omega_2) = D_{\Gamma_t(\omega_2)} (\omega_1), \qquad t \geq 0.
\label{tc}
\ee
It can be shown easily, by conditioning, that $S$ inherits the martingale property of $D$, with respect to its natural filtration, under any measure 
\begin{equation}\label{eq.Qx.def}
\QQ^x = \QQ^{D,x}\times\QQ^{\Gamma}
\end{equation}

\begin{remark}
Considered as a function of the forward spatial variable,
the PDF of $S_t$ (when it exists) may possess some unusual properties for $t$ small
but not infinitesimal. For example, when $a$ is constant, it is easy to see that
the PDF is infinite at $x$, for times $t < t^*/2$.
As a result, for short term options, the graph of value against strike
will be $C^1$, but not $C^2$. Similarly, gamma will not exist for
short term ATM options. For a piecewise continuous $a$, we conjecture that the PDF of $S_t$ has a jump at every point of discontinuity of $a$. Then, at every such point, the call price is $C^1$, but not $C^2$, viewed as a function of strike.
\end{remark}

As a diffusion time changed with an independent L\'evy clock, the process $S$, along with $\left\{\QQ^x\right\}$ (defined in (\ref{eq.Qx.def})), for $x\in\RR$, form a Markov family. The following proposition makes this statement rigorous and, in addition, shows how to reduce the computation of option prices in LVG model to the case of driftless diffusion.
Recall that, in a model where the risk-neutral dynamics of underlying are given by $S$, with initial value $S_0=x$, the time $t$ price of a European type option with payoff function $\phi$ and maturity $T$ is given by
\begin{equation}\label{eq.LVG.Vphit}
V^{\phi}_t(T) = \EE^x\left(\left. \phi(S_T) \right| \mathcal{F}^{S}_t \right),
\end{equation}
where $\mathcal{F}^{S}$ is the filtration generated by $S$, and, with a slight abuse of notation, we denote by $\EE^x$ the expectation with respect to $\QQ^x$. The following proposition, in particular, shows that option prices in a LVG model are determined by the price function, defined as
\begin{equation}\label{eq.contClaimPriceFunc}
V^{\phi}(\tau,x) = \EE^x \phi(S_{\tau}),
\end{equation}

\begin{prop}\label{th:lvg.le1}
The process $S$ and the family of measures $\left\{\QQ^x\right\}_{x\in\RR}$ form a Markov family (cf. Definition 2.5.11 in \cite{KaratzasShreve}). In particular, for any Borel measurable and exponentially bounded function $\phi$, the following holds, for all $T\geq0$ and $x\in\RR$: 
\begin{equation}\label{eq.lvg.call1}
V^{\phi}_t(T) = V^{\phi}(T-t,S_t),
\,\,\,\,\,\,\,\,\,\,\,\,\,\,\,V^{\phi}(\tau,x) = \int_0^{\infty} \frac{u^{\tau/t^* - 1}e^{-u/t^*}}{(t^*)^{\tau/t^*} \Gamma(\tau/t^*)} V^{D,\phi}(u,x)  du,
\end{equation}
where $V^{\phi}_t$, $V^{\phi}$, and $V^{D,\phi}$ are given by (\ref{eq.LVG.Vphit}), (\ref{eq.contClaimPriceFunc}), and (\ref{eq.DD.Vphi}), respectively. 
\end{prop}
\begin{proof}
According to Definition 2.5.11 and Proposition 2.5.13 in \cite{KaratzasShreve}, in order to prove that $S$ and $\left\{\QQ^x\right\}_{x\in\RR}$ form a Markov family, we need to show two properties. First, for any $F\in \mathcal{B}\left(C\left([0,\infty)\right)\right)\otimes\mathcal{F}^{\Gamma}$, the mapping $x\mapsto \QQ^x(F) = \left(\QQ^{D,x}\times\QQ^{\Gamma}\right)(F)$ is universally measurable. This property can be deduced easily from the universal measurability of the mapping $x\mapsto \QQ^{D,x}(F)$, for any $F\in\mathcal{B}\left(C\left([0,\infty)\right)\right)$, via the monotone class theorem. Secondly, we need to verify that, for any $B\in\mathcal{B}(\RR)$ and any $u,t\geq0$,
$$
\QQ^x\left(\left.S_{u+t}\in B  \right| \mathcal{F}^S_u\right) = \left.\QQ^y\left(S_t\in B\right)\right|_{y=S_u}
$$
The latter is done easily by conditioning on $\Gamma$ and using the Markov property of $D$. Similarly, one obtains the second equation in (\ref{eq.lvg.call1}). The first equation in (\ref{eq.lvg.call1}) follows from the Markov property.
\end{proof}

Suppose that time to maturity $\tau$ is equal to the characteristic time
$t^*$ of the gamma process.
Then (\ref{eq.lvg.call1}) yields
\begin{equation}\label{eq.pricefunc.LVG}
V^{\phi}(t^*,x) = \int_0^{\infty} \frac{e^{-u/t^*}}{t^*} V^{D,\phi}(u,x)du.
\end{equation}
In particular, for the call options, we obtain:
$$
C_t(T,K) = \EE^x\left(\left. (S_{T} - K)^+ \right| \mathcal{F}^{S}_t \right) = C(T-t,K,S_t),
$$
where $C(\tau,K,x)$ is the call price function in LVG model, satisfying
\begin{equation}\label{eq.callPriceFunc}
C(\tau,K,x) = \int_0^{\infty} \frac{u^{\tau/t^* - 1}e^{-u/t^*}}{(t^*)^{\tau/t^*} \Gamma(\tau/t^*)} C^D(u,K,x)  du,
\end{equation}
with $C^D$ given by (\ref{eq.DD.CD}).
When $\tau=t^*$, we have
\be
C(t^*,K,x) = \int_0^{\infty} \frac{e^{-u/t^*}}{t^*} C^D(u,K,x)du.
\label{sgs1}
\ee
The next section shows how the above representations can be used to
generate new equations that govern option prices.

\section{PDDE's for Option Prices}

In this section, we show that the additional structure
imposed on $S$ in Subsection \ref{fpidec}
causes the equations presented in
the previous section to reduce to much simpler \emph{partial differential difference equations} (PDDE's). The new equations can be used for a faster computation of option prices, as well as for exact calibration of the model to market prices of call options.
In this section and the next,
we will assume that the local variance rate function
$a^2(D)$ of the diffusion and the
characteristic time $t^*$ of the gamma process are
somehow known.
In the following section, we will discuss various ways in which this
positive function and this positive constant can be identified from market data.

\subsection{Black-Scholes PDDE for option prices}

Consider a European type contingent claim, which pays out $\phi(S_T)$ at the time of maturity $T$. Recall the price function of this claim in a LVG model, $V^{\phi}$, defined in (\ref{eq.contClaimPriceFunc}). Equation (\ref{eq.pricefunc.LVG}) implies that this price function is just a Laplace-Carson
\footnote{The Laplace-Carson transform
of a suitable function $f(t)$ is defined as
$\int_0^t \lambda e^{- \lambda t} f(t) dt$, where $\lambda$, in general, is a complex number
whose real part is positive.
}  
transform of the price function in diffusion model, where the transform argument
is evaluated at $1/t^*$.
Integrating both sides of the Black-Scholes equation (\ref{eq.DD.BS}), and observing that
$$
\lim_{\tau\downarrow 0} V^{\phi}(\tau,x) = \phi(x),
$$
we, heuristically, derive the Black-Scholes PDDE (\ref{eq.pdde.BS}) for the price function in LVG model.
Recall that the Black-Scholes PDE (\ref{eq.DD.BS}) is understood in a weak sense, due to the possible discontinuities of the coefficient $a^2$. The following theorem addresses this, as well as some other, difficulties, and makes the derivations rigorous.

\begin{theorem}\label{th:pdde.BS.th1}
Assume that $\phi(x)$ is once continuously differentiable in $x\in(L,U)$ and that it has zero limits at $L_+$ and $U_-$. Assume, in addition, that $\phi'$ is absolutely continuous and has a square integrable derivative. Then, the following holds.
\begin{enumerate}
\item For any $\tau\geq t^*$, the function $V^{\phi}(\tau,\cdot)$ (defined in (\ref{eq.contClaimPriceFunc})) possesses the same properties as $\phi$: it has zero limits at the boundary, it is once continuously differentiable, with absolutely continuous first derivative, and its second derivative is square integrable.
\item In addition, for all $x\in(L,U)$, except the points of discontinuity of $a$, $V^{\phi}(t^*,x)$ is twice continuously differentiable in $x$ and satisfies
\begin{equation}\label{eq.pdde.BS}
\frac{1}{2} a^2(x) \partial^2_{xx} V^{\phi}(t^*,x) - \frac{1}{t^*} \left( V^{\phi}(t^*,x) - \phi(x) \right) = 0
\end{equation}
\end{enumerate}
The properties 1-2 determine function $V^{\phi}(t^*,\cdot)$ uniquely.
\end{theorem}

The proof of Theorem $\ref{th:pdde.BS.th1}$ is given in Appendix A.
Note that the value of a contingent claim with maturity $T>t^*$, at an arbitrary future time $t\in(0,T-t^*]$, $V_t^{\phi}$, can be viewed as a payoff of another claim, with maturity $t$ and the payoff function $V^{\phi}(T-t,\cdot)$. Indeed, Theorem $\ref{th:pdde.BS.th1}$ states that function $V^{\phi}(T-t,\cdot)$ possesses the same properties as $\phi$. Therefore, if the current time to maturity is $t^*+\tau$, with $\tau\geq t^*$, the price function has to satisfy equation (\ref{eq.pdde.BS}), with $V^{\phi}(\tau,\cdot)$ in lieu of $\phi$:
\begin{equation}\label{eq.pdde.BS.1}
\frac{1}{2} a^2(x) \partial^2_{xx} V^{\phi}(t^*+\tau,x) - \frac{1}{t^*} \left( V^{\phi}(t^*+\tau,x) -V^{\phi}(\tau,x) \right) = 0
\end{equation}
The PDDE (\ref{eq.pdde.BS.1}) can be used to compute numerically the price function at all $\tau = nt^*$, for $n=1,2,\ldots$, by solving a sequence of one-dimensional ODE's, as opposed to a parabolic PDE (compare to (\ref{eq.DD.BS})). Namely, the price function can be propagated forward in $\tau$, starting from
the initial condition:
$$
V^{\phi}(0,x)=\phi(x),
$$
and solving (\ref{eq.pdde.BS.1}) recursively, to obtain the values at each next $\tau$.
In fact, we will show that, if $a$ is chosen to be 
piecewise constant, the above ODE can be solved in a closed form.

Furthermore, one can approximate the option value at an arbitrary time to maturity $\tau>0$. For $\tau=t^*$, one can compute option prices via the ODE (\ref{eq.pdde.BS}).
For $\tau\in(0,t^*)\cup(t^*,2t^*)$, the price function $V^{\phi}(\tau,x)$ can be approximated by Monte Carlo methods, or analytically, by computing a numerical solution to (\ref{eq.DD.BS}) and integrating it with the density of $\Gamma_{\tau}$. Having done this, one can use (\ref{eq.pdde.BS.1}) to propagate the price values forward in $\tau$, as described above.

\subsection{Dupire's PDDE for Call Prices}


In this subsection we focus on call options. We will derive equations that, although looking similar to the Black-Scholes equations, are of a very different nature, and are specific to the call (or put) payoff function.
As before, we, first, use equation (\ref{sgs1}) to conclude that the price function of a European call
in a LVG model is a Laplace-Carson transform of its price function in a diffusion model.
Then, similar to the heuristic derivation of (\ref{eq.pdde.BS}), we integrate both sides of the Dupire's equation (\ref{dup}) and, heuristically, derive the Dupire's PDDE for call prices in the LVG model:
\be
\frac{1}{2} a^2(K) \partial^2_{KK} C(t^*,K,x) - \frac{1}{t^*}\left(C(t^*,K,x) - (x-K)^+\right) = 0
\label{cfpdde}
\ee

One of the main obstacles in making the above derivation rigorous is that, as stated in Theorem $\ref{th:DD.th2}$, the Dupire's equation (\ref{dup}) can only be understood in a very weak sense. Let us show how to overcome this obstacle and prove that the call prices in LVG model do, indeed, satisfy equation (\ref{cfpdde}).

\begin{lemma}\label{le:pdde.DUP.le1}
For any $x\in(L,U)$, the call price function $C(t^*,K,x)$ (given by (\ref{sgs1})) is continuously differentiable, as a function of $K\in(L,U)$, and its derivative is absolutely continuous. Moreover, the second order derivative $\partial^2_{KK}C(t^*,\,.\,,x)$ is the density of $S_{t^*}$ on $(L,U)$.
\end{lemma}

The proof of Lemma $\ref{le:pdde.DUP.le1}$ is given in Appendix A. Using this result, it is not hard to derive the desired equation. 

\begin{theorem}\label{th:pdde.DUP.th1}
For any $x\in(L,U)$, the call price function $C(t^*,K,x)$ (given by (\ref{sgs1})) satisfies the boundary conditions: 
\begin{equation}\label{eq.LVG.pdde.boundcond}
\lim_{K\downarrow L} \left( C(t^*,K,x) - (x-K)^+\right) = \lim_{K\uparrow U} C(t^*,K,x) = 0
\end{equation}
In addition, the partial derivative of the call price function, $\partial_K C(t^*,K,x)$, is absolutely continuous, and its second derivative, $\partial^2_{KK} C(t^*,K,x)$, is square integrable in $K\in(L,U)$. Moreover, everywhere except the points of discontinuity of $a$, $\partial^2_{KK} C(t^*,\cdot,x)$ is continuous and satisfies (\ref{cfpdde}).  The call price function $C(t^*,\cdot,x)$ is determined uniquely by these properties. 
\end{theorem}

It is shown in the next section how the PDDE (\ref{cfpdde}) can be used for exact calibration of the LVG model to market prices of call options with a single maturity. In order to handle the case of multiple maturities, we need a mild generalization of (\ref{cfpdde}).
Note that (\ref{cfpdde}) can be interpreted as a no arbitrage constraint holding at $t=0$
between the value of a discrete calendar spread,
$$
\frac{C(t+t^*,K,x) - C(t,K,x)}{t^*},
$$
and the value of a limiting butterfly spread,
$$
\sop{}{K} C(t+t^*,K,x) = \lim\limits_{\triangle K \downarrow 0}
\frac{ C(t+t^*,K +\triangle K,x) - 2 C(t+t^*,K,x) + C(t+t^*,K - \triangle K,x)}
{(\triangle K)^2}
$$
This no arbitrage constraint would hold at all prior times
$\tau \leq 0$
and starting at any spot price level $x$.
Furthermore, due to the time homogeneity of the $S^x$ process,
we may write:
\be
\frac{a^2(K)}{2} \sop{}{K} C(\tau+t^*,K,x) - \frac{1}{t^*} \left(C(\tau+t^*,K,x) - C(\tau,K,x)\right) = 0.
\label{cfpdde1}
\ee


\begin{theorem}\label{th:pdde.DUP.th2}
For any $\tau \geq 0$ and $x\in(L,U)$, the call price function $C(\tau+t^*,\cdot,x)$ (given by (\ref{eq.callPriceFunc})) is the unique function that satisfies the boundary conditions (\ref{eq.LVG.pdde.boundcond}), has an absolutely continuous first derivative and a square integrable second derivative, and such that, everywhere except the points of discontinuity of $a$, its second derivative is continuous and satisfies (\ref{cfpdde1}).
\end{theorem}

A rigorous proof of Theorem $\ref{th:pdde.DUP.th2}$ is given in Appendix A.

\begin{remark}
We note that European put prices also satisfy the PDDE (\ref{cfpdde1}). This follows easily from the put-call parity.
\end{remark}

One can differentiate (\ref{cfpdde1}), to obtain, formally, a PDDE for deltas, gammas, and higher order spatial derivatives of
option prices. We, however, do not provide a rigorous derivation of such equations in this paper.
As in the case of Black-Scholes PDDE (\ref{eq.pdde.BS.1}), equation (\ref{cfpdde1}) can be used to compute numerically the price function at all $\tau = nt^*$, for $n=1,2,\ldots$, by solving a sequence of one-dimensional ODE's (\ref{cfpdde1}).
Having an approximation of $C(\tau,K,x)$, for $\tau\in(0,t^*)$, one can, similarly, compute the price values for all $\tau\geq0$. 

In the next section, we show that, if $a$ is chosen to be piecewise constant, the PDDE (\ref{cfpdde1}) can be solved in a closed form. We also show how this equation can be used to calibrate the model to market prices of call options of multiple strikes and maturities.

\section{Calibration}
\label{se:calib}

In the previous sections, we assumed that
the local variance function
$a^2(x)$ and the characteristic time $t^*$ are
somehow known.
In this section, we first discuss how one can deduce $a$ and $t^*$ from a family of observed call prices, for a single maturity and continuum of strikes. We, then, show how a mild extension of the LVG model can be calibrated to a finite number of price curves (or, implied smiles), for different maturities, and, still, continuum of strikes. Finally, we consider a more realistic setting and show how these continuous price curves (or, implied volatility curves) can be reconstructed from a finite number of option prices, by means of a LVG model with piecewise constant diffusion coefficient.
The resulting calibration algorithm only requires solution to a single linear feasibility problem and a finite number of one-dimensional root-search problems. It allows for exact calibration to an arbitrary (finite) number of strikes and maturities!

\subsection{Calibrating LVG model to a Continuous Smile}

Suppose that we are given current prices of a family of call options, $\left\{\bar{C}(K)\right\}$, for a single time to maturity $\tau^*$ and all strikes $K$ changing in the interval $(L,U)$. We can also observe the current level of underlying $x\in(L,U)$. Of course, we can, equivalently, assume the availability of implied smile and the current level of underlying. We assume that:
\begin{enumerate}

\item $\lim_{K\downarrow L} \left( \bar{C}(K) - (x-K)^+ \right) = \lim_{K\uparrow U} \bar{C}(K) = 0$;

\item ${\bar{C}}''$ is strictly positive and continuous everywhere in $[L,U]$, except a finite number of discontinuities of the first order;

\item $\left(\bar{C}(K) - (x-K)^+\right)/\bar{C}''(K)$ is bounded from above and away from zero, uniformly over $K\in(L,U)$. 

\end{enumerate}

Then, we can set $t^*=\tau^*$ and
$$
a^2(K) = \frac{2}{\tau^*}\frac{\bar{C}(K) - (x-K)^+}{\bar{C}''(K)}
$$
It is easy to see, that, under the above assumptions, $\bar{C}''$ is square integrable.
Then, Theorem $\ref{th:pdde.DUP.th2}$ implies that the LVG model with the above parameters reproduces the market call prices $\left\{ \bar{C}(K) \right\}$.

The method presented here is an example of an explicit exact calibration of a time-homogeneous martingale model to a continuous smile of an arbitrary (single) maturity. A more general construction, which, however, may lead to time-inhomogeneous dynamics, is described in \cite{CoxHobsonObloj}.

To motivate the above conditions 1-3, recall that our standing assumptions include zero carrying costs and the existence of a martingale measure $\QQ$ which produces prices of all contingent claims as the expectations of respective payoffs (cf. Subsection \ref{subse:ModelAss}). It is, then, natural to require that the observed market prices are consistent with this assumption: i.e. they are given by expectations of respective payoffs in some martingale model. In this case, the market call prices have to satisfy condition 1 above, along with $\bar{C}''(K)\geq 0$ and $\bar{C}(K) - (x-K)^+ \geq 0$. Thus, the above conditions 1-3 can be viewed as a stronger version of the no-arbitrage assumption. A more realistic setting, with only a finite number of traded options, is considered in Subsection \ref{subse:ContFromFinite}, where slightly less restrictive assumptions on the market data are introduced.



Now suppose that we plan to
recalibrate the model on a daily basis.
In the single smile setting, we may distinguish two types of options markets.
The first type is a fixed term market, where the time to maturity of the single observable smile remains constant as
calibration time moves forward. An example of a fixed term options market is the OTC
FX options market for an EM currency, where only one term is liquid.
The other type of options market that we may distinguish is a fixed expiry market, where
the time to maturity of the single observable smile declines linearly toward zero as
calibration time moves forward.
An example of a fixed expiry options market is the market for options on commodity futures\footnote{These
options are American-style but are rarely exercised early.}.

In a fixed term options market, the $t^*$ parameter would remain constant as
calibration time moves forward. In particular, if the result of the first calibration is $(a,t^*)$, then, it is \emph{possible} (albeit it unlikely) that the output of all subsequent calibrations will be the same: for example, if the true dynamics of the underlying are, indeed, given by an LVG model with these parameters.
Now consider a fixed expiry options market.
If we calibrate to one day options on a daily basis, there is no issue.
However, if we calibrate to longer dated options on a daily basis, then
the $t^*$ parameter would
drop through calibration time as we near maturity.
Hence in this case, calibrating an LVG model, we know a priori that the result of the next calibration will always be different from the current one. In particular, even if the underlying truly follows an LVG model with the parameters captured by the initial calibration, $(a,t^*)$, each subsequent recalibration of the model will lead to different parameter values. 
In this case, the LVG model can be used as a tool for arbitrage-free interpolation of option prices, but one cannot have much faith in the model itself.

\subsection{Calibrating to Multiple Continuous Smiles}
\label{subse:calib.multsmiles}

For many types of underlying, options of multiple maturities and strikes trade liquidly.
For OTC currency options, the maturity dates at which there is price transparency move through
calendar time, so that the time to maturity for each liquid option remains constant as
calendar time evolves.
In contrast, for listed stock options,
the maturity dates at which there is price transparency are fixed calendar times.
Hence, the time to maturity of each listed stock option shortens as
calendar time evolves.
In this section, we address the issue of calibrating to multiple smiles.
In order to do this, we, first, develop a non-homogeneous extension of LVG model.

Recall that, if one wishes to match market-given quotes at a single strike and multiple discrete maturities,
then one can extend the constant volatility Black-Scholes model
by assuming that the square of volatility is a piecewise constant function of time.
Analogously, in order to match market-given smiles at several discrete maturities,
we extend the LVG model so that the local variance function $a^2$ and the parameter $t^*$
are piecewise constant functions of time.
Consider a finite collection of LVG families, 
$$
\left\{S^m,\left\{\QQ^{m,x}\right\}_{x\in\RR}\right\}_{m=1}^M,
$$ 
with the parameters $\left\{a_m,t^*_m,L_m,U_m\right\}_{m=1}^M$, respectively. 
Recall that we extend the definition of each LVG process to all initial conditions $x\in\RR$ by assuming that $\QQ^{m,x}\left(S^m_t=x,\,\,\text{for all}\,\,t\geq0\right)=1$, for any $x\in \RR\setminus(L_m,U_m)$.
For $m=1,\ldots,M$, we denote by $\QQ^{m,x}_t$ the marginal distribution of $S^{m}_t$ under $\QQ^x$. We assume that $\QQ^{M+1,x}_t=\QQ^{M,x}_t$. We also introduce a sequence of times $\left\{T_m\right\}_{m=0}^{M+1}$ via $T_m = \sum_{i=1}^m t^*_i$, $T_0=0$, and $T_{M+1}=\infty$.
Finally, we define the {\bf non-homogeneous LVG process} $\tilde{S}$, with parameters $\left\{a_m,t^*_m,L_m,U_m\right\}_{m=1}^M$, as a non-homogeneous Markov process, whose transition kernel is defined, for all $t\in[T_{i},T_{i+1})$ and $T\in [T_{i+j},T_{i+j+1})$, with $t\leq T$, $0\leq i\leq M$, and $0\leq j \leq M-i$, as follows:
$$
p(t,x;T,B) = \int_{\RR}\int_{\RR}\cdots\int_{\RR}  \QQ^{i+j+1,x_{i+j}}_{T-T_{i+j}}(B) 
\QQ^{i+j,x_{i+j-1}}_{T_{i+j}-T_{i+j-1}}(dx_{i+j})\cdots 
\QQ^{i+2,x_{i+1}}_{T_{i+2}-T_{i+1}}(dx_{i+2}) \QQ^{i+1,x}_{T_{i+1}-t}(dx_{i+1}),
$$
where $B\subset \RR$ is an arbitrary Borel set. Notice that the above integral is well defined, since every mapping $x\mapsto \QQ^{m,x}_t$ is universally measurable (cf. Definition 2.5.11, on p. 74, in \cite{KaratzasShreve}).
Using the above transition kernel, for any fixed initial condition $x\in\RR$, it is a standard exercise to construct a candidate family of finite-dimensional distributions of $\tilde{S}$, so that it is consistent. Then, due to the Kolmogorov's existence theorem, there exists a unique probability measure $\tilde{\QQ}^x$, on the space of paths, that reproduces these finite-dimensional distributions. For every initial condition $x\in\RR$, the LVG process $\tilde{S}$ is constructed as a canonical process on the space of paths, under the measure $\tilde{\QQ}^x$.

Intuitively, the process $\tilde{S}$ evolves as $S^{m+1}$, between the times $T_m$ and $T_{m+1}$, with the initial condition being the left limit of the process $\tilde{S}$ at the end of the previous interval, $[T_{m-1},T_m]$ . However, making such definition rigorous is not straightforward, since it requires certain properties of the processes $S^{m}$ as functions of the initial condition $x\in\RR$ (e.g. the measurability in $x$). Recall that, due to potential discontinuity of $a$, we had to construct each LVG process $S^{m}$ using the weak, rather than strong, solutions to equation (\ref{D}). This, in particular, makes it rather hard to analyze the dependence of $S^{m}$ on the initial condition $x$ in the almost sure sense. This is why we introduced the family of measures $\left\{\QQ^{D,x}\right\}$, and, in turn, $\left\{\QQ^{x}\right\}$, and established the measurability of these families as functions of $x$. However, once the distribution of $\tilde{S}$ is constructed, we no longer need to keep track of its dependence on the initial condition. In particular, it is not necessary to construct the process $\tilde{S}$ on a canonical space (which supports all measures $\tilde{\QQ}^x$): one can construct $\tilde{S}$ for every different value of the initial condition $x$ separately, possibly, on a different probability space. We chose to construct the process $\tilde{S}$ as shown above, only in order to be consistent with the constructions made in preceding sections.



Let us compute the call prices in a model where the risk-neutral dynamics of the underlying are given by $\tilde{S}$ and the market filtration, $\mathcal{F}^{\tilde{S}}$, is generated by the process.
Recall that $\tilde{S}$ is a non-homogeneous Markov process. In particular, for $T_m\leq t\leq T\leq T_{m+1}$, we can compute the time $t$ price of a European call option with strike $K$ and maturity $T$ as follows:
$$
\tilde{C}_t(T,K) = \tilde{\EE}^x\left((\tilde{S}_T - K)^+ \,|\,\mathcal{F}^{\tilde{S}}_t \right) 
= \int_{\RR} p(t,\tilde{S}_t;T,y) (y-K)^+ dy
= \int_{\RR} (y-K)^+ \QQ^{m+1,\tilde{S}_t}_{T-t}(dy)
$$
\begin{equation}\label{eq.callPrice.nonhomLVG}
= C^{m+1}(T-t,K,\tilde{S}_t),
\end{equation}
where $\tilde{\EE}^x$ denotes the expectation with respect to $\tilde{\QQ}^x$ and $C^{m+1}$ is the call price function associated with the LVG model $S^{m+1}$ (given by (\ref{eq.callPriceFunc})). Notice that the time zero call price in a non-homogeneous LVG model, with initial condition $x$, is given by
$$
\tilde{C}_0(T,K) = \tilde{C}(T,K,x) 
$$
where the call price function $\tilde{C}$ is defined as
$$
\tilde{C}(T,K,x)= \tilde{\EE}^x\left(\tilde{S}_T - K\right)^+.
$$
Then
$$
\tilde{C}(T_{m+1},K,x) - \tilde{C}(T_m,K,x)
=\tilde{\EE}^x \left( \tilde{C}_{T_m}(T_{m+1},K) - \tilde{C}_{T_m}(T_{m},K) \right)
$$
$$
= \tilde{\EE}^x \left( C^{m+1}(T_{m+1}-T_m,K,\tilde{S}_{T_m}) - (\tilde{S}_{T_m} - K)^+ \right)
= \tilde{\EE}^x \left( C^{m+1}(t^*_{m+1},K,\tilde{S}_{T_m}) - (\tilde{S}_{T_m} - K)^+ \right)
$$
$$
= \frac{t^*_{m+1}}{2} a_{m+1}^2(K) \tilde{\EE}^x \left( \partial^2_{KK} C^{m+1}(t^*_{m+1},K,\tilde{S}_{T_m}) \right)
= \frac{t^*_{m+1}}{2} a_{m+1}^2(K) \partial^2_{KK} \tilde{\EE}^x C^{m+1}(t^*_{m+1},K,\tilde{S}_{T_m})
$$
$$
= \frac{t^*_{m+1}}{2} a_{m+1}^2(K) \partial^2_{KK} \EE^x C_{T_m}(T_m+t^*_{m+1},K)
= \frac{t^*_{m+1}}{2} a_{m+1}^2(K) \partial^2_{KK} \tilde{C}(T_{m+1},K,x)
$$
where we made use of (\ref{cfpdde}). In the above, we interchanged the differentiation and expectation, using the Fubini's theorem. To justify this, we recall that the function $\partial^2_{KK} C^{m+1}(t^*_{m+1},K,\cdot)$ is measurable, as a limit of continuous functions (since this derivative exists in a classical sense, everywhere except a finite number of points $K$). In addition, it is absolutely bounded, due to equation (\ref{cfpdde1}) and the fact that $C^{m+1}(t^*_{m+1},K,x) - (x-K)^+$ vanishes at $x\uparrow U$ and $x\downarrow L$ (which, in turn, can be shown by a dominated convergence theorem)
Thus, we conclude that, for any $m=1,\ldots,M$, the call price function $\tilde{C}(T_m,\cdot,x)$ satisfies
\be
\frac{1}{2} a_{m}^2(K) \partial^2_{KK} \tilde{C}(T_m,K,x) - \frac{1}{t^*_m}\left(\tilde{C}(T_m,K,x) - \tilde{C}(T_{m-1},K,x)\right) = 0
\label{eq.calib.callPDDE}
\ee

\begin{theorem}\label{th:calib.multsmiles.th1}
For any $x\in(L_m,U_m)$, consider the time zero call price in a non-homogeneous LVG model, $\tilde{C}(T_m,K,x)$ (given by (\ref{eq.callPrice.nonhomLVG})). Then, $\tilde{C}(T_m,K,x)$ is the unique function of $K\in(L_m,U_m)$ that satisfies the boundary conditions (\ref{eq.LVG.pdde.boundcond}), has an absolutely continuous first derivative and a square integrable second derivative, and such that its second derivative is continuous and satisfies (\ref{eq.calib.callPDDE}) everywhere except the points of discontinuity of $a_m$.
\end{theorem}
\begin{proof}
We have already shown that $\tilde{C}(T_m,\cdot,x)$ satisfies (\ref{eq.calib.callPDDE}) and possesses all the properties stated in the above theorem. Notice that the homogeneous version of (\ref{eq.calib.callPDDE}) (with zero in place of $\tilde{C}(T_{m-1},K,x)$) is the same as the homogeneous version of (\ref{cfpdde1}). The uniqueness of solution to the latter equation follows from Theorem $\ref{th:pdde.DUP.th2}$. This completes the proof of Theorem $\ref{th:calib.multsmiles.th1}$.
\end{proof}

Now, the calibration strategy for multiple maturities becomes obvious.
Suppose that the market provides us with continuous price curves for call options at multiple maturities $0<T_1 < T_2 < \cdots < T_M<\infty$:
$$
\left\{\bar{C}^m(K)\,:\, K\in(L_m,U_m)\right\},\,\,\text{for}\,\,m=1,\ldots,M,
$$
as well as the current underlying level $x\in\RR$.
Equivalently, we can assume the knowledge of implied smiles for a continuum of strikes and multiple maturities. Using the notational convention $T_0=0$, we define $\bar{C}^0(K) = (x-K)^+$, for all $K\in\RR$. In addition, we extend each market price curve $\bar{C}^m(K)$ to all $K\in\RR$, recursively, starting with $m=1$, and defining $\bar{C}^m(K) =\bar{C}^{m-1}(K)$, for $K\notin (L_m,U_m)$.

\begin{assump}\label{assump:ass.1}
For all $m=1,\ldots,M$, we assume that:
\begin{enumerate}
\item $\lim_{K\downarrow L_m} \left( \bar{C}^m(K) - \bar{C}^{m-1}(K) \right) = \lim_{K\uparrow U_m} \left( \bar{C}^m(K) - \bar{C}^{m-1}(K) \right) = 0$;

\item $\partial^2_{KK} \bar{C}^m$ is positive and continuous everywhere in $[L_m,U_m]$, except a finite number of discontinuities of the first order;

\item $\left(\bar{C}^m(K) - \bar{C}^{m-1}(K)\right)/\partial^2_{KK}\bar{C}^m(K)$ is bounded from above and away from zero, uniformly over $K\in(L_m,U_m)$. 
\end{enumerate}
\end{assump}

Then, we can set $t^*_m=T_m - T_{m-1}$ and
\begin{equation}\label{eq.calib.a}
a^2_m(K) = \frac{2}{t^*_m}\frac{\bar{C}^m(K) - \bar{C}^{m-1}(K)}{\partial^2_{KK}\bar{C}^m(K)},
\end{equation}
for $m=1,\ldots,M$
It is easy to see, that, under the above assumptions, each $\partial^2_{KK}\bar{C}^m$ is square integrable, for $m=1,\ldots,M$.
Then, Theorem $\ref{th:calib.multsmiles.th1}$ implies that the non-homogeneous LVG model with the parameters $\left\{a_m,t^*_m,L_m,U_m\right\}_{m=1}^M$ reproduces the market price curves $\left\{ \bar{C}^m \right\}$.

When the maturity spacing is not uniform, then the above calibration strategy
may lead to grossly time-inhomogeneous dynamics
for the resulting gamma process.
In particular, even if the true risk-neutral dynamics of the underlying are given by a time-homogeneous Markov process, the calibrated process may have very inhomogeneous paths. This occurs when the times between available maturities vary significantly: then, $t^*_i = T_i - T_{i-1}$ varies with $i$, and the associated unbiased Gamma process has either many small jumps, or few large ones, depending on the time interval. This phenomenon is discussed in Remark $\ref{rem:biasedGamma}$, where it is also suggested that, in order to fix, or mitigate, the problem, one may use a biased, rather than unbiased, Gamma process (i.e. with $\alpha\neq 1/t^*$), which provides more flexibility for controlling the path properties of the process.
If one is nonetheless worried about the lack of homogeneity in the maturities,
one can sometimes add data to induce a more uniform maturity spacing. Of course, in this case, the additional data will affect the dynamics of calibrated process, and one may need to choose this data accordingly, to reproduce the desired characteristics of the paths.

\begin{remark}\label{rem:rem.AltMatInterp}
Notice that the above model for $\tilde{S}$ can be modified to obtain other methods of interpolating option prices across maturities. Namely, instead of using a Gamma process to construct each $S^{m,x}$, we can time change a driftless diffusion by any independent increasing process whose distribution at time $t^*_m$ is exponential. It is easy to see that, in this case, option prices will satisfy equation (\ref{eq.calib.callPDDE}). This was already observed in \cite{CoxHobsonObloj}, in the case of a single maturity. However, when only one maturity is available, the use of Gamma process is justified by the fact that it is the only L\'evy process that has an exponential marginal. Therefore, it is the only possible time change that produces time-homogeneous  Markov dynamics for the underlying. On the other hand, once we calibrate to multiple maturities, the time homogeneity is, typically, lost, and there is no particular reason to use Gamma process anymore.
\end{remark}

\subsection{Constructing Arbitrage-free Smiles from a Finite Family of Call Prices}
\label{subse:ContFromFinite}




In this subsection, we show how to construct continuous call price curves, satisfying Assumption $\ref{assump:ass.1}$ in Subsection \ref{subse:calib.multsmiles}, from a finite number of observed option prices. This can be viewed as an interpolation problem. Nevertheless, due to the non-standard constraints given by Assumption $\ref{assump:ass.1}$, this problem cannot be solved by a simple application of existing methods, such as polynomial splines. In particular, the second part of Assumption $\ref{assump:ass.1}$ implies that we have to restrict the interpolating curves to the space of convex functions, while the other parts of Assumption $\ref{assump:ass.1}$ introduce several additional levels of difficulty.
Perhaps, the most popular existing method for cross-strike interpolation of option prices is known as the SVI (or SSVI) parameterization (see, for example, \cite{SVI} and references therein). In order to fit a function from the SVI family to the observed implied volatilities, one solves numerically a multivariate optimization problem, to find the right values of the parameters. Then, provided these values satisfy certain no-arbitrage conditions, one can easily obtain the interpolating call price curves, for each maturity, such that they satisfy conditions $1-2$ of Assumption $\ref{assump:ass.1}$. The SVI family has many advantages: in particular, it allows one to produce smooth implied volatility (and call price) curves using only few parameters. In addition, each of the parameters has a certain financial interpretation, which makes SVI useful for developing intuition about the observed set of option prices. As it is shown, for example, in \cite{SVI}, the empirical quality of SVI fit, applied to the options on S$\&$P 500, is quite good. However, the SVI family was not designed to fit an arbitrary combination of arbitrage-free option prices, which also reveals itself in the numerical results of \cite{SVI}, where the interpolated implied volatility, occasionally, crosses the bid and ask values.
In this subsection, we provide an interpolation method that is guaranteed to succeed for any given (strictly admissible) set of option prices. We also present an explicit algorithm for constructing such an interpolation, which does not use any multivariate optimization. In addition, the interpolation method proposed here produces price curves that satisfy the, rather non-standard, condition $3$ of Assumption $\ref{assump:ass.1}$ (which is needed for cross-maturity calibration). Finally, the proposed method is particularly appealing in the setting of the present paper, as the interpolating price curves are constructed via the LVG models.

We start by considering a specific class of LVG processes that have a {\bf piecewise constant} variance function $a^2$. Assume that $L$ and $U$ are both finite, and
$$
a(x) = \sum_{j=1}^{R+1} \sigma_{j} \bone_{\left[\nu_{j-1},\nu_j\right)}(x),
$$
for a partition $L=\nu_{0}<\nu_1<\ldots < \nu_{R} < \nu_{R+1}=U$ and a set of strictly positive numbers $\left\{\sigma_j\right\}_{j=1}^{R+1}$.
The choice of $t^*$, in this case, is not important. To simplify the notation, we denote $z=\sqrt{2/t^*}$.
In addition, for fixed $x$ and $z$, we denote 
$$
\chi(K) = C\left(\frac{1}{z^2},K,x\right)
$$
The PDDE (\ref{cfpdde}), in this model, becomes
\begin{equation}\label{eq.5}
a^2(K) \chi''(K) - z^2 \chi(K) = -z^2(x-K)^+
\end{equation}
Notice that, for $K\in(L,x)$, $V(K)=\chi(K)-(x-K)$ satisfies the homogeneous version of the above equation:
\begin{equation}\label{eq.calib.TV.2}
a^2(K) V''(K) - z^2 V(K) = 0
\end{equation}
For $K\in(x,U)$, $V(K)=\chi(K)$ satisfies the above equation as well. Notice that we have introduced the time-value function 
$$
V(K) = \chi(K) - (x-K)^+
$$
Function $V$ is {\bf not} a global solution to equation (\ref{eq.calib.TV.2}), as its weak second derivative contains delta function at $K=x$. However, on each interval, $(L,x)$ and $(x,U)$, it is a once continuously differentiable solution to (\ref{eq.calib.TV.2}), satisfying zero boundary conditions at $L$ or $U$, respectively. In addition, $V(K)$ is continuous at $K=x$ and $V'(K)$ has jump of size $-1$ at $K=x$. It turns out that these conditions determine function $V$ uniquely, as follows.

\begin{itemize}

\item For $K\in(L,x)$, $V(K)=V^1(K)$, where $V^1\in C^1(L,x)$ satisfies (\ref{eq.calib.TV.2}), with zero initial condition at $K=L$, and, hence, has to be of the form:
\begin{equation}\label{eq.2}
V^1(K) = \sum_{j=1}^{R+1} \left( c_j^{1,1} e^{-Kz/\sigma_j} + c_j^{2,1} e^{Kz/\sigma_j} \right) \bone_{\left[\nu_{j-1},\nu_j\right)}(K)
\end{equation}
with the coefficients $c_j^{1,1}$ and $c_j^{2,1}$ determined recursively:
\begin{eqnarray}
&&c_{j+1}^{1,1} e^{-z\nu_{j}/\sigma_{j+1}} = \frac{1}{2}  \left( \left(1+\frac{\sigma_{j+1}}{\sigma_j}\right) c_{j}^{1,1} e^{-z\nu_{j}/\sigma_j}
+ \left(1-\frac{\sigma_{j+1}}{\sigma_j}\right) c_{j}^{2,1} e^{z\nu_{j}/\sigma_j} \right),\nonumber\\
&&\label{eq.3}c_{j+1}^{2,1} e^{z\nu_{j}/\sigma_{j+1}} = \frac{1}{2} \left(
\left(1-\frac{\sigma_{j+1}}{\sigma_j}\right) c_{j}^{1,1} e^{-z\nu_{j}/\sigma_j}
+ \left(1+\frac{\sigma_{j+1}}{\sigma_j}\right) c_{j}^{2,1} e^{z\nu_{j}/\sigma_j} \right),
\end{eqnarray}
for $j\geq1$, starting with $c_{1}^{1,1}=-\lambda_1 e^{zL/\sigma_1}$ and $c_{1}^{2,1} = \lambda_1 e^{-zL/\sigma_1}$, with some $\lambda_1>0$.

\item For $K\in(x,U)$, $V(K)=V^2(K)$, where $V^2\in C^1(x,U)$ satisfies (\ref{eq.calib.TV.2}), with zero initial condition at $K=U$, and, hence, has to be of the form:
\begin{equation}\label{eq.2.1}
V^2(K) = \sum_{j=1}^{R+1} \left( c_j^{1,2} e^{-Kz/\sigma_j} + c_j^{2,2} e^{Kz/\sigma_j} \right) \bone_{\left[\nu_{j-1},\nu_j\right)}(K)
\end{equation}
with the coefficients $c_j^{1,2}$ and $c_j^{2,2}$ determined recursively:
\begin{eqnarray}
&&c_{j}^{1,2} e^{-z\nu_{j}/\sigma_{j}} = \frac{1}{2} \left( \left(1+\frac{\sigma_{j}}{\sigma_{j+1}}\right) c_{j+1}^{1,2} e^{-z\nu_{j}/\sigma_{j+1}}
+ \left(1-\frac{\sigma_{j}}{\sigma_{j+1}}\right) c_{j+1}^{2,2} e^{z\nu_{j}/\sigma_{j+1}} \right),\nonumber\\
&&\label{eq.6}c_{j}^{2,2} e^{z\nu_{j}/\sigma_{j}} = \frac{1}{2}  \left(
\left(1-\frac{\sigma_{j}}{\sigma_{j+1}}\right) c_{j+1}^{1,2} e^{-z\nu_{j}/\sigma_{j+1}}
+ \left(1+\frac{\sigma_{j}}{\sigma_{j+1}}\right) c_{j+1}^{2,2} e^{z\nu_{j}/\sigma_{j+1}} \right),
\end{eqnarray}
for $j\leq R$, starting with $c_{R+1}^{1,2}= \lambda_2 e^{zU/\sigma_{R+1}}$ and $c_{R+1}^{2,2} = -\lambda_2 e^{-zU/\sigma_{R+1}}$, with some $\lambda_2>0$.

\item The constants $\lambda_1>0$ and $\lambda_2>0$ should be chosen to fulfill the $C^1$ property of $\chi(K)$ at $K=x$. Namely, we denote the above functions $V^1(K)$ and $V^2(K)$, by $V^1(\lambda_1,K)$ and  $V^2(\lambda_2,K)$, respectively. We need to show that there exists a unique pair $(\lambda_1,\lambda_2)$, such that: $V^1(\lambda_1,x)=V^2(\lambda_2,x)$ and $\partial_KV^1(\lambda_1,x)=\partial_K V^2(\lambda_2,x) + 1$. It is clear that
$$
V^1(\lambda_1,K) = \lambda_1 V^1(1,K),
\,\,\,\,\,\,\,\,\,V^2(\lambda_2,K) = \lambda_2 V^2(1,K)
$$
Thus, we need to find $\lambda_1>0$ and $\lambda_2>0$, satisfying
$$
\lambda_1 V^1(1,x) = \lambda_2 V^2(1,x),
\,\,\,\,\,\,\,\,\,\,\lambda_1 \partial_KV^1(1,x) = \lambda_2 \partial_KV^2(1,x) + 1
$$

\begin{lemma}\label{le:calib.le1}
Function $V^1(1,K)$ is strictly increasing in $K\in(L,U)$, and $V^2(1,K)$ is strictly decreasing in $K\in(L,U)$.
\end{lemma}
The proof of Lemma $\ref{le:calib.le1}$ is straightforward. The $C^1$ property and equation (\ref{eq.calib.TV.2}) imply convexity of $V^1$. In addition, the choice of the coefficients $c_{1}^{1,1}$ and $c_{1}^{2,1}$ ensures that $V^1(L_+)=0$ and ${V^1}'(L_+)>0$. Hence, $V^1$ stays strictly increasing in $(L,U)$. Similarly, we can show that $V^2$ is strictly decreasing.
Lemma $\ref{le:calib.le1}$ ensures that there is a unique solution $(\lambda^*_1>0,\lambda^*_2>0)$ to the above system of equations.

\item Thus, the time value function is uniquely determined by
$$
V(K) = V^1(\lambda^*_1,K)\bone_{K\leq x} + V^2(\lambda^*_2,K)\bone_{K> x}
$$

\end{itemize}

We denote by $V^{\nu,\sigma,z,x}$ the time value function produced by an LVG model with piecewise constant diffusion coefficient $a$, given by $\nu=\left\{ \nu_j\right\}_{j=0}^{R+1}$ and $\sigma=\left\{ \sigma_j\right\}_{j=1}^{R+1}$, with $t^*=2/z^2$, and with the initial level of underlying $x$.
Similarly, we denote the call prices produced by such model, for maturity $t^*=2/z^2$ and all strikes $K\in(L,U)$, via
\begin{equation}\label{eq.C.nu.sigma.z.x}
C^{\nu,\sigma,z,x}(K) = V^{\nu,\sigma,z,x}(K) + (x-K)^+
\end{equation}

Now, we can go back to the problem of calibration. 

\begin{assump}\label{assump:ass.2}
We make the following assumptions on the structure of available market data.
\begin{enumerate}
\item We are given a finite family of maturities $0<T_1 < T_2 < \cdots < T_M<\infty$, and, for each maturity $T_i$, there is a set of available strikes $-\infty < K^i_1 < \cdots < K^i_{N_i} < \infty$. In addition, the following is satisfied, for each $i=1,\ldots,M-1$:
$$
\left(\left\{ K^{i+1}_j \right\}_{j=1}^{N_{i+1}} \setminus \left\{ K^{i}_j \right\}_{j=1}^{N_i}\right)
\cap \left(K^i_1, K^i_{N_i}\right) = \emptyset
$$
In other words, we assume that each strike available for the later maturity is either available for the earlier one as well, or has to lie outside of the range of strikes available for the earlier maturity.

\item We observe market prices of the corresponding call options:
$$
\left\{ \bar{C}(T_i,K^i_j)\,:\, j=1,\ldots,N_i,\,i=1,\ldots,M \right\},
$$
as well as the current underlying level $x\in\RR$.

\item In addition, for each $i=1,\ldots,M$, we are given an interval $(L_i,U_i)$, such that:
$$
(K_1^i,K^i_{N_i}) \subset (L_i,U_i) \subset(L_{i+1},U_{i+1}),
$$
$$
x\in(L_i,U_i) \subset 
\left( \max_{m> i,\, j\geq1}\left\{ K^{m}_j\,:\, K^{m}_j < K^i_1 \right\} , \min_{m> i,\, j\geq1}\left\{ K^{m}_j\,:\, K^{m}_j > K^i_{N_i} \right\} \right)
$$
where we use the standard convention: $\max \emptyset = -\infty$ and $\min \emptyset = \infty$. 
The interval $[L_i,U_i]$ represents the set of possible values of the underlying on the time interval $(T_{i-1},T_i]$. Of course, there exist infinitely many families $\left\{ L_i,U_i\right\}$ that satisfy the above properties. However, the economic meaning of the underlying may restrict, or even determine uniquely, the choice of $\left\{ L_i,U_i\right\}$ (e.g. if the underlying is an asset price, then, it is natural to choose $L_i=0$). The intervals $(L_i,U_i)$ can remain the same, for all $i\geq1$, only if all strikes available for the later maturities are also available for the earlier ones. 

\end{enumerate}
\end{assump}

The next property of the market data is closely related to the absence of arbitrage, and, therefore, we present it separately.

\begin{definition}\label{def:strict.admiss}
The market data $\left\{T_i, K^i_j, \bar{C}(T_i,K^i_j), L_i, U_i,x\right\}$, satisfying Assumption $\ref{assump:ass.2}$, is called {\bf strictly admissible} if the following holds, for each $i=1,\ldots,M$:
\begin{itemize}

\item the linear interpolation of the graph 
$$
\left\{\left(x-L_i,L_i\right), \left(\bar{C}(T_i,K^i_1),K^i_1\right), \ldots, \left(\bar{C}(T_i,K^i_{N_i}),K^i_{N_i}\right), \left(0,U_i\right)\right\},
$$ 
is strictly decreasing and convex, and no three distinct points in the above sequence lie on the same line;

\item for all $j=1,\ldots,N_i$, we have $\bar{C}(T_i,K^i_j)>(x-K^i_j)^+$, and, in addition, if $i>1$ and $K^i_j \in (L_{i-1},U_{i-1})$, then $\bar{C}(T_i,K^i_j)>\bar{C}(T_{i-1},K^{i-1}_j)$.

\end{itemize}
\end{definition}

In what follows, we assume that the market data is strictly admissible. In practice, due to the presence of transaction costs, this additional assumption is no loss of generality.

\begin{remark}\label{rem:addStrikes}
The simplest possible setting in which the above assumption on the structure of available strikes (Assumption $\ref{assump:ass.2}$) is satisfied, is when the same strikes are available for all maturities. Then, we can define $L_i$ and $U_i$ to be independent of $i\geq1$.
Even though Assumption $\ref{assump:ass.2}$ allows for a slightly more general structure of available strikes, in practice, this assumption is often violated. In particular, this is typically the case when one considers discounted prices (recall that strikes shift after discounting). However, if option prices for some strikes are not available for earlier maturities, one can fill in the missing prices so that they satisfy the above strict admissibility conditions. This amounts to a multivariate optimization problem, which, nevertheless, is particularly simple, as the constraints are linear. Such problems are known as the {\bf linear feasibility} problems, and they can be solved rather efficiently, for example, by means of the {\bf interior point method} (cf. \cite{Boyd}).
\end{remark}

The problem, now, is to interpolate the observed call prices across strikes. Namely, we need to find a family of functions
$$
K \mapsto \bar{C}^i(K),\,\,\,\,\,\, K\in\RR,
$$
for $i=1,\ldots,M$, such that: they match the observed market prices, $\bar{C}^i(K_j)=\bar{C}(T_i,K_j)$, for all $i$ and $j$, and the family of price curves $\left\{\bar{C}^i\right\}$ satisfies Assumption $\ref{assump:ass.1}$, stated in Subsection $\ref{subse:calib.multsmiles}$.
We provide an explicit algorithm for constructing such functions, making use of the LVG models with piecewise constant diffusion coefficients.

\begin{theorem}\label{th:calib.th.main}
For any $z>0$, and any positive integers $M$ and $\left\{N_i\right\}_{i=1}^M$, there exists a function that maps any strictly admissible market data, $\left\{ T_i, K^i_j, \bar{C}(T_i,K^i_j), L_i, U_i, x\right\}_{i=1,\ldots,M,\, j=1,\ldots,N_i}$ to the vector of parameters: 
$$
\left\{\nu^i=\left\{L_i=\nu^i_0<\nu^i_1<\cdots<\nu^i_{3M(N+2)}<\nu^i_{3M(N+2)+1}=U_i\right\}, \sigma^i=\left\{ \sigma^i_1>0,\ldots,\sigma^i_{3M(N+2)+1}>0\right\} \right\}_{i=1}^M,
$$ 
such that the associated call price curves $\left\{\bar{C}^i=C^{\nu^i,\sigma^i,z,x}\right\}_{i=1}^M$, defined in (\ref{eq.C.nu.sigma.z.x}), match the given values $\left\{\bar{C}(T_i,K^i_j)\right\}$ and satisfy Assumption $\ref{assump:ass.1}$.

Moreover, this mapping can be expressed as a finite superposition of elementary functions ($\exp$, $+$, $\times$, $/$, $\max$), real constants, and functions $f: \RR^n \rightarrow\RR$, such that each value $f(x)$ is determined as the solution to
$$
g(x,y)=0,\,\,\,\,\,\,y\in[a(x),b(x)],
$$
where $a,b:\RR^{n}\rightarrow\RR\cup\left\{-\infty\right\}\cup\left\{+\infty\right\}$ and $g:\RR^{n+1}\rightarrow\RR$ are finite superpositions of elementary functions and real constants, satisfying: $a(x)<b(x)$ and $g(x,\cdot)$ is continuous, with exactly one zero, on $\left(a(x),b(x)\right)$.
\end{theorem}

\begin{remark}
The vectors $\nu^i$ and $\sigma^i$ do not always need to have a length of exactly $3M(N+2)+2$ and $3M(N+2)+1$. In fact, if the length of these vectors is smaller than the aforementioned number, we can always add dummy entries to $\nu^i$, making the corresponding entries of $\sigma^i$ be equal to each other and to one of the adjacent entries. Thus, $3M(N+2)+2$ should be understood as the maximal possible length of $\nu^i$.
\end{remark}

\begin{remark}\label{rem:diffInterp_vs_diffCalib}
The above theorem provides a method for \textbf{explicit exact interpolation} of call prices across strikes.
To interpolate these price curves across maturities, we apply the method described in Subsection $\ref{subse:calib.multsmiles}$.
Note that the pice-wise constant diffusion coefficients, which produce the interpolated price curves, {\bf do not} have to coincide with the diffusion coefficients of the non-homogeneous LVG model, calibrated to these price curves as shown in Subsection $\ref{subse:calib.multsmiles}$. More precisely, the diffusion coefficients only coincide for the smallest maturity ($i=1$). The LVG models with piecewise constant coefficients, introduced in this subsection, are only used to obtain the cross-strike interpolation of options prices. The final calibrated non-homogeneous LVG model will have piecewise continuous, but not necessarily piecewise constant, diffusion coefficients. 
\end{remark}

Figures $\ref{fig:1}$--$\ref{fig:3}$ show the results of cross-maturity interpolation of the market prices of European options written on the S\&P 500 index, on January 12, 2011\footnote{The option prices, as well as the dividend and interest rates, are provided by Bloomberg.}.
Figure $\ref{fig:1}$ contains the resulting price curves of call options, as functions of strike. Each curve corresponds to a different maturity: $2$, $7$, $27$, $47$, and $67$ working days, respectively. In particular, Figure $\ref{fig:1}$ demonstrates that the monotonicity of option prices across maturities is preserved by the interpolation. The quality of the fit is shown on Figures $\ref{fig:2}$--$\ref{fig:3}$, via the implied volatility curves. It is easy to see that the fit is perfect, in the sense that the implied volatility of interpolated prices always falls within the implied volatilities of the bid and ask quotes. 

It is worth mentioning that, in the given market data, the first part of Assumption $\ref{assump:ass.2}$ is not satisfied. In particular, for longer maturities, the new strikes often appear between the strikes that are available for shorter maturities. This occurs quite often, for example, because the option prices and strikes need to be adjusted for the non-zero dividend and interest rates (cf. Remark $\ref{rem:shortMat}$). In addition, the market data contains bid and ask quotes, as opposed to exact prices. To resolve these issues, before initiating the calibration algorithm, we implement the method outlined in Remark $\ref{rem:addStrikes}$. Namely, we solve a linear feasibility problem, to obtain a strictly admissible set of option prices, for the same strikes at every maturity, such that the price of each option satisfies the given bid-ask constraints, if this option was traded in the market (i.e. the trading volume was positive) on that day.

As stated in Theorem $\ref{th:calib.th.main}$, the interpolated call prices, for each maturity, are obtained via an associated piecewise constant diffusion coefficient (see, for example, equation (\ref{eq.5})). 
It is clear from the right hand side of Figure $\ref{fig:3}$ that, despite the perfect fit of market data, the associated diffusion coefficients may exhibit rather wild oscillations. The diffusion coefficients of the resulting LVG model, defined in (\ref{eq.calib.a}), inherit this irregularity, which may be undesirable if, for example, one uses the calibrated model to compute Greeks or the prices of exotic options with discontinuous payoffs. 
One way to resolve this difficulty is to approximate the diffusion coefficients with ``more regular" functions. We can, then, choose the precision of the approximation to be such that the resulting call prices fall within the bid-ask spreads. It may seem that, by doing this, we face the original problem of calibrating a jump-diffusion model to option prices, which we aimed to avoid in this paper. However, there is a major difference. In the present case, we do not need to solve the ill-posed inverse problem of matching option prices: we only need to construct a sequence of regular enough functions that approximate the already known diffusion coefficients, and the associated call prices will converge to the given market values. The reason why we managed to reduce the notoriously difficult calibration problem to a well-posed one, is that we formulate it on the right space of parameters (models), which is not too large, and, yet, is large enough to contain a solution to the calibration problem. 
To simplify the computations, here, we choose to approximate the diffusion coefficient $a$ plotted on the right hand side of Figure $\ref{fig:3}$ with a piecewise constant function that has fewer discontinuities. The approximating function is constructed by choosing a partition of $(L,M)$ and setting the value of the function on each subinterval to be such that the average of $1/a^2$ is matched. The motivation for such metric comes from the observation that, for short maturities, the asymptotic behavior of interpolated option prices is determined by the geodesic distance $\int dx/a^2(x)$. 
The right hand side of Figure $\ref{fig:4}$ shows the two diffusion coefficients, and its left hand side demonstrates the implied volatility produced by the new coefficient. It is easy to see that the new diffusion coefficient is much smoother, while the resulting implied volatility still fits within the bid-ask spread.

\begin{figure}[!htb]
\begin{center}
\includegraphics[width = 0.8\textwidth]{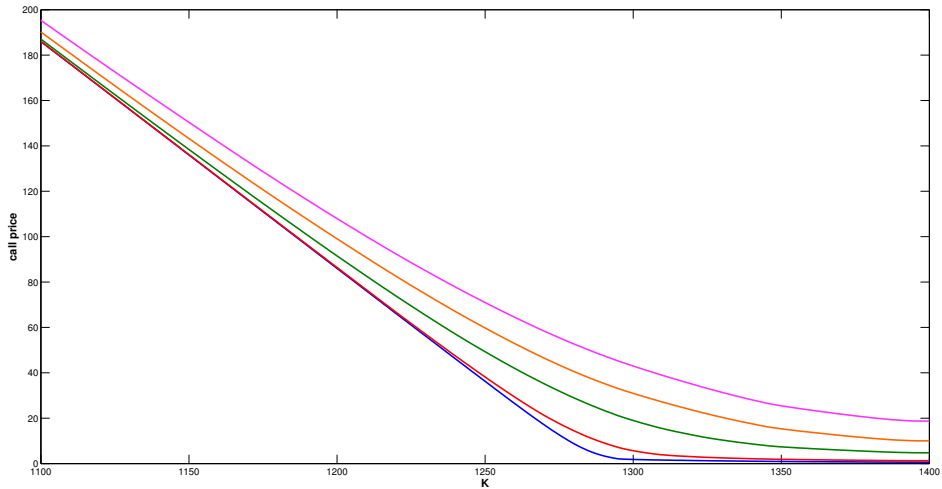}
\caption{Interpolated call prices as functions of strike. Different colors correspond to different maturities. The spot is at $1286$.}\label{fig:1}
\end{center}
\end{figure}

\begin{figure}
\begin{center}
  \begin{tabular} {cc}
    {
    \includegraphics[width = 0.48\textwidth]{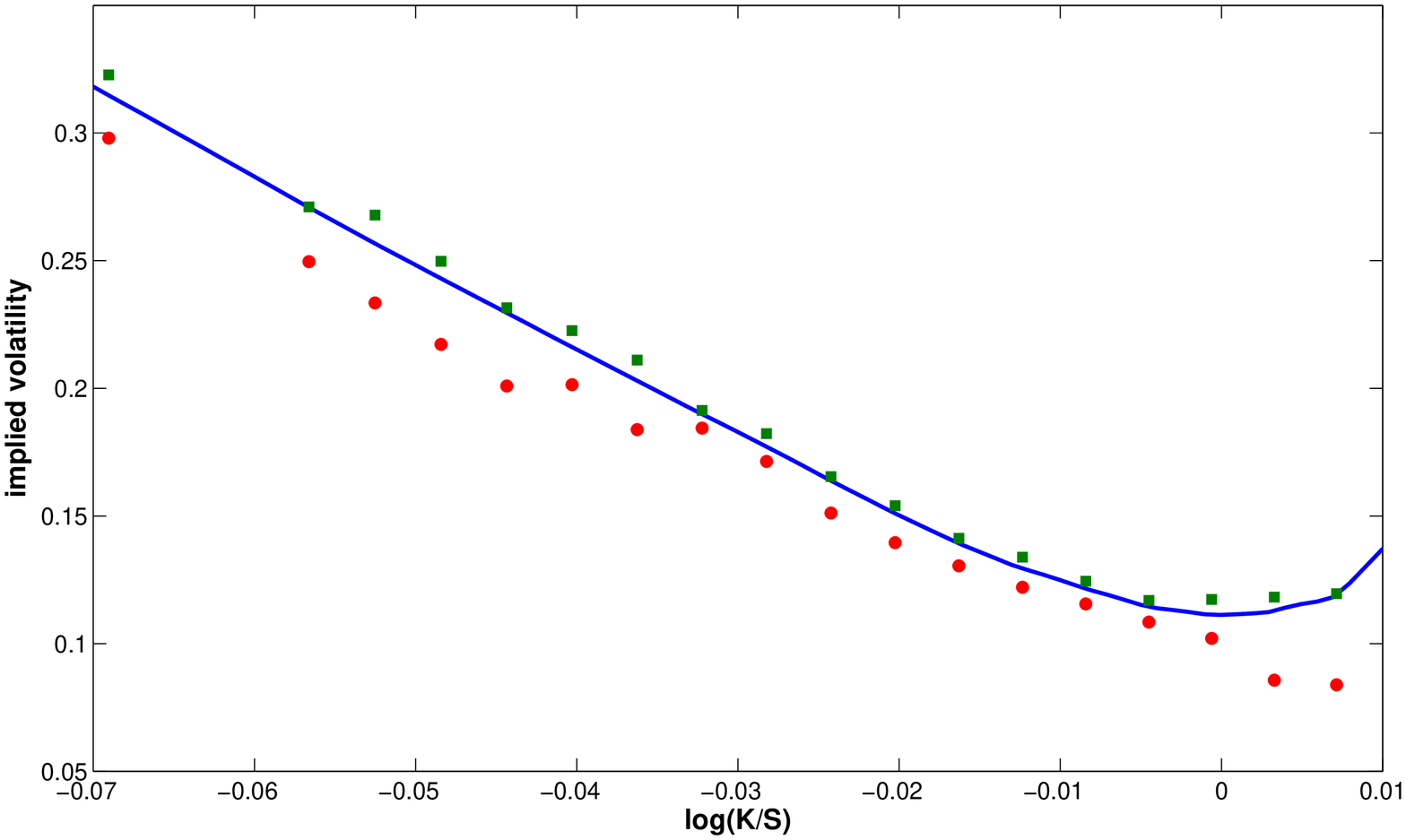}
    } & {
    \includegraphics[width = 0.48\textwidth]{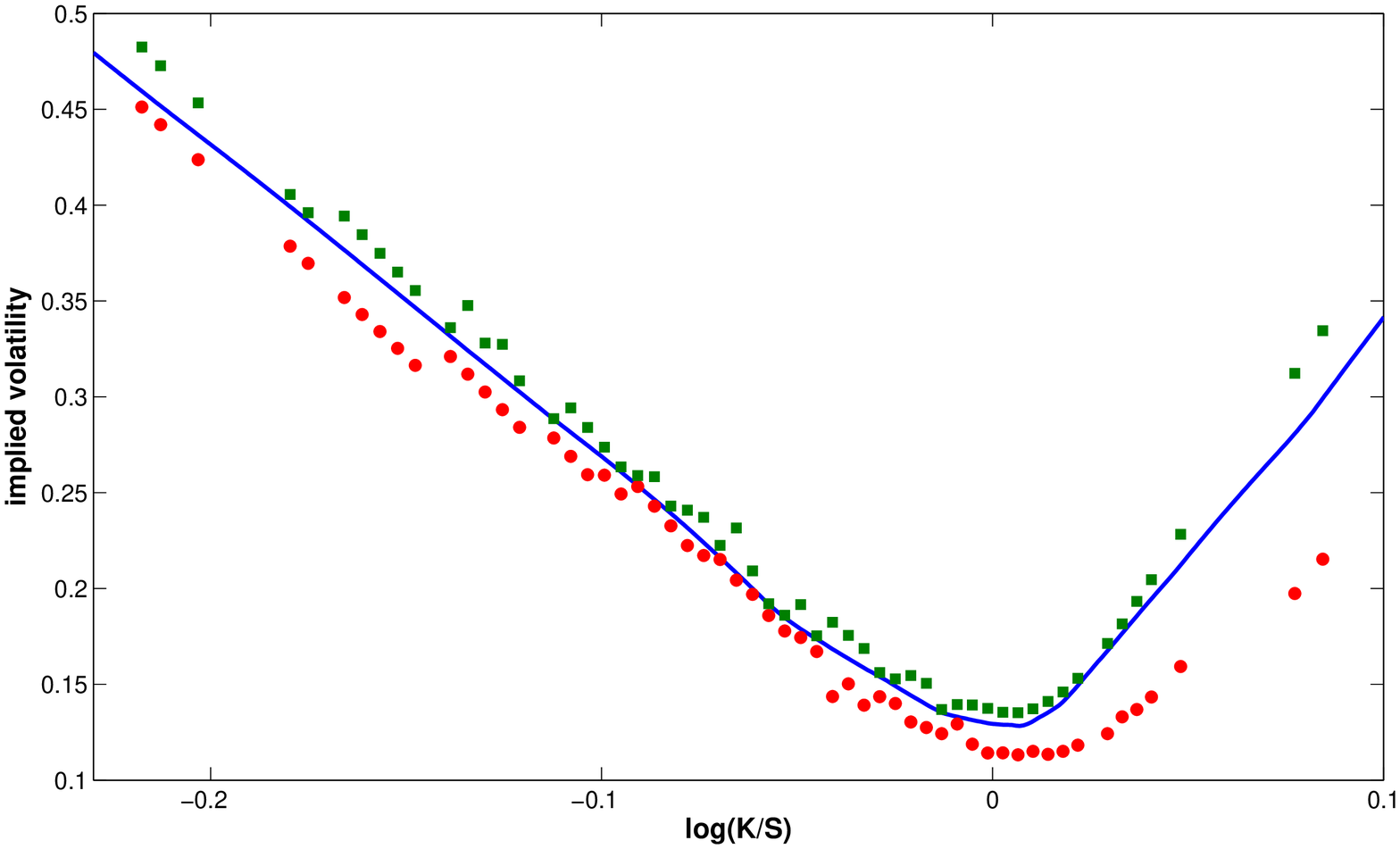}
    }\\
    {
    \includegraphics[width = 0.48\textwidth]{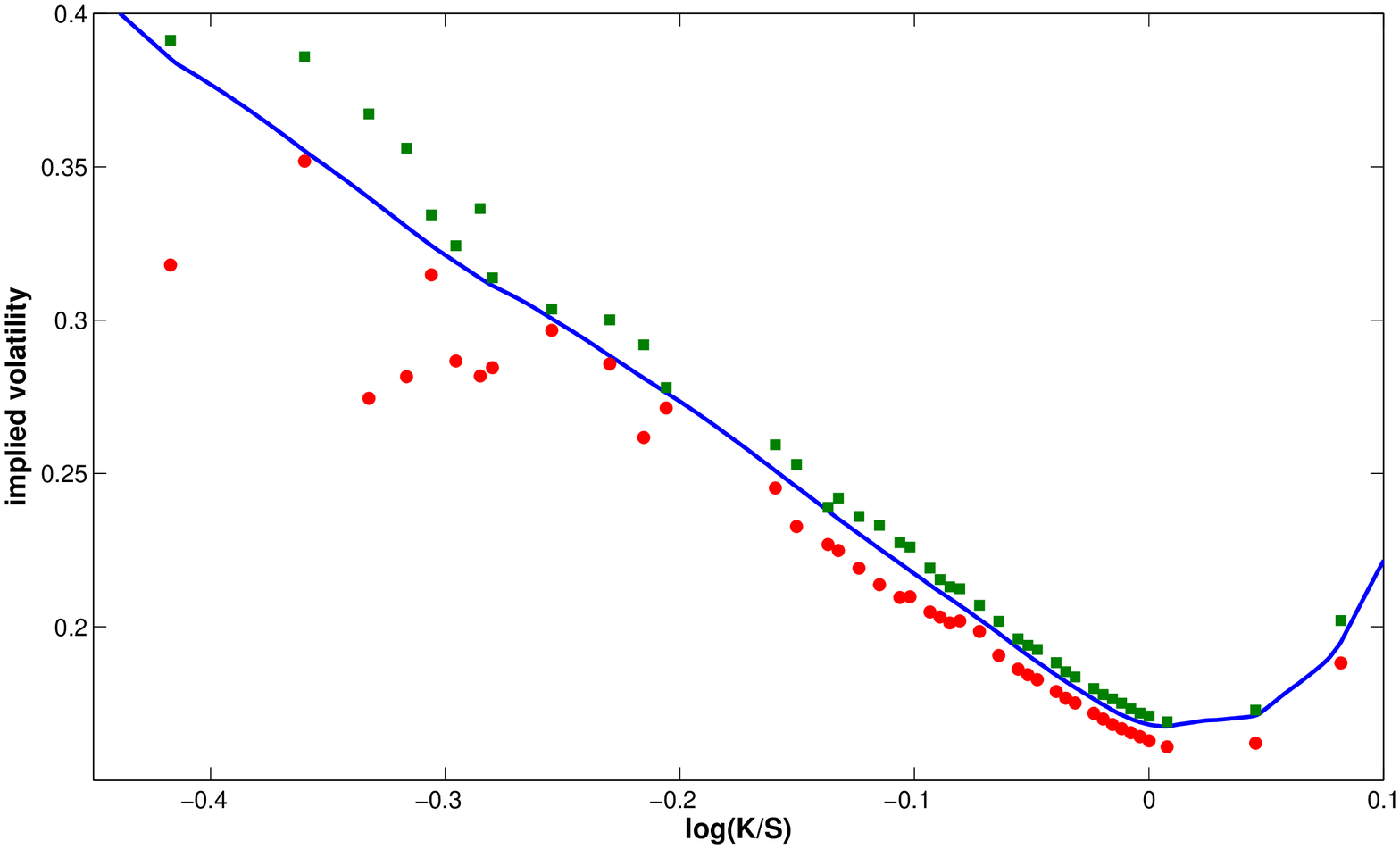}
    } & {
    \includegraphics[width = 0.48\textwidth]{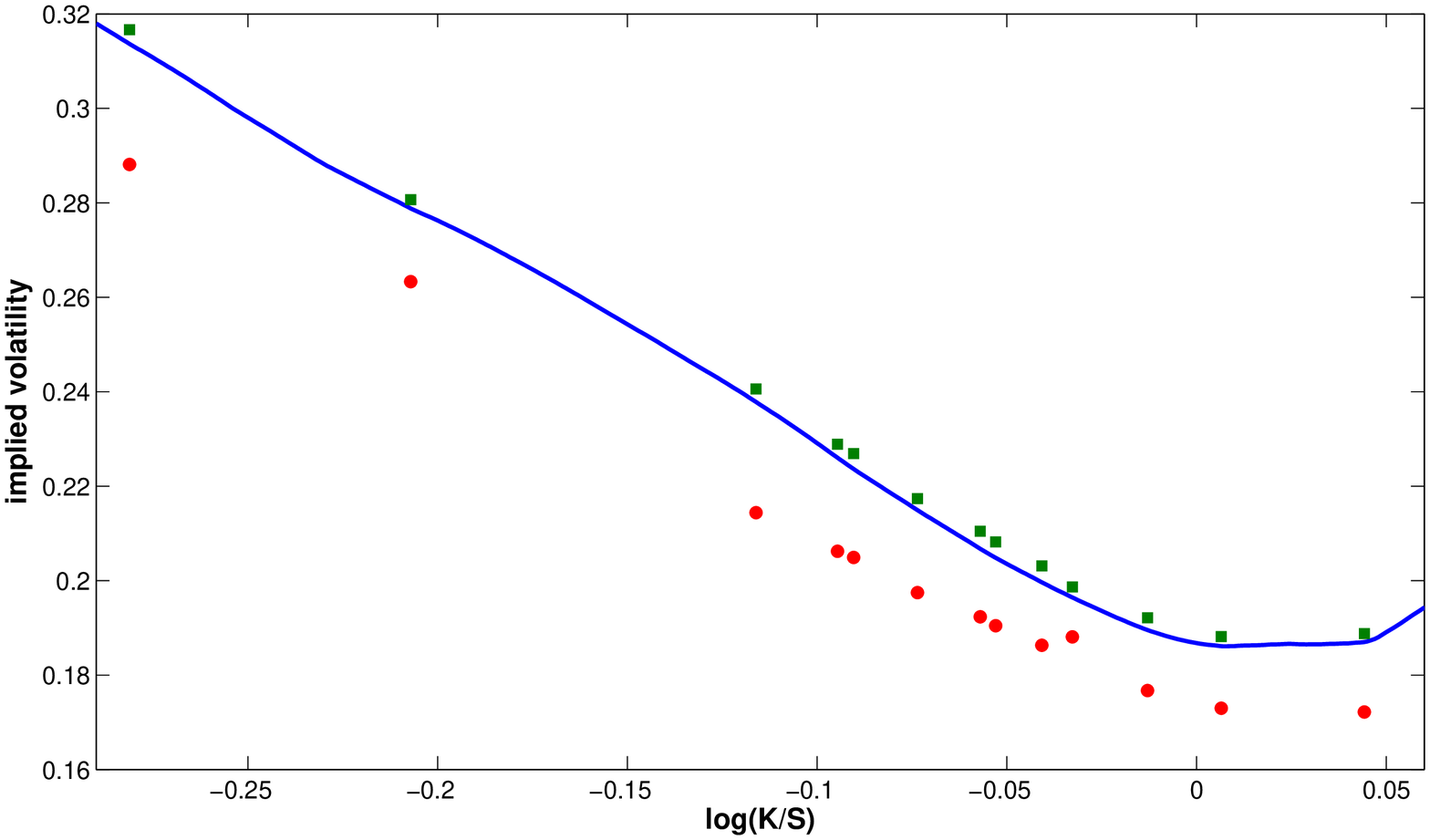}
    }
  \end{tabular}
  \caption{Implied volatility fit for the times to maturity: $2$ (top left), $7$ (top right), $47$ (bottom left), and $67$ (bottom right) working days. Blue lines represent the interpolated implied volatilities, red circles and green squares correspond to the bid and ask quotes, respectively.}
    \label{fig:2}
  \end{center}
\end{figure}

\begin{figure}
\begin{center}
  \begin{tabular} {cc}
    {
    \includegraphics[width = 0.48\textwidth]{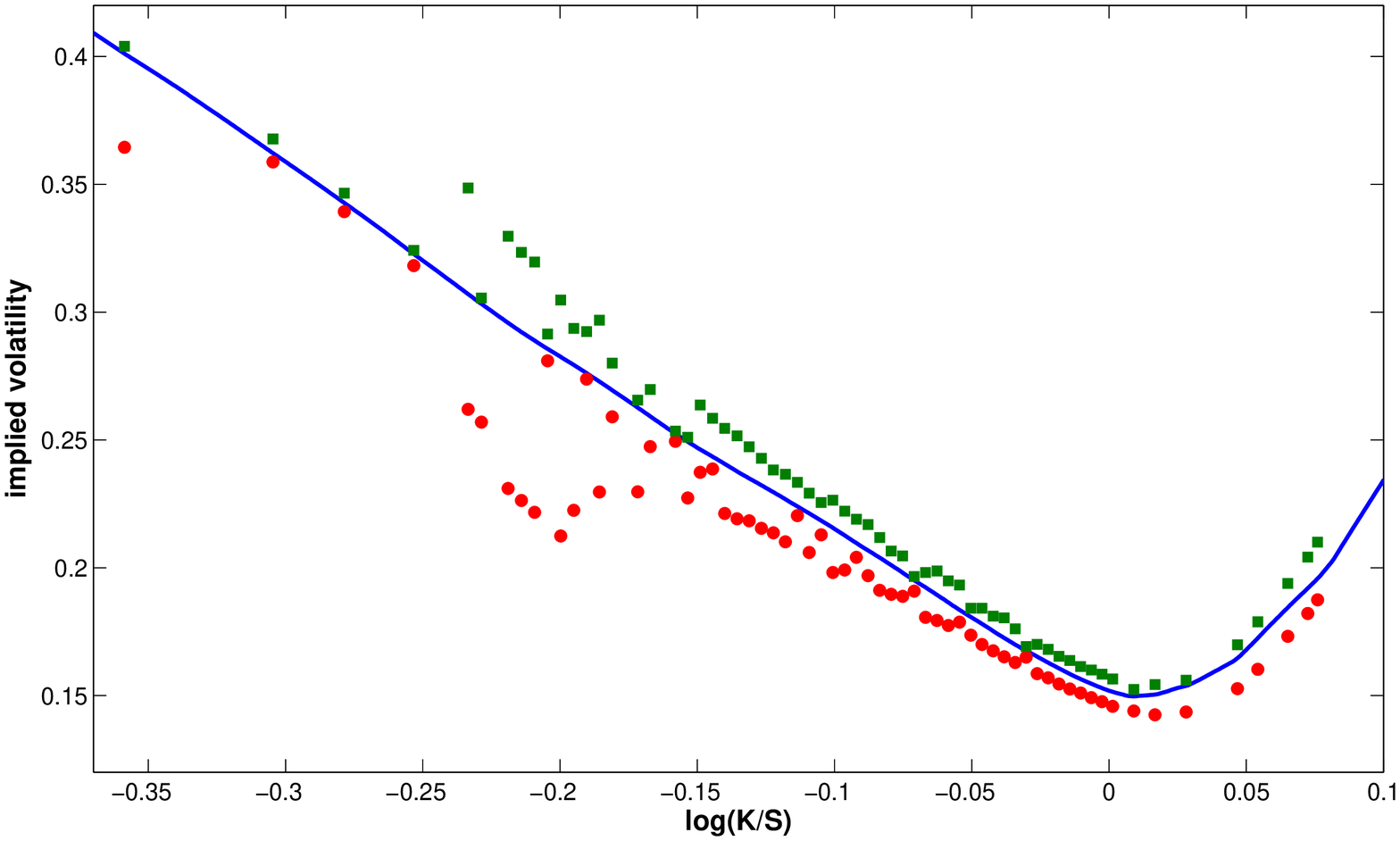}
    } & {
    \includegraphics[width = 0.48\textwidth]{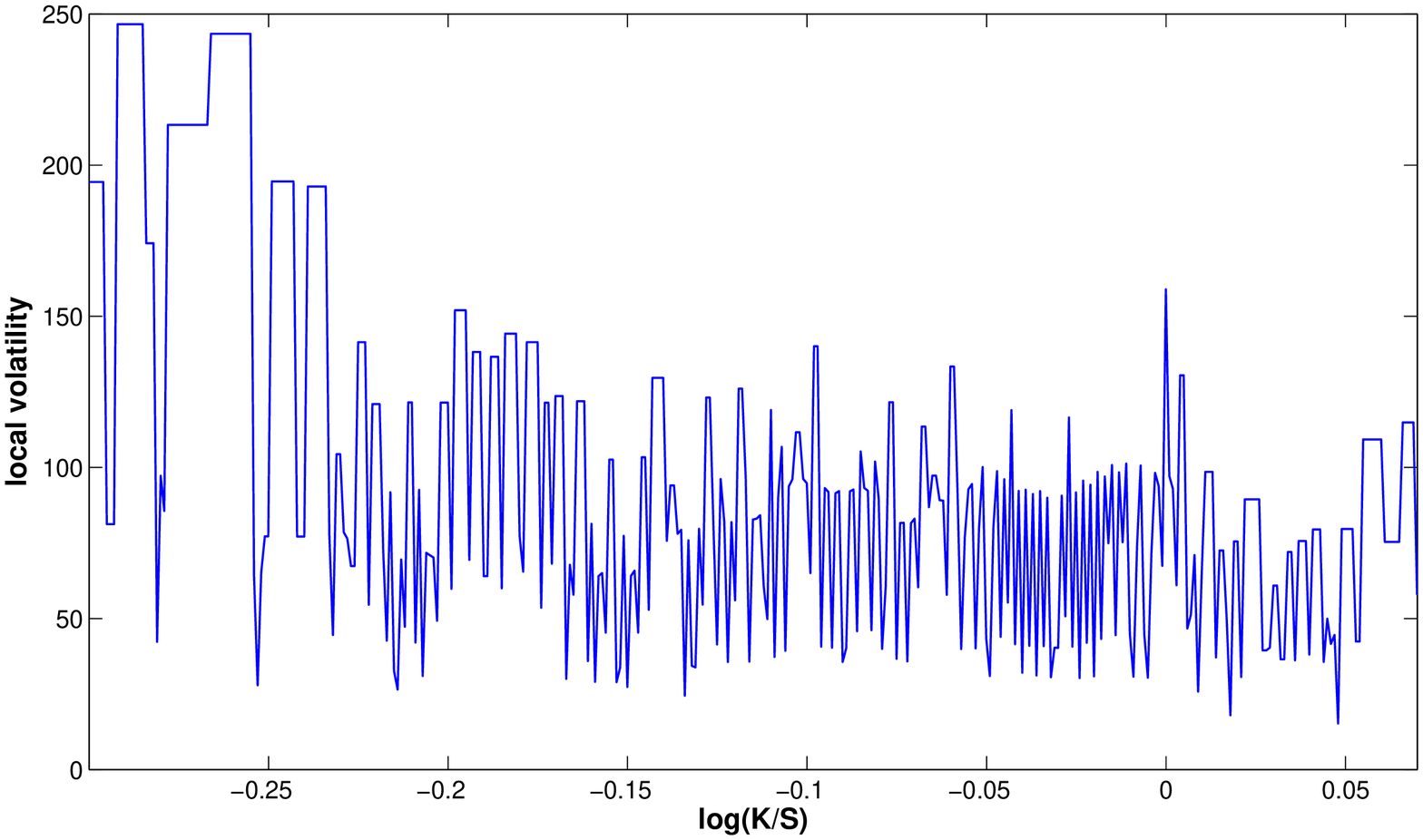}
    }\\
  \end{tabular}
  \caption{On the left: the implied volatility fit for $27$ working days to maturity. On the right: the corresponding piecewise constant diffusion coefficient.}
    \label{fig:3}
\end{center}
\end{figure}

\begin{figure}
\begin{center}
  \begin{tabular} {cc}
    {
    \includegraphics[width = 0.48\textwidth]{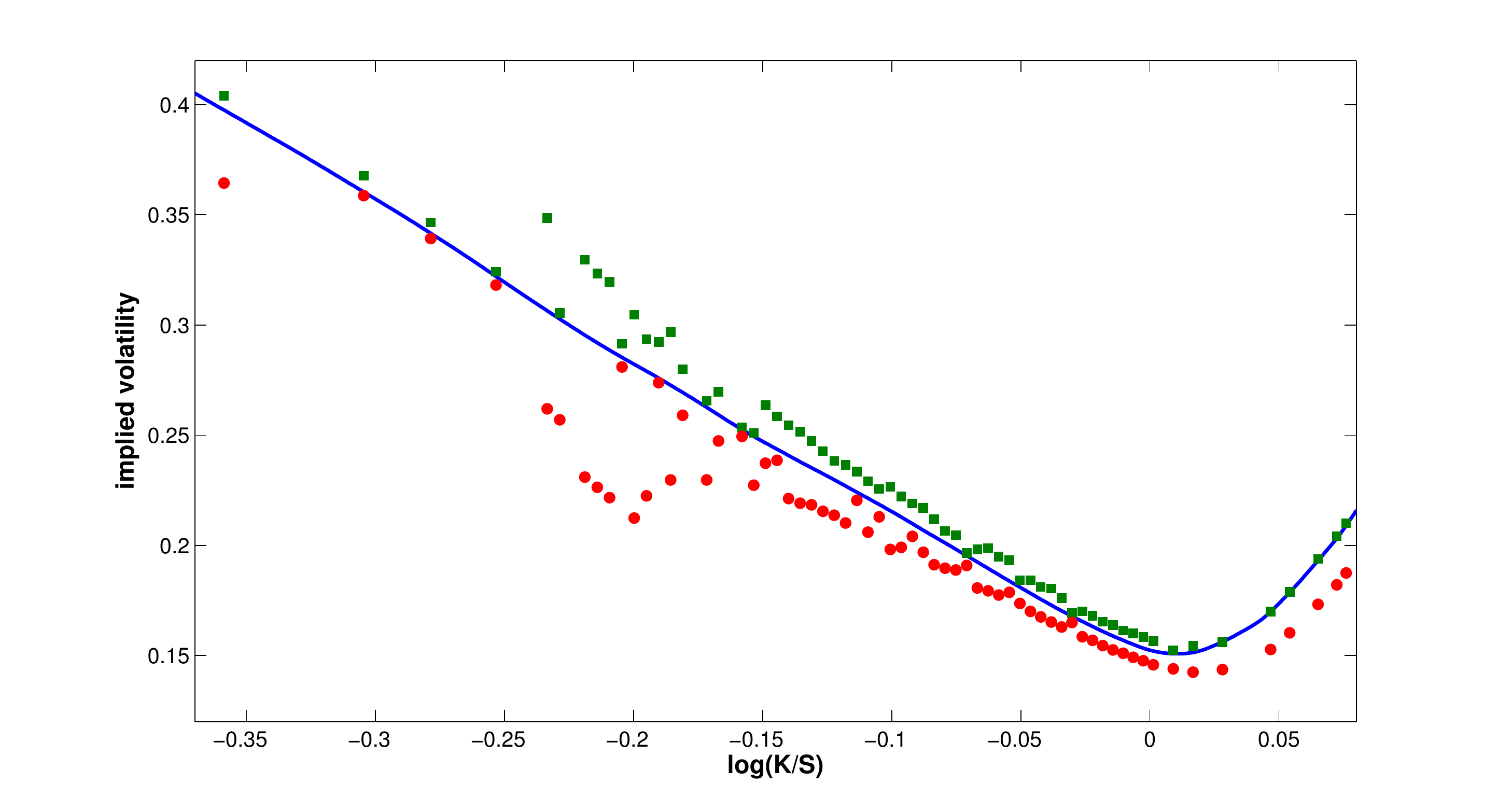}
    } & {
    \includegraphics[width = 0.48\textwidth]{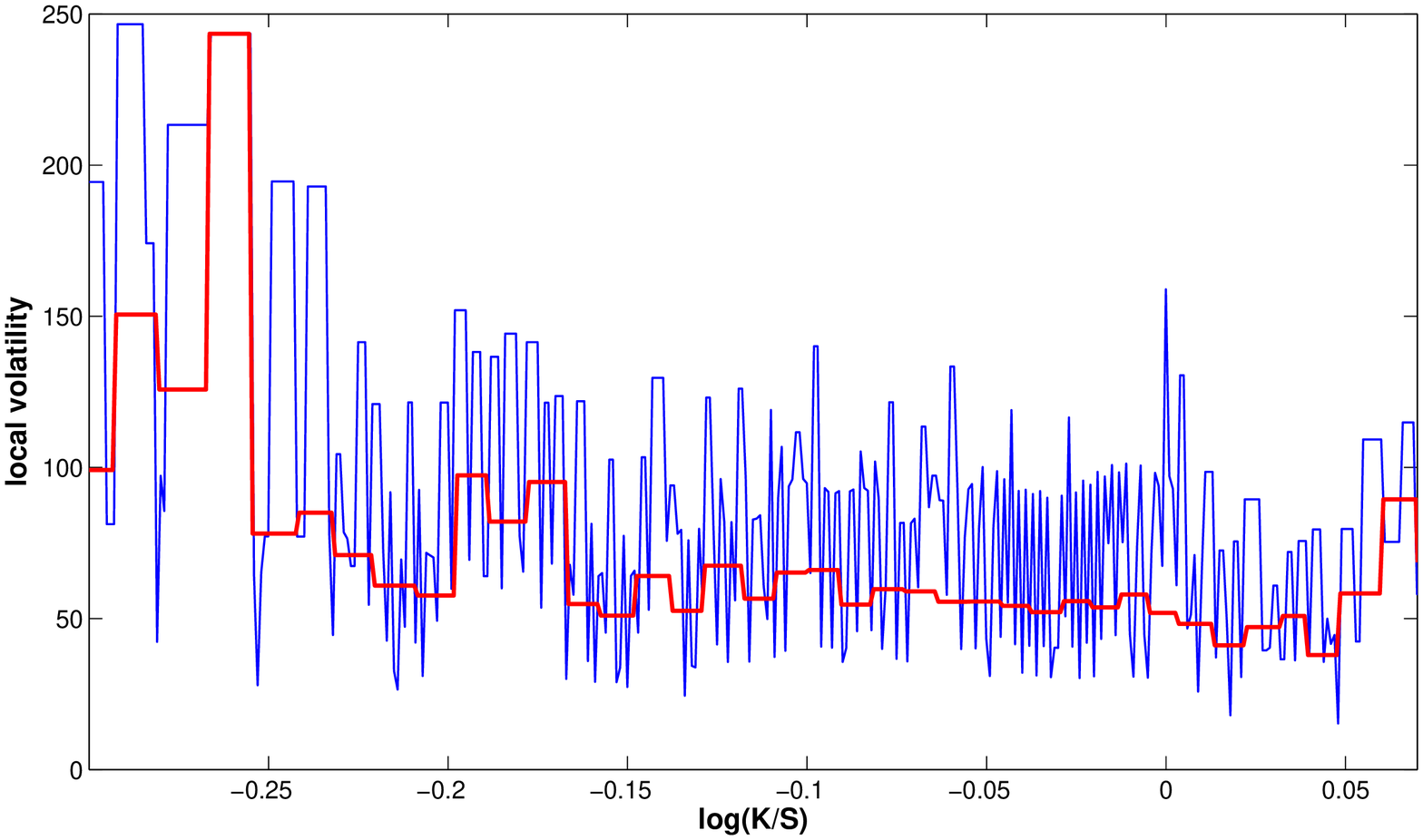}
    }\\
  \end{tabular}
  \caption{On the right: the diffusion coefficient corresponding to $27$ working days to maturity (in blue) and its approximation (in red). On the left: the implied volatility associated with the approximate diffusion coefficient (in blue), as well as the implied volatilities of bid and ask quotes (in red and green).}
    \label{fig:4}
\end{center}
\end{figure}

\begin{remark}\label{rem:shortMat}
Notice that we only use the first $5$ maturities in our numerical analysis, although there are market quotes for options with $5$ longer matures available on that day. This restriction is explained by the fact that the available bid and ask quotes may not allow for a strictly admissible set of option prices (in the sense of cf. Definition $\ref{def:strict.admiss}$). This, however, does not necessarily generate arbitrage, due to the following two observations. First, the option quotes for different strikes and maturities, provided in the database, are not always recorded simultaneously. Second, the interest rate is not identically zero, and the stocks included in the S$\&$P 500 index do pay dividends. In order to address the second issue, we had to assume that the dividend and interest rates are deterministic. Then, we discounted the option prices and strikes accordingly, to obtain expectations of the payoff functions applied to a martingale. The assumption of deterministic rates, of course, may not always be consistent with the pricing rule chosen by the market. However, this problem (as well as the first issue outlined above) goes beyond the scope of the present paper.
\end{remark}

It is important to mention that the method of cross-strike interpolation, presented in this subsection, is valuable on its own, not only in the context of LVG models. Indeed, having constructed the price curves, which have the $C^1$ and piecewise $C^2$ properties, one can interpolate them across maturities via a non-homogeneous LVG model. However, as discussed in Remark $\ref{rem:rem.AltMatInterp}$, by considering driftless diffusions run on different stochastic clocks, one can easily find other models that allow for a cross-maturity interpolation in a very similar way.
Another way to interpolate option prices across maturities is to define
$$
C(T,K) = \frac{e^{T_{i}}-e^{T}}{e^{T_i}-e^{T_{i-1}}} \bar{C}^{i-1}(K)
+ \frac{e^{T}-e^{T_{i-1}}}{e^{T_i}-e^{T_{i-1}}} \bar{C}^i(K),
$$
for $T\in[T_{i-1},T_i)$.
Then, at least formally, one can define a non-homogeneous local volatility model that reproduces the above option price surface. The corresponding diffusion coefficient $\tilde{a}$ is given by the Dupire's formula:
$$
\tilde{a}^2(T,K) = \frac{2 \partial_T C(T,K)}{\partial^2_{KK} C(T,K)},\,\,\,\,\,\,\,\,K\in(L_i,U_i),
$$
for all $T\in[T_{i-1},T_i)$, with $i=1,\ldots,M$. Note that the parabolic PDE's associated with the above local volatility (i.e. the Dupire's and Black-Scholes equations) are well posed, as follows, for example, from the results of \cite{KrylovVMO}. However, defining the associated diffusion process may still be a challenging problem, due to the discontinuities of $\tilde{a}$. 
Of course, using the results of this section, one can also find a non-homogeneous LVG model, which has ``more regular" diffusion coefficients and which approximates option prices up to the bid-ask spreads. This is done by approximating the pice-wise constant diffusion coefficient, which results from the calibration algorithm, as stated in Theorem $\ref{th:calib.th.main}$, with functions that possess the desired regularity. Such an approximation is discussed in the paragraph preceding Remark $\ref{rem:shortMat}$.

The rest of this subsection is devoted to the proof of Theorem $\ref{th:calib.th.main}$. In fact, this proof provides a detailed algorithm for computing the parameters $\left\{ \nu^i_j,\sigma^i_j \right\}$ that reproduce market call prices.

\noindent{\bf Structure of the proof}. The proof is presented in the form of an algorithm, which facilitates the implementation. However, it is rather technical and uses a lot of new notation. Therefore, here, we outline the structure and the main ideas of the proof. It is easy to see that any set of call price curves (as functions of strikes) produced by a family of homogeneous LVG models (different model for each curve), with piecewise constant diffusion coefficients, satisfies conditions $1-2$ of Assumption $\ref{assump:ass.1}$. 
However, we still need to ensure that the LVG models are chosen so that the resulting price curves also satisfy part $3$ of Assumption $\ref{assump:ass.1}$, which is a stronger version of the absence of calendar spread arbitrage. 
Thus, for each maturity $T_i$, we need to construct a pice-wise constant diffusion coefficient $a_i$ and the associated time value function $V^i$ (which is uniquely defined by (\ref{eq.calib.TV.2}) and the subsequent paragraph), such that: each time value curve $V^i$ matches the time values observed in the market, and the absence of calendar spread arbitrage is preserved ($V^i>V^{i-1}$, for all $i$).
In order to match the observed market values, each $V^i$ is constructed recursively, passing from $K_j$ to $K_{j+1}$. This allows us to construct $a_i$ and $V^i$ locally, so that $V^i$ solves (\ref{eq.calib.TV.2}) on $(L_i,K_{j+1})$, with $a_i$ in lieu of $a$. However, in order to obtain a bona fide time value function, we have to ensure that $V^i$ satisfies the zero boundary conditions at $L_i$ and $U_i$, and has a jump of size $-1$ at $x$. These conditions force us to control the value of the left derivative of $V^i$ at each strike $K_{j}$, in addition to the value of $V^i(K_j)$ itself (which must coincide with the market value). In order to satisfy the constraints on derivatives, along with the positivity of time value, we introduce two additional partition points between every $K_j$ and $K_{j+1}$. The choice of the additional partition points, as well as the proof that they are sufficient to fulfill the aforementioned conditions, takes most of the proof of Theorem $\ref{th:calib.th.main}$ (Steps 1.1-1.3). 
The proof, itself, has an inductive nature: with the induction performed over maturities, and, for every fixed maturity, over strikes. We show that the results of each iteration possess the necessary properties to serve as the initial condition for the next iteration.

\noindent{\bf Step 0}. 
Let us denote the market time values by $\bar{V}$: 
$$
\bar{V}(T_i,K^i_{j}) = \bar{C}(T_i,K^i_{j}) - (x-K^i_j)^+,
$$ 
for $i=1,\ldots,M$, $j=1,\ldots,N_i$.
It is clear that matching the market call prices is equivalent to matching their time values.
To simplify the notation, we add the strikes $K^i_0=L_i$ and $K^i_{N_i+1}=U_i$, for all $i=1,\ldots,M$, to the set of available market strikes, with the corresponding market time values being zero (since, in the calibrated model, the underlying cannot leave the interval $[L_i,U_i]$ by the time $T_i$). We will construct the interpolated price curves that match the original market prices, along with these additional ones.

Our construction will be recursive in $i$, starting from $i=1$.
Assume that we have constructed the interpolated time values for each maturity $T_m$,
$$
V^{m}(K) = V^{\nu^{m},\sigma^{m},z,x}(K),\,\,\,\,\,\,K\in\RR,
$$
with $m=0,\ldots,i-1$. For $m=0$, we set $V^0\equiv0$. In addition to matching the market time values, we assume that these functions satisfy:
$$
V^{m}(K) > V^{m-1}(K),\,\,\,\,\,K\in(L_m,U_m),
$$
for $m=1,\ldots,i-1$, and $V^{i-1}$ is strictly smaller than the market time values for maturities $T_i$ and larger. 
Our goal now is to construct $V^i=V^{\nu^{i},\sigma^{i},z,x}$, such that the extended family $\left\{V^m\right\}_{m=1}^i$ still satisfies the above monotonicity properties.

Without loss of generality, we can assume that the underlying level $x$ coincides with one of the strikes. If $x$ is not among the available strikes, we can always add a new market call price, for strike $x$ and maturities $T_1,\ldots,T_i$, as follows: if $x\in(K^i_j,K^i_{j+1})$, then, we choose an arbitrary $\delta_1\in(0,1)$ and introduce the additional market time values
$$
\bar{V}(T_m,x) = V^m(x),\,\,\,\,m=1,\ldots,i-1,
$$
$$
\bar{V}(T_i,x) = \delta_1 \left( \bar{C}\left(T_i,K^i_{j}\right)\frac{K^i_{j+1} - x}{K^i_{j+1} - K^i_j} 
+ \bar{C}\left(T_i,K^i_{j+1}\right) \frac{x - K^i_{j}}{K^i_{j+1} - K^i_j} \right) 
+ (1-\delta_1) \max\left(V^{i-1}(x),  \bar{C}\left(T_i,K^i_{j+1}\right)\right)
$$
It is easy to see that the new family of market call prices, for $m=i,\ldots,M$, is strictly admissible, and the constructed time value functions, $V^1,\ldots, V^{i-1}$, satisfy the properties discussed in the previous paragraph, with respect to the new data.

The construction of $V^i(K)$ is done recursively in $K$, splitting the real line into several segments.
We set $V^i(K)= 0$, for all $K\in(-\infty,L_i]$. We, then, extend $V^i$ to each $[L_i,K^i_j]$, increasing $j$ by one, as long as $K^i_{j+1}<x$.

\noindent{\bf Step 1}. Assume that $K^i_{j+1}<x$ and that, for all $K\in[L_i,K^i_{j}]$, we have constructed the time value function $V^i(K)$ in the form (\ref{eq.2}), with some $\left\{L_i=\nu_0<\cdots<\nu_{l(j)}=K^i_j\right\}$ and $\left\{\sigma_1>0,\ldots,\sigma_{l(j)}>0 \right\}$. In addition, we assume that $V^i(K)>V^{i-1}(K)$, for all $K\in[L_i,K^i_j]$, and the left derivative of $V^i(K)$ at $K=K^i_j$, denoted $B$, satisfies:
\begin{equation}\label{eq.calib.B.ineq}
B < \frac{\bar{V}(T_i,K^i_{j+1}) - \bar{V}(T_i,K^i_j)}{K^i_{j+1}-K^i_j}
\end{equation}
We need to extend $V^i$ to $[L_i,K^i_{j+1}]$ in such a way that: it remains in the form (\ref{eq.2}), on this interval, it matches the market time value at $K^i_{j+1}$, and its left derivative at $K^i_{j+1}$ satisfies the above inequality. 

Note that the case $j=0$ is not included in the above discussion. In this case, the left derivative of $V^i$ at $K^i_0=L_i$ is zero. However, the derivative of $V^i$ does not have to be continuous at this point, hence, we can change its value to a positive number. We choose this number as follows: fix an arbitrary $\delta_2\in(0,1)$ and define
\begin{equation}\label{eq.calib.initderiv}
B = \delta_2 \frac{\bar{V}\left(T_i, K^i_{1}\right)}{K^i_{1} - L_i}  
+ (1-\delta_2) {V^{i-1}_+}\left(L_i\right),
\end{equation} 
where ${V^{i-1}_+}$ is the right derivative of $V^{i-1}$.
Due to the strict admissibility of market data, as well as the convexity of $V^{i-1}(K)$ for $K\in[L_i,x]$, and the fact that $V^{i-1}(K^i_1)<\bar{V}(T_i,K^i_1)$, we easily deduce that $B$ given by (\ref{eq.calib.initderiv}) satisfies (\ref{eq.calib.B.ineq}), with $j=0$, and, in addition, $B>V^{i-1}_+(L_i)$.


\noindent {\bf Step 1.1}. Choose an arbitrary $\delta_3\in(0,1)$, and define
$$
B_1 = \delta_3 \frac{\bar{V}(T_i,K^i_{j+2}) - \bar{V}(T_i,K^i_{j+1})}{K^i_{j+2}-K^i_{j+1}}
+ (1-\delta_3) \frac{\bar{V}(T_i,K^i_{j+1}) - \bar{V}(T_i,K^i_{j})}{K^i_{j+1}-K^i_{j}}
$$
This will be the target derivative of extended time value function $V^i(K)$ at $K=K^i_{j+1}$.
To match this value, we need to introduce an additional jump point for the associated diffusion coefficients.

\noindent {\bf Step 1.2}. Introduce
$$
w=\frac{\bar{V}(T_i,K^i_{j+1}) + B K^i_{j} - A - B_1 K^i_{j+1}}{B-B_1}
$$
where $A>0$ is the left limit of $V^i(K)$ at $K=K^i_j$.
Due to the definition of $B_1$ and the inequality (\ref{eq.calib.B.ineq}), we have $w \in (K^i_{j},K^i_{j+1})$. In other words, $w$ is the abscissa of the intersection point of two lines: one that goes through $(K^i_j,A)$, with slope $B$, and the other one, that goes through $(K^i_{j+1},\bar{V}(T_i,K^i_{j+1}))$, with slope $B_1$ (cf. Figure $\ref{fig:5}$).

\begin{figure}
\begin{center}
  \begin{tabular} {cc}
    {
    \includegraphics[width = 0.48\textwidth]{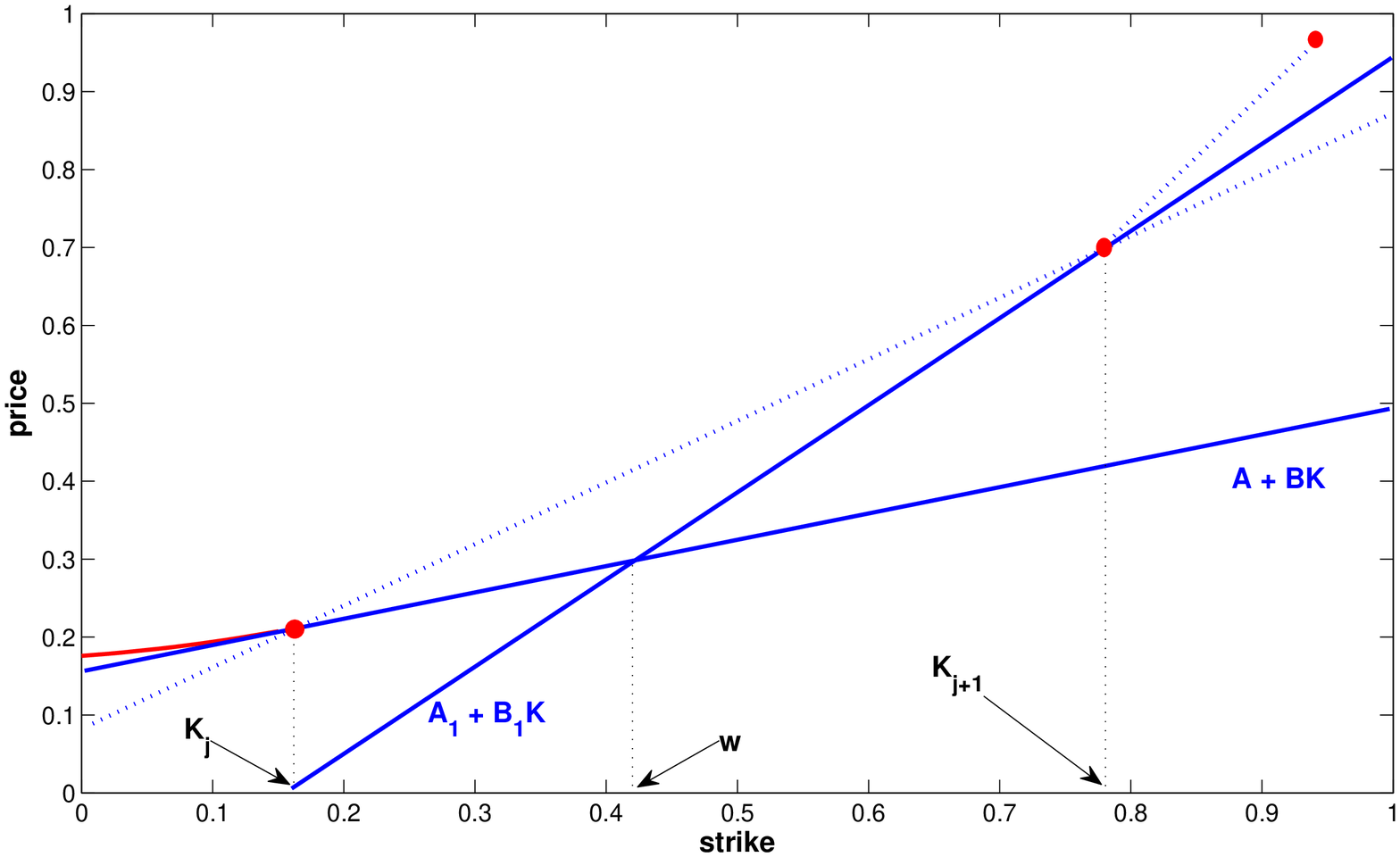}
    } & {
    \includegraphics[width = 0.48\textwidth]{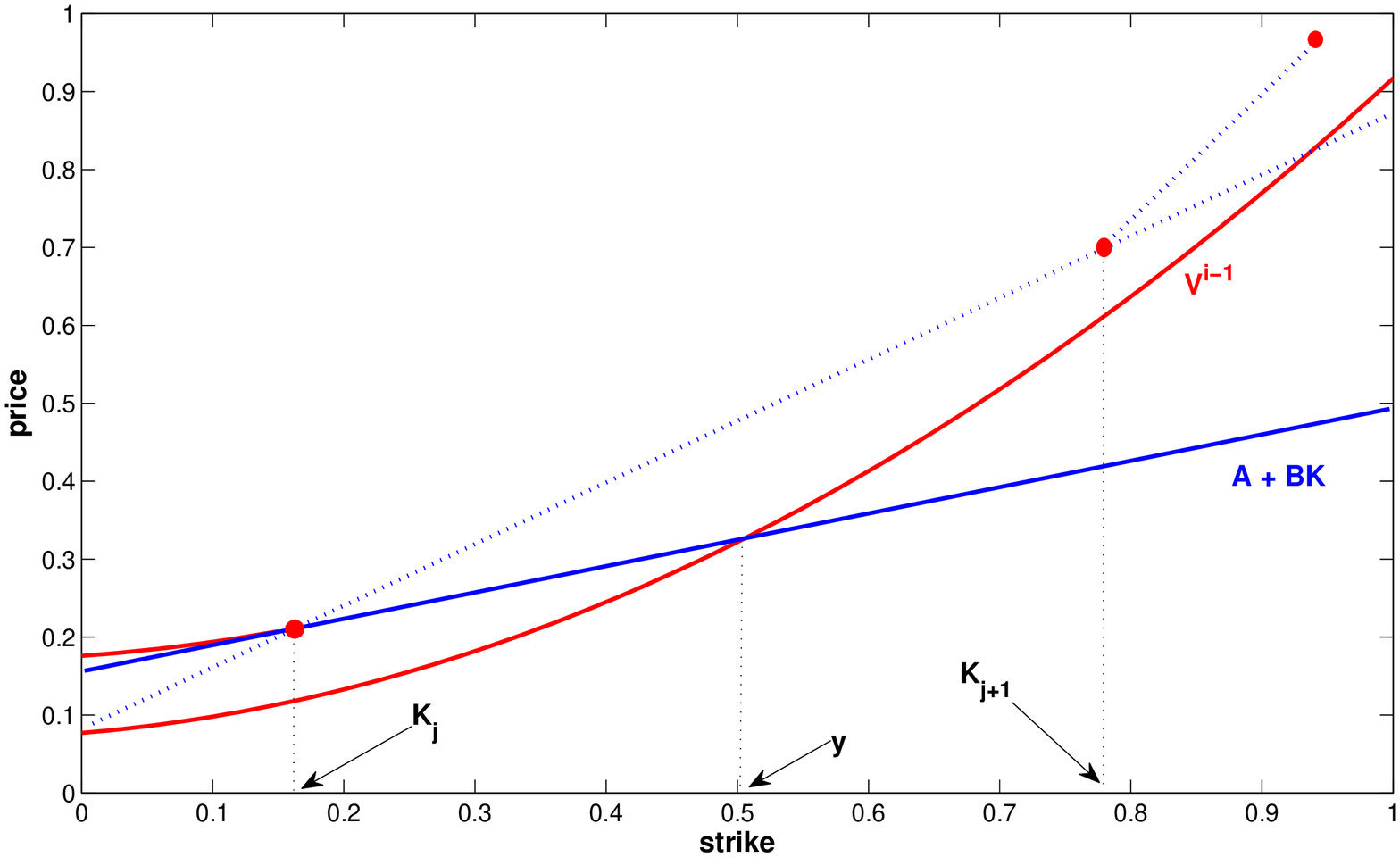}
    }\\
  \end{tabular}
  \caption{Visualization of the points $w$ (on the left) and $y$ (on the right). Red circles represent market time values for the current maturity. Red lines represent the already constructed time value interpolations.}
    \label{fig:5}
\end{center}
\end{figure}

\noindent {\bf Step 1.3}. Another potential problem is that the extended function $V^i$ needs to be strictly larger than $V^{i-1}$ on $[L_i,K^i_{j+1}]$. The two functions can intersect only if $A + (K^i_{j+1}-K^i_j)B < V^{i-1}(K^i_{j+1})$. If the latter inequality is satisfied, we introduce $y\in(K^i_j,K^i_{j+1})$ as the abscissa of the point at which the linear function intersects $V^{i-1}$:
$$
A + (y-K^i_j)B = V^{i-1}(y)
$$
Since $V^{i-1}$ is strictly convex and, either $A>V^{i-1}(K^i_j)$, or $A>V^{i-1}(K^i_j)$ and $B>V^{i-1}_+(K^i_j)$, the solution to the above equation exists and is unique in the interval $(K^i_j,K^i_{j+1})$. It can be computed, for example, via the bisection method.
If $A + (K^i_{j+1}-K^i_j)B \geq V^{i-1}(K^i_{j+1})$, we set $y=K^i_{j+1}$.

Next, we need to consider two cases.

\noindent {\bf Step 1.3.a}. If $w\leq y$, we search for the desired time value function $V^i(K)$, for $K\in[K^i_{j},w]$, in the form
\begin{equation}\label{eq.calib.mainth.eq1}
V^{i}(K) = f^1_{l(j)+1} e^{-zK/\sigma} + f^2_{l(j)+1} e^{zK/\sigma},
\end{equation}
where $\sigma$ is an unknown variable.
In order to guarantee that the interpolated time value function remains in the form (\ref{eq.2}), the coefficients $f^1_{l(j)+1}$ and $f^2_{l(j)+1}$ must satisfy:
$$
f^2_{l(j)+1} e^{zK^i_j/\sigma} = \frac{1}{2}\left(A + \frac{\sigma}{z} B\right),\,\,\,\,\,\,\,
f^1_{l(j)+1} e^{-zK^i_j/\sigma} = \frac{1}{2}\left(A - \frac{\sigma}{z} B\right),
$$
Since the extended time value function, $V^i$, given above, is a convex function with derivative $B$ at $K=K_j$, and because of the assumption $w\leq y$, it is easy to see that, for any choice of $\sigma>0$, $V^i(K)$ dominates the time value of earlier maturity, $V^{i-1}(K)$, from above, on the interval $K\in[L_i,w]$ (cf. Figure $\ref{fig:6}$). 

Next, we show that there exists a unique value of $\sigma$, such that the time value function (\ref{eq.calib.mainth.eq1}) allows for a further extension to $[L_i,K^i_{j+1}]$, such that it matches the market time value at $K^i_{j+1}$ and remains in the form (\ref{eq.2}). 
Notice that the value of $V^i(w)$ is given by $\tilde{A}$, and its left derivative is $\tilde{B}$, where:
\begin{equation}\label{eq.1.3.a.tildeA}
\tilde{A}= \tilde{A}(\sigma) = \frac{1}{2}\left(A - \frac{\sigma}{z} B\right)  e^{-z(w-K^i_j)/\sigma}
+ \frac{1}{2}\left(A + \frac{\sigma}{z} B\right) e^{z(w-K^i_j)/\sigma},
\end{equation}
\begin{equation}\label{eq.1.3.a.tildeB}
\tilde{B} = \tilde{B}(\sigma) = -\frac{z}{2\sigma}\left(A - \frac{\sigma}{z} B\right)  e^{-z(w-K^i_j)/\sigma}
+ \frac{z}{2\sigma}\left(A + \frac{\sigma}{z} B\right) e^{z(w-K^i_j)/\sigma}
\end{equation}
 
\begin{lemma}\label{le:temp.0}
Functions $\tilde{A}(\sigma)$ and $\tilde{B}(\sigma)$, defined in (\ref{eq.1.3.a.tildeA})-(\ref{eq.1.3.a.tildeB}), are strictly decreasing in $\sigma>0$. Moreover, the range of values of $\tilde{A}$ is given by
$$
\left(A + B (w-K^i_j),\infty\right),
$$
and the range of values of $\tilde{B}$ is $(B,\infty)$.
\end{lemma}

The proof of Lemma $\ref{le:temp.0}$ is given in Appendix B.
Recall that, for any $\sigma>0$, the time value $V^i$ can be extended to $[L_i,w]$ via (\ref{eq.calib.mainth.eq1}). Thus, we need to find $\sigma$, such that this time value function $V^i$ can be extended further, to $[L_i,K^i_{j+1}]$, so that it matches the given market time value, and the target derivative, $B_1$, at $K^i_{j+1}$.
Since $V^i(K)$ remains in the form (\ref{eq.2}), on the interval $K\in[w,K^i_{j+1}]$, it has to be given by
$$
\frac{1}{2}\left(\tilde{A}(\sigma) + \frac{\tilde{\sigma}}{z} \tilde{B}(\sigma)\right) e^{z(K-w)/\tilde{\sigma}}
+ \frac{1}{2}\left(\tilde{A}(\sigma) - \frac{\tilde{\sigma}}{z} \tilde{B}(\sigma)\right) e^{-z(K-w)/\tilde{\sigma}},
$$
where $\tilde{\sigma}$ has to be such that the market time value at $K=K^i_{j+1}$ is matched.
From Lemma $\ref{le:temp.0}$, we know that, for any fixed $\sigma>0$, the function
$$
\tilde{\sigma}\mapsto
\frac{1}{2}\left(\tilde{A}(\sigma) + \frac{\tilde{\sigma}}{z} \tilde{B}(\sigma)\right) e^{z(K^i_{j+1}-w)/\tilde{\sigma}}
+ \frac{1}{2}\left(\tilde{A}(\sigma) - \frac{\tilde{\sigma}}{z} \tilde{B}(\sigma)\right) e^{-z(K^i_{j+1}-w)/\tilde{\sigma}},
$$
defined for all $\tilde{\sigma}>0$, is strictly decreasing, and its range of values is
$$
\left(\tilde{A}(\sigma) + \tilde{B}(\sigma) (K^i_{j+1}-w),\infty\right)
$$
Due to (\ref{eq.calib.B.ineq}), we have
$$
A + B (w-K^i_j) + B(K^i_{j+1}-w) < \bar{V}(T_i,K^i_{j+1})
$$
Using the above and Lemma $\ref{le:temp.0}$, we conclude that there exists a unique $\hat{\sigma}$ satisfying
$$
\tilde{A}(\hat{\sigma}) + \tilde{B}(\hat{\sigma}) (K^i_{j+1}-w) = \bar{V}(T_i,K^i_{j+1})
$$
The above observation, along with the aforementioned monotonicity, implies that, for each $\sigma>\hat{\sigma}$, there exists a unique $\tilde{\sigma}=\tilde{\sigma}(\sigma)>0$ satisfying
$$
\frac{1}{2}\left(\tilde{A}(\sigma) + \frac{\tilde{\sigma}}{z} \tilde{B}(\sigma)\right) e^{z(K^i_{j+1}-w)/\tilde{\sigma}}
+ \frac{1}{2}\left(\tilde{A}(\sigma) - \frac{\tilde{\sigma}}{z} \tilde{B}(\sigma)\right) e^{-z(K^i_{j+1}-w)/\tilde{\sigma}}
 = \bar{V}(T_i,K^i_{j+1}),
$$
which can be obtained via the bisection method.
Notice that $\tilde{\sigma}(\sigma)$ is strictly decreasing in $\sigma\in(\hat{\sigma},\infty)$. 

Thus, for each choice of $\sigma>\hat{\sigma}$, we can extend the time value function $V^i$ from $[L_i,K^i_j]$ to $[L_i,K^i_{j+1}]$, so that it matches the market time value at $K^i_{j+1}$ and remains in the form (\ref{eq.2}).
It only remains to find $\sigma$ such that the left derivative of $V^i(K)$ at $K=K^i_{j+1}$ coincides with the target value, $B_1$. This condition can be expressed as follows:
\begin{equation}\label{eq.1.3.a.last}
\frac{z}{2\tilde{\sigma}(\sigma)} \left(\tilde{A}(\sigma) + \frac{\tilde{\sigma}(\sigma)}{z} \tilde{B}(\sigma)\right) e^{z(K^i_{j+1}-w)/\tilde{\sigma}(\sigma)} 
-\frac{z}{2\tilde{\sigma}(\sigma)}\left(\tilde{A}(\sigma) - \frac{\tilde{\sigma}(\sigma)}{z} \tilde{B}(\sigma)\right)
e^{-z(K^i_{j+1}-w)/\tilde{\sigma}(\sigma)}
= B_1
\end{equation}

\begin{lemma}\label{le:temp.2}
The left hand side of (\ref{eq.1.3.a.last}) is a strictly increasing function of $\sigma\in(\hat{\sigma},\infty)$, and its range of values includes $B_1$.
\end{lemma}

The proof of Lemma $\ref{le:temp.2}$ is given in Appendix B.
Lemma $\ref{le:temp.2}$ implies that there exists a unique solution $\sigma=\bar{\sigma}$ to the equation (\ref{eq.1.3.a.last}), and it can be computed via the bisection method.
Finally, we set: $\nu_{l(j)+1}=w$, $\sigma_{l(j)+1}=\bar{\sigma}$, $\nu_{l(j)+2}=K^i_{j+1}$ and $\sigma_{l(j)+2}=\tilde{\sigma}(\bar{\sigma})$, and $l(j+1)=l(j)+2$.

\begin{figure}
\begin{center}
  \begin{tabular} {cc}
    {
    \includegraphics[width = 0.48\textwidth]{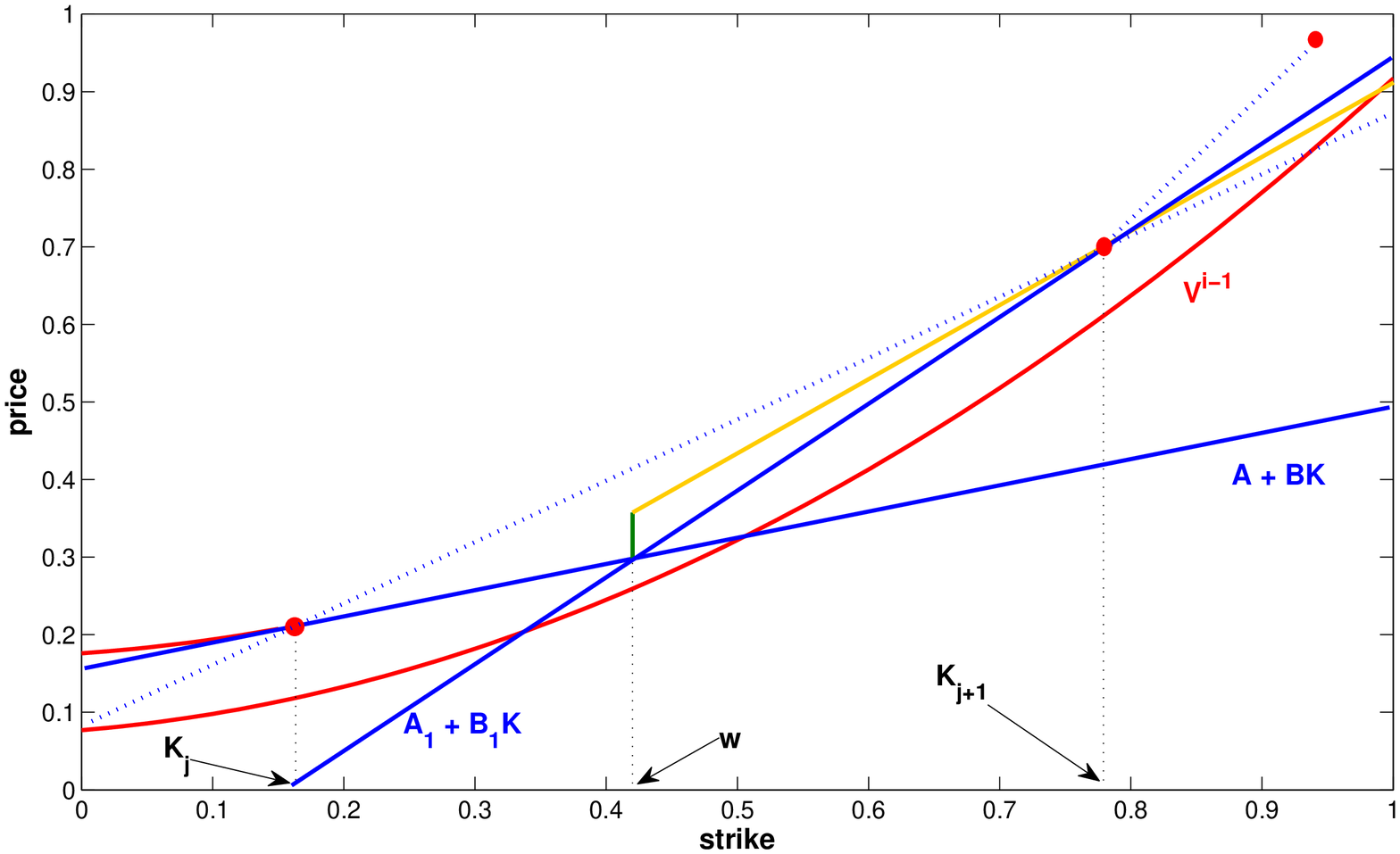}
    } & {
    \includegraphics[width = 0.48\textwidth]{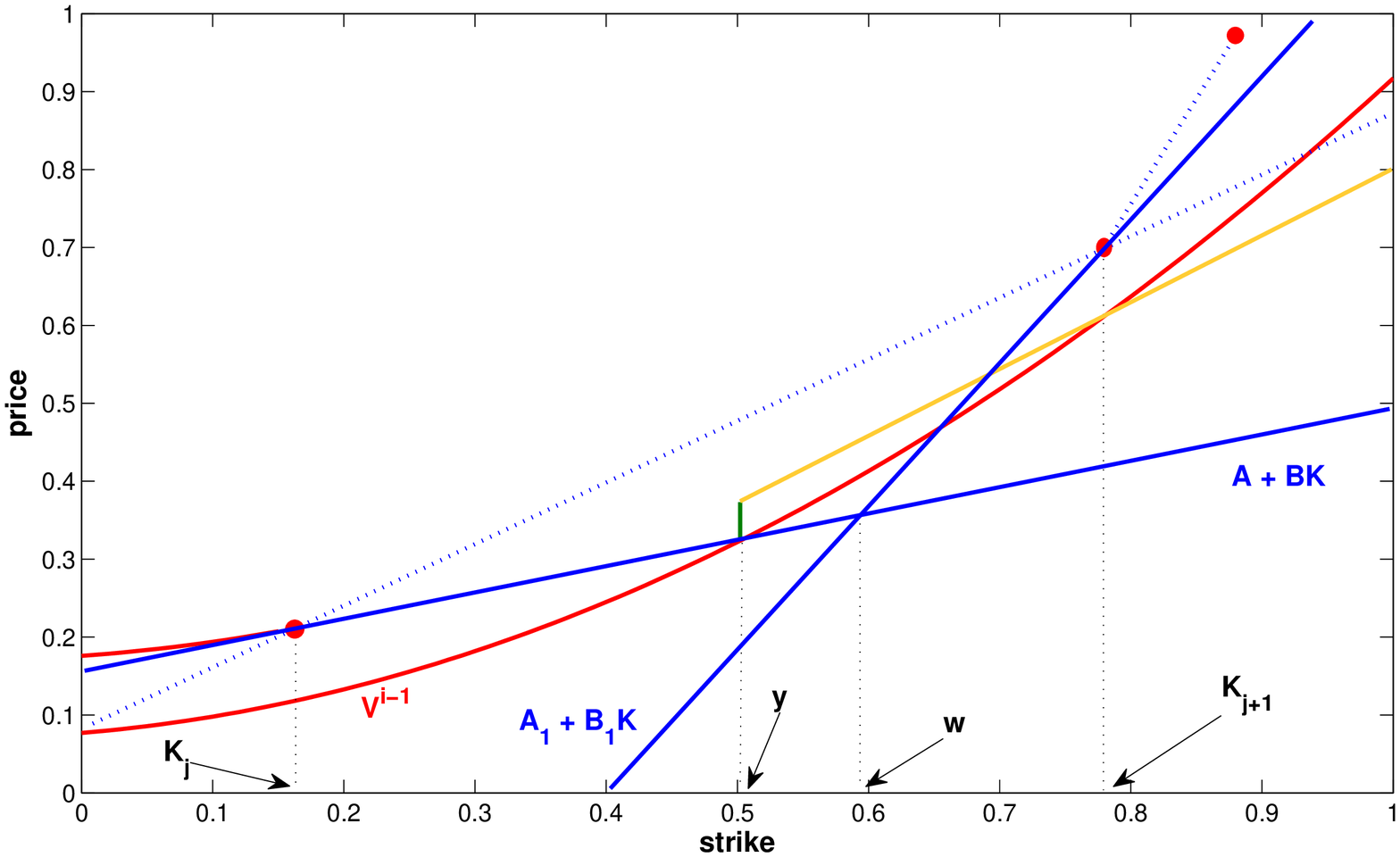}
    }\\
  \end{tabular}
  \caption{On the left: the case described in Step 1.3.a, where the yellow line represents the function $\tilde{A}(\hat{\sigma}) + \tilde{B}(\hat{\sigma}) K$, for $K\geq w$, and the vertical interval in green represents the range of values of $\tilde{A}(\sigma)$, for $\sigma\geq\hat{\sigma}$. On the right: the case described in Step 1.3.b, where the yellow line represents the function $\tilde{A}(\bar{\sigma}) + \tilde{B}(\bar{\sigma}) K$, for $K\geq y$, and the vertical interval in green represents the range of values of $\tilde{A}(\sigma)$, for $\sigma\geq\bar{\sigma}$.}
    \label{fig:6}
\end{center}
\end{figure}

\noindent {\bf Step 1.3.b}. If $w>y$, we search for the time value function $V^i(K)$, on the interval $K\in[K^i_j,y]$, in the following form
$$
V^i(K) = \frac{1}{2}\left(A - \frac{\sigma}{z} B\right)  e^{-z(K - K^i_j)/\sigma}
+ \frac{1}{2}\left(A + \frac{\sigma}{z} B\right) e^{z(K - K^i_j)/\sigma}
$$
The above expression guarantees that, for $K\in[K^i_j,y]$,
$$
V^i(K) > A + B(K-K^i_j) \geq V^{i-1}(K),
$$
due to convexity of $V^{i-1}$ and the definition of $y$ (cf. Figure $\ref{fig:6}$). 
In addition, for any $\sigma>0$, the extended time value function remains in the form (\ref{eq.2}), on the interval $K\in[L_i,y]$.
The values of $V^i(K)$ and its derivative, at $K=y$, are given, respectively, by:
$$
\tilde{A}(\sigma) = \frac{1}{2}\left(A - \frac{\sigma}{z} B\right)  e^{-z(y - K^i_j)/\sigma}
+ \frac{1}{2}\left(A + \frac{\sigma}{z} B\right) e^{z(y - K^i_j)/\sigma},
$$
$$
\tilde{B}(\sigma) = -\frac{z}{2\sigma}\left(A - \frac{\sigma}{z} B\right)  e^{-z(y - K^i_j)/\sigma}
+ \frac{z}{2\sigma}\left(A + \frac{\sigma}{z} B\right) e^{z(y - K^i_j)/\sigma}
$$
The functions $\tilde{A}$ and $\tilde{B}$, introduced above, are the same as those introduced in Step 1.3.a, with $y$ in place of $w$. Lemma $\ref{le:temp.0}$ can still holds for these functions (one can simply repeat its proof given in Appendix B). In particular,
$\tilde{A}(\sigma)$ and $\tilde{B}(\sigma)$ are strictly decreasing in $\sigma\in(0,\infty)$, and the range of values of $\tilde{A}(\sigma)$ is given by
$$
(A+B(y-K^i_j),\infty),
$$
while the range of values of $\tilde{B}$ is $(B,\infty)$.
These observations yield the following lemma.

\begin{lemma}
The function
$$
\sigma\mapsto \tilde{A}(\sigma) + \tilde{B}(\sigma) (K^i_{j+1}-y),
$$
defined for all $\sigma>0$, is strictly decreasing, and its range of values is $(A + B(K^i_{j+1}-K^i_j),\infty)$.
\end{lemma}

Due to the convexity of $V^{i-1}$, and the fact that $V^{i-1}(K^i_j)<V^i(K^i_j)$ and $V^{i-1}(y)=V^i(y)$, we obtain: 
$$
V^{i-1}(K^i_{j+1}) >  A + B(K^i_{j+1}-K^i_j)
$$
Thus, there exists a unique $\bar{\sigma}$ satisfying
$$
\tilde{A}(\bar{\sigma}) + \tilde{B}(\bar{\sigma}) (K^i_{j+1}-y) = V^{i-1}(K^i_{j+1}),
$$
which can be computed via the bisection algorithm (cf. Figure $\ref{fig:6}$).

\begin{lemma}\label{le:temp.3.1}
The following inequalities hold:
$$
\tilde{B}(\bar{\sigma}) < \frac{\bar{V}(T_i,K^i_{j+1}) - \tilde{A}(\bar{\sigma})}{K^i_{j+1}-y} < B_1,
\,\,\,\,\,\,\,\,\,\,\,\,\,\,\tilde{A}(\bar{\sigma}) + \tilde{B}(\bar{\sigma}) (K-y) > V^{i-1}(K),
$$
for all $K\in[y,K^i_{j+1})$.
\end{lemma}

The proof of Lemma $\ref{le:temp.3.1}$ is given in Appendix B.
We add the next element of the partition, $\nu_{l(j)+1}=y$, and the next value of the diffusion coefficient, $\sigma_{l(j)+1}=\bar{\sigma}$. Finally, we repeat Step $1$, with $y$ in lieu of $K^i_j$, and with $\tilde{A}(\bar{\sigma})$ and with $\tilde{B}(\bar{\sigma})$ in lieu of $A$ and $B$ respectively. The already constructed time value function $V^i$, dominates $V^{i-1}$, on the interval $[L_i,y]$. Due to Lemma $\ref{le:temp.3.1}$, the values of $\tilde{A}(\bar{\sigma})$ and $\tilde{B}(\bar{\sigma})$, satisfy all the properties assumed for $A$ and $B$ at the beginning of Step $1$. In addition, due to the last inequality in Lemma $\ref{le:temp.3.1}$, the graphs of $\tilde{A}(\bar{\sigma}) + \tilde{B}(\bar{\sigma})(K-y)$ and $V^{i-1}$ do not intersect on $[y,K^i_{j+1}]$ (cf. Figure $\ref{fig:6}$). Hence, the algorithm in Step $1.3$ ends up in case a.
As a result, we obtain the values of $\nu_{l(j)+2}$, $\sigma_{l(j)+2}$, $\nu_{l(j)+3}=K_{j+1}$, $\sigma_{l(j)+3}$, and set $l(j+1)=l(j)+3$.


We repeat Step $1$, increasing $j$ by one each time, as long as $K^i_{j+1}< x$. Thus, we obtain an interpolation of the time value function on the interval $K\in[L_i,K^i_{j}]$, where $j$ is such that $x=K^i_{j+1}$.

\noindent{\bf Step 2}. Perform Step $1$, moving from $K=U_i$ backwards. Namely, we consider the change of variables: from $K\in[L_i,U_i]$ to $K'=L_i+U_i-K\in[L_i,U_i]$. We apply this change of variables to the market time values, as well as to the interpolated time value functions for earlier maturities. In particular, we obtain the new underlying level $x'=L_i+U_i-x$, as well as the new market time values $\bar{V}(T_i,L_i+U_i-K_j')$, for all strikes $K_j'=L_i+U_i-K^i_{N_i+1-j}$, with $j=0,\ldots,N_i+1$. Then, we repeat Step $1$, using this new, modified, data as an input. This gives us an interpolated time value function on the interval $K'\in [L_i,K_j']$, where $j$ is such that $x'=K_{j+1}'$. In the original variables, we obtain an interpolation of the time value function on the interval $K\in[K^i_{N_i+1-j},U_i]$, where, in turn, $x=K^i_{N_i-j}$. It is easy to see that the interpolated time value, on the interval $[K^i_{N_i+1-j},U_i]$, is in the form (\ref{eq.2.1}).

\noindent{\bf Step 3}. Finally, we extend the interpolated time value function to the entire interval $[L_i,U_i]$.
Recall that $x$ coincides with one of the strikes, say $K^i_j$. In Steps $1$ and $2$, we constructed the interpolated time value function for $K \in [L_i,K^i_{j-1}]\cup[K^i_{j+1},U_i]$. Notice that, repeating the algorithm outlined in Steps $1$ and $2$, we can extend $V^i$ to the entire interval $[L_i,U_i]$, so that it matches all market time values, dominates $V^{i-1}$ from above, and, in addition, it is in the form (\ref{eq.2}) on $[L_i,x]$, in the form (\ref{eq.2.1}) on $[x,U_i]$, and it has the prescribed left and right derivatives at $x$, $B_1$ and $B_2$ respectively. In order for the interpolation function $V^i$ to be a bona fide time value, it needs to have a jump of size $-1$ at $x$. In other words, we need to choose $B_1$ and $B_2$ so that $B_1-1=B_2$. The only constraints on the choices of $B_1$ and $B_2$ come from the market data and the previous constructed interpolation function on $K \in [L_i,K^i_{j-1}]\cup[K^i_{j+1},U_i]$. These constraints are given by: 
\begin{equation}\label{eq.calb.main.eq.last}
B_1 > \frac{\bar{V}(T_i,K^i_j) - \bar{V}(T_i,K^i_{j-1})}{K^i_j - K^i_{j-1}},\,\,\,\,\,\,\,\,\,\,
B_2 < \frac{\bar{V}(T_i,K^i_{j+1}) - \bar{V}(T_i,K^i_{j})}{K^i_{j+1} - K^i_{j}}
= \frac{\bar{C}(T_i,K^i_{j+1}) - \bar{C}(T_i,K^i_{j})}{K^i_{j+1} - K^i_{j}}
\end{equation}
Thus, we choose an arbitrary $\delta_4\in(0,1)$ and define
$$
B_1 = \delta_4 + 
\delta_4 \frac{\bar{V}(T_i,K^i_{j+1}) - \bar{V}(T_i,K^i_{j}) }{K^i_{j+1}-K^i_{j}}
+ \left(1 - \delta_4\right) \frac{\bar{V}(T_i,K^i_{j})  - \bar{V}(T_i,K^i_{j-1}) }{K^i_{j}-K^i_{j-1}}
$$
$$
 =
\delta_4 \frac{\bar{C}(T_i,K^i_{j+1}) - \bar{C}(T_i,K^i_{j}) }{K^i_{j+1}-K^i_{j}}
- \delta_4 \frac{\bar{C}(T_i,K^i_{j})  - \bar{C}(T_i,K^i_{j-1}) }{K^i_{j}-K^i_{j-1}}
+ \frac{\bar{V}(T_i,K^i_{j})  - \bar{V}(T_i,K^i_{j-1}) }{K^i_{j}-K^i_{j-1}}
$$
Similarly, we define
$$
B_2 =B_1-1
= \delta_4 - 1 + 
\delta_4 \frac{\bar{V}(T_i,K^i_{j+1}) - \bar{V}(T_i,K^i_{j}) }{K^i_{j+1}-K^i_{j}}
+ \left(1 - \delta_4\right) \frac{\bar{V}(T_i,K^i_{j})  - \bar{V}(T_i,K^i_{j-1}) }{K^i_{j}-K^i_{j-1}}
$$
$$
=
\delta_4\frac{\bar{C}(T_i,K^i_{j+1}) - \bar{C}(T_i,K^i_{j})}{K^i_{j+1}-K^i_{j}}
+(1-\delta_4) \frac{\bar{C}(T_i,K^i_{j}) - \bar{C}(T_i,K^i_{j-1})}{K^i_{j}-K^i_{j-1}}
$$
Since the market call prices are strictly admissible (in particular, strictly convex across strikes), we conclude that $B_1$ and $B_2$ satisfy (\ref{eq.calb.main.eq.last}).

The following lemma, whose proof is given in Appendix B, completes the proof of the theorem.
 
 \begin{lemma}\label{le:calib.le.1.2}
For any $z>0$ and any strictly admissible market data, $\left\{ T_i, K^i_j, \bar{C}(T_i,K^i_j), L_i, U_i, x\right\}$, the vectors $\nu^i$ and $\sigma^i$, for $i=1,\ldots,M$, produced by the above algorithm (Steps $0-3$) are such that the associated price curves $\left\{\bar{C}^i=C^{\nu^i,\sigma^i,z,x}\right\}_{i=1}^M$ satisfy Assumption $\ref{assump:ass.1}$.
\end{lemma}

\begin{remark}
Notice that the constants $\delta_i$, for $i=1,2,3,4$, appearing in the proof, can be chosen arbitrarily (within the interval $(0,1)$). The choice of the boundary points $\left\{L_i,U_i \right\}$ introduces additional flexibility in the algorithm. It is worth mentioning that the properties of interpolated implied volatilities and the associated diffusion coefficients may depend on the choices of these constants. Therefore, one may try to optimize over different values of $\left\{\delta_i\right\}$ and $\left\{L_i,U_i \right\}$, for example, to improve the smoothness of resulting diffusion coefficients. Continuing along these lines, one may also want to iterate over all possible choices of arbitrage-free option prices, consistent with the market. Recall that the market data contains only bid and ask quotes, and we propose to choose the exact values of option prices by solving a linear feasibility problem (cf. Remark $\ref{rem:addStrikes}$). Optimizing over all possible choices of option prices, in theory, may improve the smoothness properties of the output. However, this would come at a very high computational cost, and would lead to a multivariate nonlinear optimization problem, which we tried to avoid in the present paper.
\end{remark}

\section{Summary and Extensions}

In this paper, we considered the
risk-neutral process for the (forward) price of an asset, underlying
a set of European options, in the form of a driftless time-homogeneous diffusion
subordinated to an unbiased gamma process. We named this model LVG
and showed that it yields forward and backward PDDE's for option prices.
These PDDE's can be used to speed up the computations, but, more importantly,
they allow for explicit exact calibration of a LVG model to European options prices.
In particular, the diffusion coefficient of the underlying can be represented explicitly through a
single (continuous) implied volatility smile.
A mild extension of LVG allows for analogous calibration to multiple continuous smiles.
Finally, we showed how the aforementioned PDDE's can be used to construct continuously differentiable 
(with piecewise continuous second derivative) arbitrage-free implied volatility smiles from a finite set of option prices, with multiple strikes and maturities.
The resulting calibration algorithm does not involve any multivariate optimization, and it only requires solutions to one-dimensional root-search problems and elementary functions.
We illustrated the theory with numerical results based on a real market data.


There are many possible avenues for future research. 
One of the potential criticism of the proposed calibration method is the fact that the resulting diffusion coefficient has (a finite number of) discontinuities. 
This feature may create difficulties when computing the greeks, or prices of path-dependent and American-style options in the calibrated model. An idea on how to resolve this problem is outlined briefly in the few paragraphs preceding the proof of Theorem $\ref{th:calib.th.main}$. However, a more detailed analysis would be very useful.

Another limitation of the results presented herein is that the risk-neutral price process of the underlying is assumed to be a martingale. 
Indeed, in many cases, the dividend and interest rates are not negligible. If they are assumed to be deterministic, then, the calibration problem can be easily reduced to the driftless case by discounting. However, as discussed in Remark $\ref{rem:shortMat}$, the assumption of deterministic rates may not be consistent with the market. Therefore, an extension of the proposed model that allows for a non-zero drift of the underlying is very desirable.
In particular, it seems possible to extend the results of this paper to a risk-neutral price process obtained as a time change of a diffusion with drift. However, certain challenges may arise and they need to be addressed: in particular, the associated PDDE's may no longer have closed form solutions even if the diffusion and drift coefficients are piecewise constant.

A more thorough numerical analysis, in particular, investigating the prices of exotic products produced by the calibrated models, is also very desirable. We did not conduct such a study in the present paper, due to its already quite extensive length.

Finally, one can examine the problem of calibrating a LVG model to other financial instruments, such as the barrier or American options.

\section*{Appendix A}

\emph{\bf Proof of Theorem $\ref{th:DD.th1}$}.
Theorems 2.2 and 2.7 in \cite{KrylovVMO}, as well as Theorem 1.II in \cite{Campanato} (a respective theorem is applicable depending on whether each of the boundary points is finite) provide existence and uniqueness results for a class of parabolic PDEs with initial and boundary conditions. This class includes the following equation:
\begin{equation*}
\partial_{\tau} u(\tau,x) = \frac{1}{2} a^2(x) \partial^2_{x} u(\tau,x) + \frac{1}{2} a^2(x) \phi''(x) ,
\,\,\,\,\,\,\,\,\,\, u(0,x)=u(\tau,L_+)=u(\tau,U_-)=0,
\end{equation*}
implying that it has a unique weak solution.
In addition $u$ is absolutely bounded by a constant (which follows from the estimates in Theorems 2.2 and 2.7, in \cite{KrylovVMO}, and in Theorem 1.II, in \cite{Campanato}), and it is vanishing at $L$ and $U$.
Then, $v=u+\phi$ solves (\ref{eq.DD.BS}) and satisfies $v(0,x)=\phi(x)$. In addition, it satisfies: $v(\tau,L_+)=\phi(L_+)$ and/or $v(\tau,U_-)=\phi(U_-)$, whenever $L$ and/or $U$ are finite. To show that this solution coincides with $V^{D,\phi}$, we apply the generalized It\^o's formula (cf. Theorem 2.10.1 in \cite{KrylovControlledDiff}), to obtain
$$
\EE^x v(T-\zeta^{r,R}\wedge T, D_{\zeta^{r,R}\wedge T}) = v(T,x),
$$
where $\zeta^{r,R}$ is the first exit time of $D$ from the interval $(r,R)$, with arbitrary $L<r<R<U$. Notice that $v(\tau,x)$ is bounded by an exponential of $x$, uniformly over $\tau\in(0,T)$. Thus, we can apply the dominated convergence theorem, to conclude that, as $r\downarrow L$ and $R\uparrow U$,
$$
v(T,x) = \EE^x v(T-\zeta^{r,R}\wedge T, D_{\zeta^{r,R}\wedge T}) \rightarrow \EE^x v(T-\zeta\wedge T, D_{\zeta\wedge T})
= \EE^x v(T-\zeta, D_{\zeta})\bone_{\left\{\zeta\leq T\right\}} 
+ \EE^x \phi(D_{T})\bone_{\left\{\zeta > T\right\}},
$$
where $\zeta$ is the time of the first exit of $D$ from $(L,U)$. Recall that the diffusion can only exit the interval $(L,U)$ through a finite boundary point. Thus,
$$
v\left(T-\zeta, D_{\zeta}\right)\bone_{\left\{\zeta\leq T\right\}} 
= v\left(T-\zeta, L_+\right)\bone_{\left\{\zeta\leq T,\, D_{\zeta} = L_+\right\}} 
+ v\left(T-\zeta, U_-\right)\bone_{\left\{\zeta\leq T,\, D_{\zeta} = U_-\right\}}
$$
$$
= \phi(L_+)\bone_{\left\{\zeta\leq T,\, D_{\zeta} = L_+\right\}} 
+ \phi(U_-)\bone_{\left\{\zeta\leq T,\, D_{\zeta} = U_-\right\}}
 = \phi(D_{\zeta}) \bone_{\left\{\zeta\leq T\right\}}
$$
Since the diffusion is absorbed at the boundary points, we have:
$$
v(T,x) = \EE^x\left( \phi(D_{T})\bone_{\left\{\zeta > T\right\}}\right) + \EE^x\left( \phi(D_{\zeta}) \bone_{\left\{\zeta\leq T\right\}}\right)
= \EE^x\left( \phi(D_{T})\bone_{\left\{\zeta > T\right\}}\right) + \EE^x\left( \phi(D_{T}) \bone_{\left\{\zeta\leq T\right\}}\right)
= \EE^x \phi(D_{T}),
$$
which completes the proof of the theorem.

\emph{\bf Proof of Theorem $\ref{th:DD.th2}$}. Let us apply the It\^o-Tanaka rule (cf. Theorem 7.1, on page 218, in \cite{KaratzasShreve}) to the call payoff:
$$
(D_t - K)^+ = (D_0 - K)^+ + \int_0^t a(D_s) \bone_{\left\{ D_s \geq K \right\}} d W_s + \Lambda_t(K),
$$
where $\Lambda$ is the local time of $D$. Taking expectations, we obtain
$$
C^D(t,K,x) = \EE^x (D_t - K)^+ = (x - K)^+ + \EE \Lambda_t(K).
$$
Multiplying by $\phi$ and integrating, we use the properties of local time (cf. Theorem 7.1, on page 218, in \cite{KaratzasShreve}), to obtain
$$
\int_L^U C^D(t,K,x) \phi(K) dK = \int_L^U (x - K)^+ \phi(K) dK + \frac{1}{2} \int_0^t \EE^x \left(a^2(D_u) \phi(D_u)\right) du.
$$
Using Fubini's theorem, we obtain the following identity:
$$
\QQ^{D,x}(D_u \leq K) = \EE^x \bone_{\left\{ D_u \leq K \right\}} = \partial_K C^D(u,K,x),
$$
which holds for almost all $K\in(L,U)$.
Therefore, the measure generated by $\partial_K C^D(u,\cdot,x)$ is the distribution of $D_u$ under measure $\QQ^{D,x}$.
Thus, we conclude:
$$
\int_L^U C^D(t,K,x) \phi(K) dK = \int_L^U (x - K)^+ \phi(K) dK + \frac{1}{2} \int_0^t \int_L^U a^2(K) \phi(K) d\left(\partial_K C^D(u,K,x)\right) du.
$$

\emph{\bf Proof of Theorem $\ref{th:pdde.BS.th1}$}. 
Let us multiply the left hand side of equation (\ref{eq.DD.BS}) by the density of $\Gamma_{t}$ and an arbitrary test function $\varphi\in C^{\infty}_0\left((L,U)\right)$, then, divide it by $a^2/2$ and integrate with respect to $\tau$ and $x$:
$$
\int_{L}^U \int_0^T \frac{\tau^{t/t^* - 1}e^{-\tau/t^*}}{(t^*)^{t/t^*} \Gamma(t/t^*)} \partial_{\tau}V^{D,\phi}(\tau,x) \frac{2\varphi(x)}{a^2(x)} d\tau dx
= \int_{L}^U \int_0^T \frac{\tau^{t/t^* - 1}e^{-\tau/t^*}}{(t^*)^{t/t^*} \Gamma(t/t^*)} \partial_{\tau} \left( V^{D,\phi}(\tau,x) - \phi(x) \right) d\tau \frac{2\varphi(x)}{a^2(x)} dx
$$
\begin{equation}\label{eq.PfThm5.conv.1}
= \frac{T^{t/t^* - 1}e^{-T/t^*}}{(t^*)^{t/t^*} \Gamma(t/t^*)} \int_{L}^U \left(V^{D,\phi}(T,x) - \phi(x)\right)\frac{2\varphi(x)}{a^2(x)} dx 
\end{equation}
$$
- \int_0^T \frac{\left((t/t^* - 1)\tau^{t/t^* - 2} - \tau^{t/t^* - 1}/t^* \right)e^{-\tau/t^*}}{(t^*)^{t/t^*} \Gamma(t/t^*)} 
\int_{L}^U \left( V^{D,\phi}(\tau,x) - \phi(x) \right) \frac{2\varphi(x)}{a^2(x)} dx d\tau
$$  
$$
\rightarrow \int_{L}^U \int_0^{\infty} \frac{\left( \tau^{t/t^* - 1}/t^* - (t/t^* - 1)\tau^{t/t^* - 2} \right)e^{-\tau/t^*}}{(t^*)^{t/t^*} \Gamma(t/t^*)} 
 \left( V^{D,\phi}(\tau,x) - \phi(x) \right) d\tau \frac{2\varphi(x)}{a^2(x)} dx,
$$
as $T\rightarrow\infty$.
In the above, we applied Fubini's theorem and used the fact that $\phi$ is absolutely bounded and, hence, so is $V^{D,\phi}$. Recall also that $\tau\geq t^*$. 
Applying Fubini's theorem once more, we interchange the order of integration in the left hand side of (\ref{eq.PfThm5.conv.1}), and, using equation (\ref{eq.DD.BS}), integrate by parts with respect to $x$. This, together with the convergence in (\ref{eq.PfThm5.conv.1}), implies:
\begin{equation}\label{eq.PfThm5.conv.2}
\int_0^{\infty}  \frac{\tau^{t/t^* - 1}e^{-\tau/t^*}}{(t^*)^{t/t^*} \Gamma(t/t^*)} \int_L^U V^{D,\phi}(\tau,x) \varphi''(x)dx d\tau
\end{equation}
$$
=
\int_L^U \int_0^{\infty} \frac{\left( \tau^{t/t^* - 1}/t^* - (t/t^* - 1)\tau^{t/t^* - 2} \right)e^{-\tau/t^*}}{(t^*)^{t/t^*} \Gamma(t/t^*)} \left(V^{D,\phi}(\tau,x) - \phi(x) \right) d\tau \frac{2}{a^2(x)} \varphi(x) dx,
$$
where both integrals are absolutely convergent.
Using Fubini's theorem and the definition of $V^{\phi}$, we conclude that 
\begin{equation}\label{eq.PfThm5.conv.2.1}
\int_0^{\infty}  \frac{\tau^{t/t^* - 1}e^{-\tau/t^*}}{(t^*)^{t/t^*} \Gamma(t/t^*)} \int_L^U V^{D,\phi}(\tau,x) \varphi''(x)dx d\tau
= \int_L^U V^{\phi}(t,x) \varphi''(x)dx
\end{equation}
Next, we notice that, since $t\geq t^*$, the expression inside the outer integral in the right hand side of (\ref{eq.PfThm5.conv.2}) is a locally integrable function of $x\in(L,U)$. Then, the fact that equations (\ref{eq.PfThm5.conv.1}) and (\ref{eq.PfThm5.conv.2}) hold for all test functions $\varphi$ yields:
\begin{equation}\label{eq.PfThm5.conv.2.2}
\partial^2_{xx}V^{\phi}(t,x)  
= \frac{2}{a^2(x)} \int_0^{\infty} \frac{\left( \tau^{t/t^* - 1}/t^* - (t/t^* - 1)\tau^{t/t^* - 2} \right) e^{-\tau/t^*}}{(t^*)^{t/t^*} \Gamma(t/t^*)} \left(V^{D,\phi}(\tau,x) - \phi(x) \right) d\tau,
\end{equation}
which holds for every $t\geq t^*$ and almost all $x\in(L,U)$.
Since the second derivative of $V^{\phi}(\tau,\cdot)$ is locally integrable, its first derivative is absolutely continuous. To show that $V^{\phi}(\tau,L_+)=V^{\phi}(\tau,U_-)=0$, we recall the definition of $V^{\phi}$ as an expectation and apply the dominated convergence theorem.
When $t=t^*$, equation (\ref{eq.PfThm5.conv.2.2}) becomes
$$
\partial^2_{xx}V^{\phi}(t^*,x)
= \frac{2}{a^2(x) t^*}
\int_0^{\infty} \frac{e^{-\tau/t^*}}{t^*} \left(V^{D,\phi}(\tau,x) - \phi(x) \right) d\tau
= \frac{2}{a^2(x) t^*} \left(V^{\phi}(t,x) - \phi(x) \right),
$$
which holds for almost all $x\in(L,U)$.
Since $V^{\phi}(\tau,\cdot)$ and $\phi$ are continuous, the above equation holds for all points $x$ at which $a$ is continuous. Therefore, $V^{\phi}(t^*,\cdot)$ possesses all the properties stated in the theorem.

Next, we show uniqueness. Consider any function $u\in C^1(L,U)$, which has an absolutely continuous derivative and satisfies (\ref{eq.pdde.BS}), along with zero boundary conditions. Applying the generalized It\^o's formula (cf. Theorem 2.10.1 in \cite{KrylovControlledDiff}), we obtain:
$$
u(D_{t\wedge\zeta^{r,R}}) - u(D_0) = \int_0^{t\wedge\zeta^{r,R}} \frac{1}{2} a^2(D_s) u''(D_s) ds + \int_0^{t\wedge\zeta^{r,R}} u'(D_s) dD_s,
$$
where, as before, $\zeta^{r,R}$ is the first exit time of $D$ from the interval $(r,R)$, with arbitrary $L<r<R<U$. Notice that, since $u$ is continuous and satisfies zero boundary conditions, it is absolutely bounded. Due to equation (\ref{eq.pdde.BS}), $u''$ is absolutely bounded as well. Thus, taking expectations and applying the dominated convergence theorem, as $r\downarrow L$ and $R\uparrow U$, we obtain
$$
\EE^x u(D_{t}) - u(x) = \EE \int_0^{t} \frac{1}{2} a^2(D_s) u''(D_s) \bone_{\left\{s\leq\zeta\right\}} ds
$$
Since $\Gamma$ is independent of $D$, we can plug $\Gamma_{t^*}$, in place of $t$, into the above expression, to obtain: 
$$
\EE^x u\left(D_{\Gamma_{t^*}}\right) - u(x) = \EE^x \int_0^{\Gamma_{t^*}} \frac{1}{2} a^2(D_s) u''(D_s) \bone_{\left\{s\leq\zeta\right\}} ds,
$$
$$
 \int_0^{\infty} \frac{e^{-t/t^*}}{t^*} \EE^x u(D_{t}) dt - u(x) = \int_0^{\infty} \frac{e^{-t/t^*}}{t^*} \EE^x \int_0^{t} \frac{1}{2} a^2(D_s) u''(D_s) \bone_{\left\{s\leq\zeta\right\}} ds dt,
$$
$$
 \int_0^{\infty} \frac{e^{-t/t^*}}{t^*} \EE^x u(D_{t}) dt - u(x) = \int_0^{\infty} \frac{e^{-t/t^*}}{t^*} \EE^x \int_0^{t} \frac{1}{t^*} \left( u(D_s) - \phi(D_s) \right) \bone_{\left\{s\leq\zeta\right\}} ds dt,
$$
\begin{equation}\label{eq.PfThm5.eq.1}
 \int_0^{\infty} \frac{e^{-t/t^*}}{t^*} \EE^x u(D_{t}) dt - u(x) = \frac{1}{t^*} \int_0^{\infty} \frac{e^{-t/t^*}}{t^*} \int_0^{t}  \EE^x \left( u(D_s) - \phi(D_s) \right) ds dt,
\end{equation}
where we made use of the fact that $u(L_+)=u(U_-)=\phi(L_+)=\phi(U_-)=0$, to obtain (\ref{eq.PfThm5.eq.1}). Thus, we continue:
\begin{equation}\label{eq.PfThm5.eq.2}
\int_0^{\infty} \frac{e^{-t/t^*}}{t^*} \EE^x u(D_{t}) dt - u(x) = \frac{1}{t^*} \int_0^{\infty} e^{-t/t^*}  \EE^x \left( u(D_t) -  \phi(D_t) \right) dt,
\end{equation}
where we integrated by parts in the right hand side of (\ref{eq.PfThm5.eq.1}), to obtain (\ref{eq.PfThm5.eq.2}). This yields:
$$
u(x) = \frac{1}{t^*} \int_0^{\infty} e^{-t/t^*}  \EE^x \phi(D_t) dt = V^{\phi}(t^*,x)
$$

\emph{\bf Proof of Lemma $\ref{le:pdde.DUP.le1}$}. Let us show that $S_{t^*}$ has a density in $(L,U)$ under any $\QQ^x$. Consider arbitrary $K\in(L,U)$ and $\varepsilon\in(0,1)$, s.t. $[K-\varepsilon,K+2\varepsilon] \subset (L,U)$. Choose a smooth function $\rho_{\varepsilon}$, which dominates the indicator function of $[K,K+\varepsilon]$ from above, whose support is contained in $[K-\varepsilon,K+2\varepsilon]$, and which satisfies
$$
\int_L^U \rho^2_{\varepsilon}(x) dx \leq 2\varepsilon.
$$
Introduce
$$
u(x) = \EE^x \left(\rho_{\varepsilon}(S_{t^*})\right),
$$
for all $x\in(L,U)$. Theorem $\ref{th:pdde.BS.th1}$ implies that $u\in C^1(L,U)$, $u(L_+)=u(U_-)=0$, $u''\in L^2(L,U)$, and it satisfies:
\begin{equation}\label{eq.AA.1}
\frac{1}{2} a^2(x) u''(x) - \frac{1}{t^*} \left(u(x) - \rho_{\varepsilon}(x)\right) = 0,
\end{equation}
for all $x\in(L,U)$, except the points of discontinuity of $a$. Next, we need to recall the standard estimates for the weak solutions of elliptic PDEs via the norm of the right hand side (see, for example, \cite{KrylovVMOelliptic} and references therein). The only issue that needs to be resolved, before we can apply the standard results, is that, although $u$ belongs to the Sobolev space $W^2_2\left((L,U)\right)$, it may not belong to $W^2_2(\RR)$. This is due to the fact that $u'$ may have a discontinuity at $L$ or $U$, if the boundary point is attainable. To resolve this problem, we consider a sequence of functions $u_n\in C_0^{\infty}(\RR)\subset W^2_2(\RR)$ which approximates $u$ in the sense of the Sobolev norm $\|\cdot\|_{W^2_2(L,U)}$. 
Applying Remark 3.3, in \cite{KrylovVMOelliptic}, to $u_n$ and passing to the limit, as $n\rightarrow\infty$, we estimate the Sobolev norm of $u$: 
\begin{equation}\label{eq.SobEst.1}
\int_L^U u^2(x) + (u'(x))^2 + (u''(x))^2 dx \leq c_1 \int_L^U \rho^2_{\varepsilon}(x) dx,
\end{equation}
where the constant $c_1$ is independent of $\varepsilon$.
Next, the Morrey's inequality yields
$$
\sup_{x\in(L,U)} |u(x)| \leq c_2 \left( \int_L^U u^2(x) + (u'(x))^2 dx \right)^{1/2} \leq c_3 \left( \int_L^U \rho^2_{\varepsilon}(x) dx\right)^{1/2} \leq c_3 \sqrt{2\varepsilon}
$$
Finally, we obtain
$$
\QQ^x(S_{t^*}\in[K,K+\varepsilon]) \leq u(x) \leq c_3 \sqrt{2\varepsilon},
$$
Therefore, the distribution of $S_{t^*}$ is absolutely continuous with respect to the Lebesgue measure on $(L,U)$ and, in particular, it has a density. It was shown in the proof of Theorem $\ref{th:DD.th2}$ that
$$
\QQ^x(S_{t^*}\leq K)=\partial_K C(t^*,K,x)
$$
Therefore, $\partial_K C(t^*,K,x)$ is absolutely continuous in $K\in(L,U)$, and its derivative, $\partial^2_{KK} C(t^*,\cdot,x)$, is the density of $S_{t^*}$ under $\QQ^x$.

\emph{\bf Proof of Theorem $\ref{th:pdde.DUP.th1}$}.
Differentiating both sides of (\ref{eq.DD.dup1}) with respect to $\tau$, then, multiplying them by an exponential and a test function $\varphi\in C^{\infty}_0\left((L,U)\right)$, we integrate, to obtain
\begin{equation}\label{eq.PfThm7.eq.0}
\int_0^{T} \frac{e^{-\tau/t^*}}{t^*} \frac{d}{d \tau} \int_L^U C^D(\tau,K,x) \frac{2}{a^2(K)} \varphi(K) dK d\tau 
= \int_0^{T} \frac{e^{-\tau/t^*}}{t^*} \int_L^U  \varphi(K) d\left(\partial_{K} C^D(\tau,K,x)\right) d\tau
\end{equation}
Let us perform integration by parts in the left hand side of (\ref{eq.PfThm7.eq.0}):
$$
\int_0^{T} \frac{e^{-\tau/t^*}}{t^*} \frac{d}{d \tau} \int_L^U C^D(\tau,K,x)  \frac{2}{a^2(K)} \varphi(K) dK d\tau 
= \frac{e^{-T/t^*}}{t^*} \int_L^U C^D(T,K,x)  \frac{2}{a^2(K)}\varphi(K) dK
$$
\begin{equation}\label{eq.PfThm7.eq.1}
- \frac{1}{t^*} \int_L^U (x-K)^+  \frac{2}{a^2(K)}\varphi(K) dK 
+ \frac{1}{t^*} \int_0^{T} \frac{e^{-\tau/t^*}}{t^*} \int_L^U C^D(\tau,K,x)  \frac{2}{a^2(K)}\varphi(K) dK d\tau,
\end{equation}
where we made use of the fact that 
$$
\lim_{\tau\downarrow 0} C^D(\tau,K,x) = (x-K)^+,
$$
which, in turn, follows from square integrability of the martingale $D$ and the dominated convergence theorem. 
Using (\ref{eq.PfThm7.eq.1}), as well as the dominated convergence and Fubini's theorems, we obtain:
$$
\lim_{T\rightarrow\infty} \int_0^{T} \frac{e^{-\tau/t^*}}{t^*} \frac{d}{d \tau} \int_L^U C^D(\tau,K,x)  \frac{2}{a^2(K)}\varphi(K) dK d\tau
=  \int_L^U \frac{2}{a^2(K)t^*}\left(C(t^*,K,x) - (x-K)^+\right) \varphi(K) dK
$$
On the other hand, integrating by parts and applying the Fubini and dominated convergence theorems once more, we derive:
$$
\lim_{T\rightarrow\infty} \int_0^{T} \frac{e^{-\tau/t^*}}{t^*} \int_L^U  \varphi(K) d\left(\partial_{K} C^D(\tau,K,x)\right) d\tau
= \lim_{T\rightarrow\infty}  \int_L^U  \varphi''(K) \int_0^{T} \frac{e^{-\tau/t^*}}{t^*} C^D(\tau,K,x) d\tau dK
$$
$$
= \int_L^U C(t^*,K,x) \varphi''(K) dK = \int_L^U \partial^2_{KK} C(t^*,K,x) \varphi(K) dK
$$
Collecting the above, we rewrite equation (\ref{eq.PfThm7.eq.0}) as
\begin{equation*} 
\int_L^U \left( \frac{1}{t^*} \left( C(t^*,K,x) - (x-K)^+\right) - \frac{1}{2} a^2(K) \partial^2_{KK} C(t^*,K,x) \right) \varphi(K) dK = 0,
\end{equation*}
which holds for any $\varphi\in C^{\infty}_0\left((L,U)\right)$. This implies that equation (\ref{cfpdde}) holds for almost every $K\in(L,U)$. Since $C(t^*,\cdot,x)$ and $(x-\cdot)^+$ are continuous functions, we conclude that (\ref{cfpdde}) holds everywhere except the points of discontinuity of $a$.

The boundary conditions for call price function follow from the dominated convergence theorem. The square integrability of $\partial^2_{KK} C(t^*,\cdot,x)$ follows from the fact that it is absolutely integrable and vanishes at the boundary (the latter, in turn, follows from the boundary conditions and equation (\ref{cfpdde})). To show uniqueness, we assume that there exists another function, with the same properties, and denote the difference between the two by $u$. Function $u$ is a weak solution of the homogeneous version of (\ref{cfpdde}):
$$
\frac{1}{2} a^2(x) u''(x) - \frac{1}{t^*} u(x) = 0,
$$
vanishing at the boundary. Finally, as in the proof of Lemma $\ref{le:pdde.DUP.le1}$, we apply Remark 3.3 in \cite{KrylovVMOelliptic} to obtain an estimate for the Sobolev norm of $u$, analogous to (\ref{eq.SobEst.1}), with zero in lieu of $\rho_{\varepsilon}$. This implies that $u\equiv 0$.

\emph{\bf Proof of Theorem $\ref{th:pdde.DUP.th2}$}.
Notice that, for a fixed $K\in(L,U)$, the function $C(t^*,K,\,.\,)$ is measurable. This follows from the measurability of $C^D(T,K,\,.\,)$, which, in turn, follows from the measurability of the mapping $x\mapsto \QQ^{D,x}$, as a property of Markov family. At any point $K\in(L,U)$ that does not coincide with a discontinuity point of $a$, the call price function is twice continuously differentiable in $K$, and, hence, its second derivative, $\partial^2_{KK}C(t^*,K,x)$, is also measurable as a function of $x$. Thus, we can use (\ref{cfpdde}) to obtain
\begin{equation}\label{eq.appendixA.pdde.DUP.th2.1}
\frac{1}{t^*} \left( \EE^xC(t^*,K,S_{\tau}) - \EE^x (S_{\tau}-K)^+\right)
 =   \frac{a^2(K)}{2} \EE^x \left(\partial^2_{KK}C(t^*,K,S_{\tau})\right),
\end{equation}
which holds for all $\tau\geq0$ and all $K\in(L,U)$ except the points of discontinuity of $a$. Since $\partial^2_{KK}C(t^*,\cdot,S_{\tau})$ is nonnegative and integrates to one, the application of Fubini's theorem yields
\begin{equation}\label{eq.PfThm8.eq.1}
\EE^x \left(\partial^2_{KK}C(t^*,K,S_{\tau})\right) = \partial^2_{KK}\EE^x C(t^*,K,S_{\tau}) 
= \partial^2_{KK}C(\tau + t^*,K,x),
\end{equation}
where the derivatives in the right hand side are understood in a weak sense, and we used Proposition $\ref{th:lvg.le1}$ to obtain the last equality. Due to (\ref{eq.appendixA.pdde.DUP.th2.1}), $\partial^2_{KK}C(\tau + t^*,\cdot,x)$ is, in fact, locally integrable on $(L,U)$, and equation (\ref{cfpdde1}) holds for almost every $K\in(L,U)$. Using the monotone convergence theorem, it is easy to show that $C(\tau,\cdot,x)$ is continuous, for any $\tau\geq0$, and, hence, (\ref{cfpdde1}) holds everywhere except the points of discontinuity of $a$. The boundary conditions follow from the dominated convergence theorem. In addition, (\ref{eq.PfThm8.eq.1}) and an application of Fubini's theorem imply that $\partial^2_{KK}C(\tau + t^*,\cdot,x)$ is absolutely integrable on $(L,U)$. This, along with the boundary conditions and equation (\ref{cfpdde1}), imply that $\partial^2_{KK}C(\tau + t^*,\cdot,x)$ is square integrable. 
To show uniqueness, we apply the same argument as at the end of the proof of Theorem $\ref{th:pdde.DUP.th1}$.

\section*{Appendix B}

\emph{\bf Proof of Lemma $\ref{le:temp.0}$}.
Notice that, for all $y>0$, we have:
$$
\left( (1+y) e^{-2y} + y - 1 \right)'
= 1 - e^{-2y}(1+2y) > 0,
$$
which follows from the inequality $e^{2y}>1+2y$. Since $(1+0) e^{-2\cdot0} + 0 - 1=0$, we obtain
$$
(1+y) e^{-2y} + y - 1 > 0,\,\,\,\,\,\,\,\,\,\,\,y>0.
$$
Therefore, for any positive $A$ and $B$,
$$
\tilde{A}'(\sigma) = \frac{1}{2}\left(\left(A - \frac{\sigma}{z} B\right) \frac{z(w-K^i_j)}{\sigma^2} - \frac{B}{z} \right) e^{-z(w-K^i_j)/\sigma}
- \frac{1}{2} \left(\left(A + \frac{\sigma}{z} B\right) \frac{z(w-K^i_j)}{\sigma^2} - \frac{B}{z} \right) e^{z(w-K^i_j)/\sigma}
$$
$$
 = \frac{1}{2}\left( A \frac{z(w-K^i_j)}{\sigma^2} - \frac{B}{z}\left(\frac{z(w-K^i_j)}{\sigma}  + 1\right) \right) e^{-z(w-K^i_j)/\sigma}
$$
$$
- \frac{1}{2} \left(A \frac{z(w-K^i_j)}{\sigma^2} + \frac{B}{z}\left( \frac{z(w-K^i_j)}{\sigma} - 1\right) \right) e^{z(w-K^i_j)/\sigma}
 = \frac{1}{2} A \frac{z(w-K^i_j)}{\sigma^2} \left(e^{-z(w-K^i_j)/\sigma} - e^{z(w-K^i_j)/\sigma} \right)
$$
$$
 - \frac{1}{2} \frac{B}{z} e^{z(w-K^i_j)/\sigma} \left( \left(\frac{z(w-K^i_j)}{\sigma}  + 1\right) e^{-2z(w-K^i_j)/\sigma}
+ \left( \frac{z(w-K^i_j)}{\sigma} - 1\right) \right) \leq 0
$$
Similarly, we have:
$$
\tilde{B}'(\sigma) =
-\frac{1}{2}\left( \left( \frac{z}{\sigma} A - B\right) \frac{z(w-K^i_j)}{\sigma^2} - \frac{z}{\sigma^2} A  \right)  e^{-z(w-K^i_j)/\sigma}
$$
$$
+ \frac{1}{2}\left( -\left(\frac{z}{\sigma} A + B\right) \frac{z(w-K^i_j)}{\sigma^2} - \frac{z}{\sigma^2} A  \right) e^{z(w-K^i_j)/\sigma}
$$
$$
=
-\frac{z}{2\sigma^2}\left( \frac{z(w-K^i_j)}{\sigma} A - B (w-K^i_j) - A  \right)  e^{-z(w-K^i_j)/\sigma}
$$
$$
- \frac{z}{2\sigma^2}\left( \frac{z(w-K^i_j)}{\sigma} A + B (w-K^i_j) + A  \right) e^{z(w-K^i_j)/\sigma}
$$
$$
\leq
\frac{z}{2\sigma^2}\left( B (w-K^i_j) + A  \right)  e^{-z(w-K^i_j)/\sigma}
- \frac{z}{2\sigma^2}\left( B (w-K^i_j) + A  \right) e^{z(w-K^i_j)/\sigma}
\leq 0
$$
Taking limits as $\sigma$ converges to $0$ and $\infty$, we complete the proof of the lemma.

\emph{\bf Proof of Lemma $\ref{le:temp.2}$}.
For any $a<\bar{V}(T_i,K^i_{j+1})$ and any
$$
b\in\left(0,\frac{\bar{V}(T_i,K^i_{j+1})-a}{K^i_{j+1}-w}\right),
$$ 
we introduce $\Sigma=\Sigma(a,b)$ as the unique solution to
\begin{equation}\label{eq.PfLem20.1}
\frac{1}{2}\left(a + \frac{\Sigma}{z} b \right) e^{z(K^i_{j+1}-w)/\Sigma}
+ \frac{1}{2}\left(a - \frac{\Sigma}{z} b \right) e^{-z(K^i_{j+1}-w)/\Sigma}
 = \bar{V}(T_i,K^i_{j+1}),
\end{equation}
It is shown in the proof of Theorem $\ref{th:calib.th.main}$ that such $\Sigma$ exists and is unique. 
It was also shown that the left hand side of (\ref{eq.PfLem20.1}) is strictly increasing in $a$ and $b$, and it is strictly decreasing in $\Sigma$. This yields that $\Sigma(a,b)$ is strictly increasing in $a$ and $b$.

Consider the associated function
$$
F(a,b;K) = \frac{1}{2}\left(a + \frac{\Sigma(a,b)}{z} b\right) e^{z(K-w)/\Sigma(a,b)}
+ \frac{1}{2}\left(a - \frac{\Sigma(a,b)}{z} b\right)  e^{-z(K-w)/\Sigma(a,b)},
$$
defined for $K\in[w,K^i_{j+1}]$. Note that $F(a,b;\cdot)$ is convex, increasing, and it solves the ODE (\ref{eq.calib.TV.2}). In addition, $F(a,b;w)=a$, $\partial_K F(a,b;w)=b$, $F(a,b;K^i_{j+1}) = \bar{V}(T_i,K^i_{j+1})$, and
$$
\partial_K F(a,b;K^i_{j+1}) = 
\frac{z}{2\Sigma(a,b)}\left(a + \frac{\Sigma(a,b)}{z} b\right) e^{z\left(K^i_{j+1}-w\right)/\Sigma(a,b)}
-\frac{z}{2\Sigma(a,b)}\left(a - \frac{\Sigma(a,b)}{z} b\right)  e^{-z\left(K^i_{j+1}-w\right)/\Sigma(a,b)}
$$
Let us show that, for a fixed $b$, $\partial_K F(a,b;K^i_{j+1})$ is strictly decreasing as a function of 
$$
a\in\left(\bar{V}\left(T_i,K^i_{j+1}\right) - b\left(K^i_{j+1}-w\right),\,\bar{V}\left(T_i,K^i_{j+1}\right)\right)
$$
Choose any two points $a_1<a_2$, from the above interval. Notice that the function $u(K)=F(a_1,b;K) - F(a_2,b;K)$ satisfies:
\begin{equation}\label{eq.PfLem20.2}
u''(K) - \frac{z^2}{\Sigma^2(a_2,b)} u(K) = \left(  \frac{z^2}{\Sigma^2(a_1,b)} -  \frac{z^2}{\Sigma^2(a_2,b)} \right) F(a_1,b;K),
\end{equation}
on the interval $K\in(w,K^i_{j+1})$
The right hand side of (\ref{eq.PfLem20.2}) is nonnegative. Notice that $u(w)=a_1 - a_2 < 0$. Assume that there is $K\in(w,K^i_{j+1})$, such that $u(K)>0$, then, there must exist a point $w_1\in(w,K^i_{j+1})$, such that $u(w_1)>0$ and $u'(w_1)>0$. From (\ref{eq.PfLem20.2}), we conclude that $u(K)$ is strictly convex in the interval $K\in[w_1,K^i_{j+1}]$. In particular, it means that $u(K^i_{j+1})>0$, which is a contradiction. Therefore, the above assumption is wrong and $u(K)\leq0$, for all $K\in[w,K^i_{j+1}]$. This means that $F(a_1,b;K) \leq F(a_2,b;K)$, for all $K\in[w,K^i_{j+1}]$, and, due to the fact that $F(a_1,b;K^i_{j+1}) = F(a_2,b;K^i_{j+1})$, we obtain $\partial_K F(a_1,b;K^i_{j+1}) \geq \partial_K F(a_2,b;K^i_{j+1})$. To show that the inequality is strict, we notice that, if $\partial_K F(a_1,b;K^i_{j+1}) = \partial_K F(a_2,b;K^i_{j+1})$, then, $u'(K^i_{j+1})=0$ and (\ref{eq.PfLem20.2}) implies that $u''(K^i_{j+1})>0$, which, in turn, implies that $u$ is strictly positive in a left neighborhood of $K^i_{j+1}$, which is impossible.
Repeating the above arguments, one can show that $\partial_K F(a,b;K^i_{j+1})$ is strictly decreasing in $b\in(0,(\bar{V}(T_i,K^i_{j+1})-a)/(K^i_{j+1}-w))$.

Finally, using both the current notation and that introduced in the proof of Theorem $\ref{th:calib.th.main}$ (Step 1.3.a), we notice that: 
$$
\tilde{\sigma}(\sigma) = \Sigma(\tilde{A}(\sigma),\tilde{B}(\sigma)),
$$
and recall that $\tilde{A}(\sigma)$ and $\tilde{B}(\sigma)$ are strictly decreasing in $\sigma\in(\hat{\sigma},\infty)$. It is easy to see that the range of values of $(\tilde{A}(\sigma),\tilde{B}(\sigma))$ belongs to the set of admissible values of $(a,b)$, defined earlier in this proof. 
Collecting the above, we conclude that
$$
\frac{z}{2\tilde{\sigma}(\sigma)} \left(\tilde{A}(\sigma) + \frac{\tilde{\sigma}(\sigma)}{z} \tilde{B}(\sigma)\right) e^{z(K^i_{j+1}-w)/\tilde{\sigma}(\sigma)} 
-\frac{z}{2\tilde{\sigma}(\sigma)}\left(\tilde{A}(\sigma) - \frac{\tilde{\sigma}(\sigma)}{z} \tilde{B}(\sigma)\right)
e^{-z(K^i_{j+1}-w)/\tilde{\sigma}(\sigma)}
$$
\begin{equation}\label{eq.PfLem20.3}
= \partial_K F\left(\tilde{A}(\sigma),\tilde{B}(\sigma);K^i_{j+1}\right),
\end{equation}
is a strictly increasing function of $\sigma\in(\hat{\sigma},0)$.
Notice that, as $\sigma\downarrow\hat{\sigma}$, we have: $\tilde{\sigma}(\sigma)\rightarrow\infty$, and the left hand side of (\ref{eq.PfLem20.3}) converges to $\tilde{B}(\hat{\sigma})<B_1$. When $\sigma\rightarrow\infty$, we obtain: $\tilde{A}(\sigma)\rightarrow A + B(w-K^i_j)$ and $\tilde{B}(\sigma)\rightarrow B$, and, in turn, $\partial_K F\left(\tilde{A}(\sigma),\tilde{B}(\sigma);K^i_{j+1}\right)\rightarrow  \partial_K F\left(A + B(w-K_j),B;K^i_{j+1}\right)$. Since $F(a,b;\cdot)$ is strictly convex, we have
$$
\partial_K F\left(A + B(w-K^i_j),B;K^i_{j+1}\right) > \frac{F\left(A + B(w-K^i_j),B;K^i_{j+1}\right) - F\left(A + B(w-K^i_j),B;w\right)}{K^i_{j+1}-w}
$$
$$
= \frac{\bar{V}(T_i,K^i_{j+1}) - A - B(w-K^i_j)}{K^i_{j+1}-w} = B_1,
$$
where the last equality is due to the definition of $w$.
This completes the proof of the lemma.

\emph{\bf Proof of Lemma $\ref{le:temp.3.1}$}.
To show the first inequality in the statement of the lemma, we notice that
$$
\tilde{A}(\bar{\sigma}) + \tilde{B}(\bar{\sigma}) (K^i_{j+1}-y) = V^{i-1}(K^i_{j+1}) 
= \bar{V}(T_{i-1},K^i_{j+1}) < \bar{V}(T_{i},K^i_{j+1})
$$
In addition, since $y<w$ and $B<B_1$, we have:
$$
\tilde{A}(\bar{\sigma}) + B_1(K^i_{j+1}-y) > A + B(y-K^i_j) + B_1(K^i_{j+1}-y) 
$$
$$
> A + B(w-K^i_j) + B_1(K^i_{j+1}-w) = \bar{V}(T_i,K^i_{j+1})
$$
The last inequality in the statement of the lemma follows from the strict convexity of $V^{i-1}$ and the following two observations:
$$
\tilde{A}(\bar{\sigma}) + \tilde{B}(\bar{\sigma}) (K^i_{j+1}-y) = V^{i-1}(K^i_{j+1}),
\,\,\,\,\,\,\,\,\,\,\,\,\tilde{A}(\bar{\sigma}) > A + B(y-K^i_j) = V^{i-1}(y)
$$

\emph{\bf Proof of Lemma $\ref{le:calib.le.1.2}$}.
Parts 1 and 2 of Assumption $\ref{assump:ass.1}$ are satisfied by construction. Let us show that the price curves $\bar{C}^i=C^{\nu^i,\sigma^i,z,x}$ satisfy the following: for all $i\geq1$, $\left(\bar{C}^i(K) - \bar{C}^{i-1}(K)\right)/\partial^2_{KK}\bar{C}^i(K)$ is bounded from above and away from zero, uniformly over $K\in(L_i,U_i)$.
Observe that
$$
\frac{\bar{C}^i(K) - \bar{C}^{i-1}(K)}{\partial^2_{KK}\bar{C}^i(K)}
= \frac{z^2}{a_i^2(K)} \frac{V^{\nu^i,\sigma^i,z,x}(K) - V^{\nu^{i-1},\sigma^{i-1},z,x}(K) }{V^{\nu^i,\sigma^i,z,x}(K)},
$$
where
$$
a_i(K) = \sum_{j=1}^{n} \sigma^i_{j} \bone_{\left[\nu^i_{j-1},\nu^i_j\right)}(K)
$$
Since each value function $V^{\nu^i,\sigma^i,z,x}$ is bounded from above and away from zero, on any compact in $(L_i,U_i)$, we only need to analyze the limiting behavior of the above ratio, as $K\downarrow L_i$ and $K\uparrow U_i$. 
Due to Assumption $\ref{assump:ass.2}$, $L_{i}\leq L_{i-1}$, for all $i$.
If $L_i < L_{i-1}$, then, as $K\downarrow L_i$,
$$
\frac{V^{\nu^i,\sigma^i,z,x}(K) - V^{\nu^{i-1},\sigma^{i-1},z,x}(K) }{V^{\nu^i,\sigma^i,z,x}(K)} \rightarrow 1
$$
If $L_i=L_{i-1}$, we have:
$$
\lim_{K\downarrow L_i} \frac{V^{\nu^i,\sigma^i,z,x}(K) - V^{\nu^{i-1},\sigma^{i-1},z,x}(K) }{V^{\nu^i,\sigma^i,z,x}(K)}
$$
\begin{equation}\label{eq.PfLem23.1}
= \lim_{K\downarrow L_i} \frac{ \lambda^i e^{z(K-L_i)/\sigma^i_1} - \lambda^i e^{-z(K-L_i)/\sigma^i_1} 
- \lambda^{i-1} e^{z(K-L_i)/\sigma^{i-1}_1} + \lambda^{i-1} e^{-z(K-L_i)/\sigma^{i-1}_1} }
{ \lambda^i e^{z(K-L_i)/\sigma^i_1} - \lambda^i e^{-z(K-L_i)/\sigma^i_1}},
\end{equation}
$$
=  \frac{ \lambda^i/\sigma^i_1 - \lambda^{i-1} /\sigma^{i-1}_1 }{ \lambda^i /\sigma^i_1 },
$$
with some positive constants $\lambda^i$ and $\lambda^{i-1}$. Notice that $V^i_+(L_i)=2\lambda^i/\sigma^i_1$, where $V^i_+$ is the right derivative of $V^i=V^{\nu^i,\sigma^i,z,x}$. Similarly, $V^{i-1}_+(L_i)=2\lambda^{i-1}/\sigma^{i-1}_1$. 
Recall (\ref{eq.calib.initderiv}), which guarantees that $V^i_+(L_i) > V^{i-1}_+(L_i)$. Hence, the right hand side of (\ref{eq.PfLem23.1}) is strictly positive. Analogous arguments hold in the case of $K\uparrow U_i$, which completes the proof of the lemma.

\bibliographystyle{abbrv}
\bibliography{MatchingEurOptions_refs}

\ed